\documentclass[11pt]{article}

\usepackage[margin=1in]{geometry}

\usepackage[usenames]{xcolor}
\usepackage{graphicx}
\usepackage[colorlinks=true,
linkcolor=webgreen,
filecolor=webbrown,
citecolor=webgreen]{hyperref}
\definecolor{webgreen}{rgb}{0,.5,0}
\definecolor{webbrown}{rgb}{.6,0,0}
\usepackage{color}

\usepackage{amsthm}
\usepackage{mathtools}

\usepackage{tikz}
\usetikzlibrary{shapes.geometric}
\usepackage{subcaption}
\newcommand{\mycaption}[2]{\caption[#1]{\, \textbf{#1.} #2}}
\usepackage{verbatim}
\usepackage{diagbox}
\usepackage{float}

\usepackage{enumitem}
\newlist{enumthm}{enumerate}{1}
\setlist[enumthm]{label=(\alph*)}

\theoremstyle{plain}
\newtheorem{theorem}{Theorem}
\newtheorem{corollary}[theorem]{Corollary}
\newtheorem{lemma}[theorem]{Lemma}
\newtheorem{proposition}[theorem]{Proposition}
\theoremstyle{definition}
\newtheorem{definition}[theorem]{Definition}
\newtheorem{example}[theorem]{Example}

\theoremstyle{definition}
\newtheorem{notation}[theorem]{Notation}
\theoremstyle{remark}
\newtheorem{remark}[theorem]{Remark}

\allowdisplaybreaks[1]
\newcommand{\floor}[1]{\left\lfloor#1\right\rfloor}
\newcommand{\ceil}[1]{\left\lceil#1\right\rceil}
\newcommand{\seqnum}[1]{\href{https://oeis.org/#1}{\underline{#1}}}

\begin{document}

\begin{center}
\vskip 1cm{\LARGE\bf 
A Class of Trees Having Near-Best Balance}
\vskip 1cm
\large
Laura Monroe
\\
Ultrascale Systems Research Center\\
Los Alamos National Laboratory\\
Los Alamos, NM 87501\\
USA \\
\href{mailto:lmonroe@lanl.gov}{\tt lmonroe@lanl.gov} 
\end{center}

\vskip .2 in

\begin{abstract}
	Full binary trees naturally represent commutative non-associative products. There are many important examples of these products: finite-precision floating-point addition and NAND gates, among others. 
Balance in such a tree is highly desirable for efficiency in calculation. 

The best balance is attained with a divide-and-conquer approach. However, this may not be the optimal solution, since the success of many calculations is dependent on the grouping and ordering of the calculation, for reasons ranging from the avoidance of rounding error, to calculating with varying precision, to the placement of calculation within a heterogeneous system. 

We introduce a new class of computational trees having near-best balance in terms of the Colless index  from mathematical phylogenetics. These trees are easily constructed from the binary decomposition of the number of terms in the problem. 
They also permit much more flexibility than the optimally balanced divide-and-conquer trees.  This gives needed freedom in the grouping and ordering of calculation, and allows intelligent efficiency trade-offs.

\bigskip



\end{abstract}

\renewcommand{\listfigurename}{Figures}
\renewcommand{\listtablename}{Tables}

\let\thefootnote\relax\footnote{This work was supported by Triad National Security, LLC, operator of the Los Alamos National Laboratory under Contract No.89233218CNA000001 with the U.S. Department of Energy, and by LANL's Ultrascale Systems Research Center at the New Mexico Consortium (Contract No. DE-FC02-06ER25750). 
The United States Government retains and the publisher, by accepting this work for publication, acknowledges that the United States Government retains a nonexclusive, paid-up, irrevocable, world-wide license to publish or reproduce this work, or allow others to do so for United States Government purposes. 
\newline
This paper is unclassified, and is assigned LANL identification number LA-UR-21-26451. 
}

\setcounter{page}{1}
\section{Introduction}
\subsection{Statement of the problem }
In this paper, we look at pairwise commutative non-associative products on $n$ terms and their binary tree representations, which express the grouping and ordering, or partitioning, of such calculations. While our motivation was in terms of balanced and correct computation, there is a rich literature in this area in mathematical phylogenetics, and we elucidate a connection between these areas.

Our motivating problem was that of finite-precision floating-point summation. It is well-known that this is prone to rounding error, due to its non-associativity under finite-precision constraints, and it has been shown that a poor choice of ordering and grouping of the terms can adversely affect the correctness of results. 
There are many other important computational examples of commutative non-associative products: 
 variable-precision calculations, calculations taking place on different processors on a heterogeneous machine, NAND gates, and others. 
 Like finite-precision floating-point summation, the implementations of any of these might have specific requirements in terms of grouping and ordering of terms; in other words, their required tree forms.
 
 Divide-and-conquer is the most balanced partitioning method, but is not always the best for the partitioning requirements of a given problem. For example, it has been shown in \cite{job20} and elsewhere that divide-and-conquer can lead to significant error in summations on real problems when a  small number of large values and a large number of small are added.
 On the other hand, divide-and-conquer is efficient. Computational efficiency is always a concern, and reasonable balance is important to attain computational efficiency on a problem. A calculation that is well-balanced will be optimal in some sense, and if run in parallel, will make efficient use of the processors. 
 
The problem, then, is to find a way to partition a commutative non-associative product that is suited to the problem at hand, in terms of correctness or other considerations, while at the same time retains much of the balance and efficiency of the divide-and-conquer partition.
 
\subsection{Approach to the problem }
To better understand commutative non-associative products, we use their natural correspondence with binary trees. We introduce the \emph{SD-tree} in Definition~\ref{def_SDtree}, a binary tree in which internal nodes are labeled $S$ if their children have the same number of leaf descendants, and $D$ if not, as shown in the example in Fig.~\ref{fig:sdtrees6}. This data structure provides a natural way of considering computational balance, because it specifically addresses the number of terms in each branch of the calculation. \input{figs/fig_sdtrees_6}

We use the SD-tree to classify several special cases of commutative non-associative products that occur often in computation. We first look at divide-and-conquer calculations, which are known to be efficient. We also consider the related full complete binary tree, and the ladder tree, which is essentially a serial calculation. 

It seemed reasonable that  a tree on $n$ leaves having a minimal number of $D$-nodes should have relatively good balance, since $D$-nodes themselves express a local lack of balance on a tree. We call this a \emph{MinD tree}. We have not seen an analysis of such trees in the literature, so we completely characterized these trees,
and look at their balance in terms of the Colless index, a widely used measure of tree balance from mathematical phylogenetics. 

It turns out that they are very easily constructed. 
If the binary decomposition of $n$ is $\sum_{i=1}^{\omega(n)}2^{\rho_i}$, then a MinD tree is formed by constructing any convenient base tree on $\omega(n)$ leaves, then attaching perfect trees having $2^{\rho_i}$ leaves in place of the leaves of the base tree. The interior nodes of the base tree are then the $D$-nodes of the constructed tree.
We show in Theorem~\ref{max_tree_s} that all MinD trees are of this form.
An example of the ease of construction of MinD trees is given here in Fig.~\ref{fig:minD_27}, where we construct all fifteen MinD trees with $27$ leaves. 
\input{figs/fig_compare27_min}

The best-case and worst-case MinD trees in terms of the Colless index are likewise easy to construct. These both have ladder base trees. 
The worst-case MinD tree arranges the perfect tree ``leaves'' in ascending order, as seen in Fig.~\ref{fig:minD_27}(\subref{fig:minD_27-a}), and shown in Theorem~\ref{th_bounds_MinD}, and the best-case arranges the perfect tree ``leaves'' in descending order, as seen in Fig.~\ref{fig:minD_27}(\subref{fig:minD_27-l}), shown in Theorem~\ref{th_bounds_MinD}.

We show in Theorem~\ref{theorem_threebounds} that all MinD trees do in fact have very good balance compared to the optimal divide-and-conquer tree, in terms of the Colless index, so these should have good performance compared to the best performance of a divide-and-conquer approach. An upper bound of the Colless index of the set of MinD trees with $n$ leaves, normalized against the best (divide-and-conquer tree) and worst (ladder tree) Colless indices of any tree with $n$ leaves, is  
$$\frac{2\floor{\log_2(n)}}{n}.$$
This upper bound is not sharp, so the MinD trees are even more balanced than this bound indicates. 

As an example, we calculate in Fig.~\ref{fig:colless_minD_27} the Colless indices of the trees from Fig.~\ref{fig:minD_27}, and compare these to the maximum and minimum Colless indices on a $27$-leaf tree. A visual comparison of the balance of MinD trees to the best- and worst-case trees for various $n$ is shown in Fig.~\ref{fig:colless_minD}.

Since all of the MinD trees have relatively good balance, the choice of base tree on $\omega(n)$ leaves is arbitrary, as is the ordering of the $\omega(n)$ perfect subtrees. Both the base tree and the ordering of the perfect subtrees may be freely chosen to address problem requirements.

Thus, trees having a minimal number of $D$-nodes are easily constructed and allow flexibility that is lacking on the most efficient divide-and-conquer trees. At the same time, they have balance approaching that of a divide-and-conquer tree, so are efficient. They therefore meet the goal of grouping a calculation in a manner suited to problem needs, while still retaining near-best efficiency. 

\subsection{Application of MinD trees to commutative non-associative calculations}
The original motivating example for this work was that of floating-point summations, which are commutative and non-associative on finite-precision floating point machines. It is well known that these are prone to rounding error \cite{goldberg91,higham93}. Methods exist to mitigate such error, in particular Kahan's method \cite{kahan_1971_ASO,kahan_1973_implementation} and methods derived from this. However, these add extra complexity, and are not always used, leading to possible error in calculation. It is therefore desirable to develop methods to minimize this error that are both easy to implement and efficient.

We use the problem of floating-point summations on an example  taken from \cite{job20}, demonstrating the application of our method, with known correctness heuristics as the problem constraint on grouping and ordering. We discuss here how the MinD algorithm addresses this constraint.

In \cite{job20}, the effects of such error were demonstrated on two summation calculations used in two real applications CLAMR~\cite{Nicholaeff2012} and the Oregonator~\cite{becker1985stationary}. The implementations of these codes used did not originally use any mitigation strategy, and did in fact show appreciable error when compared to a mitigated calculation, as seen in Fig.~\ref{fig:clamr}. This suggests that one can hope to mitigate floating-point error substantially by proper grouping and ordering. One might even approach the correctness of a Kahan-mitigated calculation, in some cases.
\begin{figure}[H]
    \centering
    \includegraphics[width=0.5\textwidth]{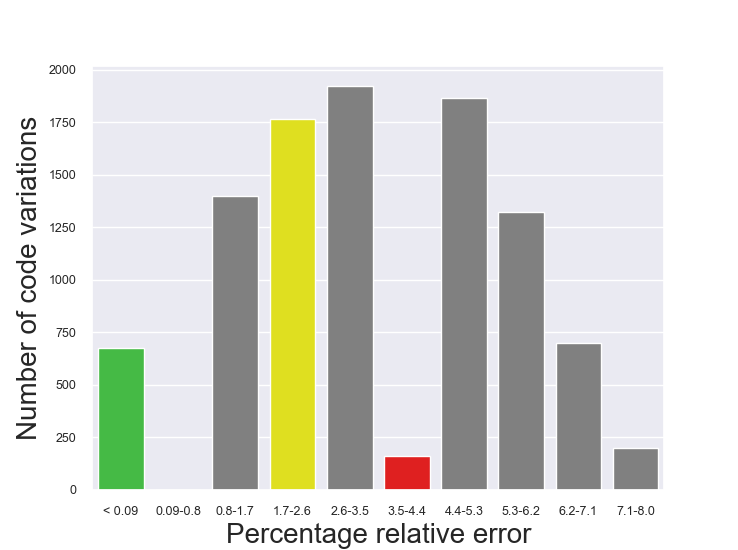}
    \mycaption{Error rate for different groupings and orderings on the CLAMR calculation}{The original ordering and grouping of the summation in the CLAMR code gave an error rate of $2.109$\%, and the divide-and-conquer ordering and grouping gave error rates of $2.109-3.517$\%. The mitigated Kahan calculation gave error rates of $0.083$\%, and the best ordering and grouping, with no mitigation, gave $0.082$\% error rates, a $25-43$X improvement in error rates, which is as good as the mitigated Kahan. The chart and error rates in this figure are taken from \cite{job20}.}
  \label{fig:clamr}
\end{figure}

The authors suggested some heuristics for ordering and grouping, including the most important heuristic: add from smallest to largest, where the largest are added last \cite{job20}. Unfortunately, standard methods such as divide-and-conquer do not permit complete application of this heuristic, as may be observed in Fig. \ref{fig:om_compare}(\subref{fig:om_compare-c}), so may lead to poor results, as shown in Fig.~\ref{fig:clamr}. 
The authors sampled from all orderings and experimentally found two most correct orderings for the problem, and established heuristics but not an algorithm for generating the most correct trees. They also did not address the problem of efficiency on relatively correct trees.
\begin{figure}[h!tb]
	\tikzstyle{mynode}=[circle, draw]
	\tikzstyle{Dnode}=[regular polygon sides=4, minimum size=0.6cm, draw]
	\tikzstyle{small_leaf}=[fill=gray!25]
	\tikzstyle{big_leaf}=[fill=gray!180]
	\tikzstyle{ghost_leaf}=[fill=gray!0,draw=none]
	\tikzstyle{myarrow}=[]
	\centering
	\begin{subfigure}[b]{.3\textwidth}
	    \centering
	    \resizebox{0.9\textwidth}{!}{%
		\begin{tikzpicture}	
			\centering
			\node[Dnode] (L1) at (0, 3) {\footnotesize\texttt{D}};
			
			\node[mynode] (L2a) at (-2, 2) {\footnotesize\texttt{S}};
			\node[mynode, big_leaf] (L2b) at (2, 2) {\footnotesize\texttt{}};
			
			\node[mynode] (L3a) at (-3, 1) {\footnotesize\texttt{S}};
			\node[mynode] (L3b) at (-1, 1) {\footnotesize\texttt{S}};
			\node[mynode, ghost_leaf] (L3c) at (1, 1) {\footnotesize\texttt{}};
			\node[mynode, ghost_leaf] (L3d) at (3, 1) {\footnotesize\texttt{}};

			\node[mynode] (L4a) at (-3.5, 0) {\footnotesize\texttt{S}};
			\node[mynode] (L4b) at (-2.5, 0) {\footnotesize\texttt{S}};
			\node[mynode] (L4c) at (-1.5, 0) {\footnotesize\texttt{S}};
			\node[mynode] (L4d) at (-0.5, 0) {\footnotesize\texttt{S}};
			\node[mynode, ghost_leaf] (L4e) at (0.5, 1) {\footnotesize\texttt{}};
			\node[mynode, ghost_leaf] (L4f) at (1.5, 1) {\footnotesize\texttt{}};
			\node[mynode, ghost_leaf] (L4g) at (2.5, 1) {\footnotesize\texttt{}};
			\node[mynode, ghost_leaf] (L4h) at (3.5, 1) {\footnotesize\texttt{}};

			\node[mynode, small_leaf] (L5a) at (-3.75, -1) {\footnotesize\texttt{}};
			\node[mynode, small_leaf] (L5b) at (-3.25, -1) {\footnotesize\texttt{}};
			\node[mynode, small_leaf] (L5c) at (-2.75, -1) {\footnotesize\texttt{}};
			\node[mynode, small_leaf] (L5d) at (-2.25, -1) {\footnotesize\texttt{}};
			\node[mynode, small_leaf] (L5e) at (-1.75, -1) {\footnotesize\texttt{}};
			\node[mynode, small_leaf] (L5f) at (-1.25, -1) {\footnotesize\texttt{}};
			\node[mynode, small_leaf] (L5g) at (-0.75, -1) {\footnotesize\texttt{}};
			\node[mynode, small_leaf] (L5h) at (-0.25, -1) {\footnotesize\texttt{}};
			
			\draw[myarrow] (L1) to node[midway,above=0pt, left=7pt]{} (L2a);
			\draw[myarrow] (L1) to node[midway,above=0pt, left=-2.5pt]{} (L2b);
			
			\draw[myarrow] (L2a) to node[midway,above=0pt, left=4pt]{} (L3a);
			\draw[myarrow] (L2a) to node[midway,above=0pt, right=-1pt]{} (L3b);

			\draw[myarrow] (L3a) to node[midway,above=0pt, left=0mm]{} (L4a);
			\draw[myarrow] (L3a) to node[midway,above=0pt, right=-2pt]{} (L4b);
			\draw[myarrow] (L3b) to node[midway,above=0pt, left=0mm]{} (L4c);
			\draw[myarrow] (L3b) to node[midway,above=0pt, right=-2pt]{} (L4d);

			\draw[myarrow] (L4a) to node[midway,above=0pt, left=0mm]{} (L5a);
			\draw[myarrow] (L4a) to node[midway,above=0pt, right=-2pt]{} (L5b);
			\draw[myarrow] (L4b) to node[midway,above=0pt, left=0mm]{} (L5c);
			\draw[myarrow] (L4b) to node[midway,above=0pt, right=-2pt]{} (L5d);
			\draw[myarrow] (L4c) to node[midway,above=0pt, left=0mm]{} (L5e);
			\draw[myarrow] (L4c) to node[midway,above=0pt, right=-2pt]{} (L5f);
			\draw[myarrow] (L4d) to node[midway,above=0pt, left=0mm]{} (L5g);
			\draw[myarrow] (L4d) to node[midway,above=0pt, right=-2pt]{} (L5h);
		\end{tikzpicture}
		} 
		\caption{A MinD tree produced by the method in this paper. It is isomorphic to the two most correct groupings and orderings in \cite{job20}, which gave rounding error in only $0.082$\% of the cells.}
		\label{fig:om_compare-a}
	\end{subfigure}\hfill%
	\begin{subfigure}[b]{.3\textwidth}
	    \centering
	    \resizebox{0.9\textwidth}{!}{%
		\begin{tikzpicture}	
			\centering
			\node[Dnode] (L1) at (0, 3) {\footnotesize\texttt{D}};
			
			\node[Dnode] (L2a) at (-2, 2) {\footnotesize\texttt{D}};
			\node[mynode] (L2b) at (2, 2) {\footnotesize\texttt{S}};
			
			\node[mynode,big_leaf] (L3a) at (-3, 1) {\footnotesize\texttt{}};
			\node[mynode] (L3b) at (-1, 1) {\footnotesize\texttt{S}};
			\node[mynode] (L3c) at (1, 1) {\footnotesize\texttt{S}};
			\node[mynode] (L3d) at (3, 1) {\footnotesize\texttt{S}};
			
			\node[mynode, ghost_leaf] (L4a) at (-3.5, 0) {\footnotesize\texttt{}};
			\node[mynode, ghost_leaf] (L4b) at (-2.5, 0) {\footnotesize\texttt{}};
			\node[mynode] (L4c) at (-1.5, 0) {\footnotesize\texttt{S}};
			\node[mynode] (L4d) at (-0.5, 0) {\footnotesize\texttt{S}};
			\node[mynode, small_leaf] (L4e) at (0.5, 0) {\footnotesize\texttt{}};
			\node[mynode, small_leaf] (L4f) at (1.5, 0) {\footnotesize\texttt{}};
			\node[mynode, small_leaf] (L4g) at (2.5, 0) {\footnotesize\texttt{}};
			\node[mynode, small_leaf] (L4h) at (3.5, 0) {\footnotesize\texttt{}};
			
			\node[mynode, ghost_leaf] (L5a) at (-3.75, -1) {\footnotesize\texttt{}};
			\node[mynode, ghost_leaf] (L5b) at (-3.25, -1) {\footnotesize\texttt{}};
			\node[mynode, small_leaf] (L5e) at (-1.75, -1) {\footnotesize\texttt{}};
			\node[mynode, small_leaf] (L5f) at (-1.25, -1) {\footnotesize\texttt{}};
			\node[mynode, small_leaf] (L5g) at (-0.75, -1) {\footnotesize\texttt{}};
			\node[mynode, small_leaf] (L5h) at (-0.25, -1) {\footnotesize\texttt{}};
			
			\draw[myarrow] (L1) to node[midway,above=0pt, left=7pt]{} (L2a);
			\draw[myarrow] (L1) to node[midway,above=0pt, left=-2.5pt]{} (L2b);
			
			\draw[myarrow] (L2a) to node[midway,above=0pt, left=4pt]{} (L3a);
			\draw[myarrow] (L2a) to node[midway,above=0pt, right=-1pt]{} (L3b);
			\draw[myarrow] (L2b) to node[midway,above=0pt, left=4pt]{} (L3c);
			\draw[myarrow] (L2b) to node[midway,above=0pt, right=-1pt]{} (L3d);
			
			\draw[myarrow] (L3b) to node[midway,above=0pt, left=0mm]{} (L4c);
			\draw[myarrow] (L3b) to node[midway,above=0pt, right=-2pt]{} (L4d);
			\draw[myarrow] (L3c) to node[midway,above=0pt, left=0mm]{} (L4e);
			\draw[myarrow] (L3c) to node[midway,above=0pt, right=-2pt]{} (L4f);
			\draw[myarrow] (L3d) to node[midway,above=0pt, left=0mm]{} (L4g);
			\draw[myarrow] (L3d) to node[midway,above=0pt, right=-2pt]{} (L4h);
			
			\draw[myarrow] (L4c) to node[midway,above=0pt, left=0mm]{} (L5e);
			\draw[myarrow] (L4c) to node[midway,above=0pt, right=-2pt]{} (L5f);
			\draw[myarrow] (L4d) to node[midway,above=0pt, left=0mm]{} (L5g);
			\draw[myarrow] (L4d) to node[midway,above=0pt, right=-2pt]{} (L5h);
		\end{tikzpicture}
		} 
		\caption{The non-MinD method used in the original CLAMR code, which gave rounding error in $2.109$\% of the cells \cite{job20}.\newline\newline}
		\label{fig:om_compare-b}
	\end{subfigure}\hfill%
	\begin{subfigure}[b]{.3\textwidth}
	    \centering
	    \resizebox{0.9\textwidth}{!}{%
		\begin{tikzpicture}	
			\centering
			\node[Dnode] (L1) at (0, 3) {\footnotesize\texttt{D}};
			
			\node[mynode] (L2a) at (-2, 2) {\footnotesize\texttt{S}};
			\node[Dnode] (L2b) at (2, 2) {\footnotesize\texttt{D}};
			
			\node[mynode] (L3a) at (-3, 1) {\footnotesize\texttt{S}};
			\node[mynode] (L3b) at (-1, 1) {\footnotesize\texttt{S}};
			\node[mynode] (L3c) at (1, 1) {\footnotesize\texttt{S}};
			\node[Dnode] (L3d) at (3, 1) {\footnotesize\texttt{D}};
			
			\node[mynode, small_leaf] (L4a) at (-3.5, 0) {\footnotesize\texttt{}};
			\node[mynode, small_leaf] (L4b) at (-2.5, 0) {\footnotesize\texttt{}};
			\node[mynode, small_leaf] (L4c) at (-1.5, 0) {\footnotesize\texttt{}};
			\node[mynode, small_leaf] (L4d) at (-0.5, 0) {\footnotesize\texttt{}};
			\node[mynode, small_leaf] (L4e) at (0.5, 0) {\footnotesize\texttt{}};
			\node[mynode, small_leaf] (L4f) at (1.5, 0) {\footnotesize\texttt{}};
			\node[mynode, big_leaf] (L4g) at (2.5, 0) {\footnotesize\texttt{}};
			\node[mynode] (L4h) at (3.5, 0) {\footnotesize\texttt{S}};
			
			\node[mynode, small_leaf] (L5o) at (3.25, -1) {\footnotesize\texttt{}};
			\node[mynode, small_leaf] (L5p) at (3.75, -1) {\footnotesize\texttt{}};

			\draw[myarrow] (L1) to node[midway,above=0pt, left=7pt]{} (L2a);
			\draw[myarrow] (L1) to node[midway,above=0pt, left=-2.5pt]{} (L2b);
			
			\draw[myarrow] (L2a) to node[midway,above=0pt, left=4pt]{} (L3a);
			\draw[myarrow] (L2a) to node[midway,above=0pt, right=-1pt]{} (L3b);
			\draw[myarrow] (L2b) to node[midway,above=0pt, left=4pt]{} (L3c);
			\draw[myarrow] (L2b) to node[midway,above=0pt, right=-1pt]{} (L3d);
			
			\draw[myarrow] (L3a) to node[midway,above=0pt, left=0mm]{} (L4a);
			\draw[myarrow] (L3a) to node[midway,above=0pt, right=-2pt]{} (L4b);
			\draw[myarrow] (L3b) to node[midway,above=0pt, left=0mm]{} (L4c);
			\draw[myarrow] (L3b) to node[midway,above=0pt, right=-2pt]{} (L4d);
			\draw[myarrow] (L3c) to node[midway,above=0pt, left=0mm]{} (L4e);
			\draw[myarrow] (L3c) to node[midway,above=0pt, right=-2pt]{} (L4f);
			\draw[myarrow] (L3d) to node[midway,above=0pt, left=0mm]{} (L4g);
			\draw[myarrow] (L3d) to node[midway,above=0pt, right=-2pt]{} (L4h);

			\draw[myarrow] (L4h) to node[midway,above=0pt, left=0mm]{} (L5o);
			\draw[myarrow] (L4h) to node[midway,above=0pt, right=-2pt]{} (L5p);
		\end{tikzpicture}
		} 
		\caption{The non-MinD pairwise tree shown in \cite{job20}, which gave rounding error in $2.109-3.517$\% of the cells, depending on the ordering of the small values.\newline}
		\label{fig:om_compare-c}
	\end{subfigure}\hfill%
	
	\mycaption{A comparison of a MinD tree to other trees shown in \cite{job20}} {This figure compares the MinD ordering from this paper in (\subref{fig:om_compare-a}), to the original ordering used in CLAMR from \cite{job20} in (\subref{fig:om_compare-b}) and the pairwise ordering from \cite{job20} in (\subref{fig:om_compare-c}). The MinD ordering is the same as the best ordering found in \cite{job20}. The single large value is represented by the black leaf node, and the eight small values are represented by grey leaf nodes.}
	\label{fig:om_compare}	
\end{figure}

The summation discussed in \cite{job20} consisted of $8$ different small values and $1$ large value. A MinD tree for a summation on $9$ terms consists of a root (the single $D$-node) and two subtrees, one a perfect tree with $8$ leaves, and the other a perfect tree with $1$ leaf. According to the heuristic above, all $8$ small values should be added first, and then the single large value. The MinD tree permits this partition. The tree shown in Fig. \ref{fig:om_compare}(\subref{fig:om_compare-a}) is the result, equivalent to the best ordering found by experiment in \cite{job20}.

It is worth noting that not every partition may be attained through this method; for example, if there had been $2$ large values rather than $1$ in this example, one of the large values would have been included in the divide-and-conquer perfect tree with $8$ elements. This method thus provides more flexibility, and likely improved correctness, but not perfect flexibility. Future work will analyze the application of the MinD method to subtrees, with the comparative efficiency that results.

This leaves aside the work that must occur to identify the larger values. In \cite{job20}, domain knowledge provided the relative values; one of the variables was known to be very much larger than the other eight very small values. On other problems, domain knowledge might suffice, or perhaps a filter would help. We do not address this issue here, as the intent is to provide a flexible and efficient method of partitioning, not to address individual use cases. 

\subsection{Contributions}
The main contributions of this paper are: 
\begin{itemize}
    \item Introduction in Section~\ref{max_S_section} of the \emph{MinD trees}, a class of trees for grouping non-associative products. These trees are easily constructed (Theorem~\ref{max_tree_s}) (which has also recently been shown independently in  \cite{kersting21}),  efficient (Theorem~\ref{theorem_threebounds}), and at the same time  allow flexibility in partitioning.
    \item The introduction of the \emph{SD-tree} structure in Section~\ref{sdtrees} and its exploitation throughout the paper; in particular, its use in characterizing tree  forms and in calculating the Colless index. 
    \item The linkage of computational commutative non-associative products to the extensive literature on the subject of ancestry trees from the mathematical phylogenetics community, and in particular, the discussion of measures of tree balance from that community in terms of the requirements of computational science.
\end{itemize}
Other contributions include:
\renewcommand\labelitemii{$\bullet$}
\begin{itemize}
    \item Counting the number of full binary trees in terms of the number of their $S$-nodes or $D$-nodes, in Theorem \ref{propkSnodes} (with illustrative tables of examples in Tables \ref{paren_nk} and \ref{table_dnodes} in Appendix B). This enables the detailed analysis of full binary trees throughout the paper.
    \item A new formula for counting $S$-nodes of divide-and-conquer (or maximally balanced) trees. A new recursive formula for the number of $S$-nodes  in a divide-and-conquer tree is given in Theorem \ref{exp_recursive}, and a new closed form in Theorem \ref{exp_closed}.
    \item Several new formulas for counting $D$-nodes of divide-and-conquer (or maximally balanced) trees, including both recursive and closed forms. Since the sequence of the number of $D$-nodes in divide-and-conquer trees with $n$ leaves is a series of dilations of the Takagi function on the dyadic rationals, the new formulas for $D$-nodes immediately apply to and are new for the Takagi function. 
    \begin{itemize}
        \item New recursive formulas for the number of $D$-nodes in a divide-and-conquer tree are given in  Corollary \ref{yet_another_dnodes_prop} and Theorem \ref{dnodes_another_recurrence}, and new closed forms in Theorems \ref{exp_D_closed} and \ref{another_explicit_D}. 
        \item A new expression for the Takagi function is given in Theorem \ref{another_takagiX}.
    \end{itemize}
    \item Several contributions to the Online Encyclopedia of Integer Sequences \cite{OEIS}, including a new sequence and new interpretations and formulas for existing sequences. These make up a family of sequences that were not formerly related in the OEIS. These sequences are listed in the Table \ref{summary_table} of Appendix A, Section \ref{appendixA}, along with their interpretations in terms of commutative, non-associative operations.
    \item An analysis of $S$- and $D$-nodes on complete full binary trees, in Section \ref{cfb}, and brief discussion of their Colless index in Section \ref{colless_cfb}.
\end{itemize}

We became aware of the rich body of literature from mathematical phylogenetics in this area late in the process of this work. Several of the theorems we discuss have been published earlier in this field in the sections discussing divide-and-conquer trees. We have identified in the text those theorems which have been previously published; however, we include our alternative proofs here, as our methods are somewhat different, and perhaps may give another approach to problems in mathematical phylogenetics. We also include the proof of the construction of MinD trees, independently shown in the recent \cite{kersting21}.

\subsection{Organization of this paper}

\hspace{\parindent}Section \ref{examples} presents some motivating examples of commutative non-associative products.

Section \ref{background} gives background on commutative non-associative products.

Section \ref{sdtrees} introduces and motivates the concept of \emph{SD-trees}, which we use for counting trees of a given form, and also in assessing balance of computation. The remainder of the paper expands upon and exploits this construct.

Section~\ref{prelims} gives some preliminary definitions, and some simple lemmas on weights of binary vectors that are used throughout the paper.

Section \ref{s-nodes-counting} counts commutative non-associative forms and products in terms of their $S$-nodes and gives lower and upper bounds on the number of these.

Section \ref{special} discusses several common tree forms that occur in a computational context in terms of SD-trees. We discuss ladder (sequential) products, parenthetic forms having $2$ $S$-nodes, parenthetic forms having a minimal number of $D$-nodes, divide-and-conquer forms (and their relationship with the Takagi function), and forms based on complete full binary trees. 

Section~\ref{max_S_section} introduces and shows the easy construction of the \emph{MinD tree}, having  minimal $D$-nodes. These trees give greater flexibility than divide-and-conquer, and hence correctness in calculation.

Section \ref{colless_dnodes} discusses the Colless index and calculates it for various trees. 
In particular, upper and lower bounds for the Colless index on MinD trees are derived, and we construct MinD trees which attain each of those bounds. We find that all MinD trees attain near-best Colless index among all trees of their size. This is the main result of the paper.

Appendix A (Section~\ref{appendixA}) is a table of a family of OEIS sequences shown in this paper to be related to SD-trees.

Appendix B  (Section~\ref{appendixB}) gives two tables showing the number of trees with $n$ leaves having a given number of $S$- and $D$-nodes.
\section{Motivating examples of commutative non-associative products} \label{examples}

\subsection{Computer and computational science}
\subsubsection{Finite-precision floating-point summations} \label{fp_sum}

Floating-point summation on a finite-precision system is commutative, but is not generally associative. 
The non-associativity is because of the rounding that may occur in finite-precision floating-point addition. For example, on a binary machine with two bits of precision, $$((10+0.1)+-0.1)=(10.1+-0.1)\approx (11+-0.1)=10.1\approx 11,$$ 
since both additions (on binary floats) result in three-bit numbers and are rounded up in the $\approx$ approximation, whereas on the same binary machine (and also mathematically), $$(10+(0.1+-0.1)) = (10+0) = 10.$$ 
Two floating-point summations are \emph{computationally  equivalent} under the IEEE 754 standard for computer arithmetic  \cite{IEEE754} if they differ only by some series of pairwise transpositions. Associativity is not guaranteed and should not be  assumed on a finite-precision machine.
In contrast to computational floating-point addition, the mathematical definition of addition is both commutative and associative. Thus, not all mathematically equivalent summations are computationally equivalent \cite{goldberg91,higham93}. 

Two computationally equivalent summations will always produce the same results on an IEEE-754-compliant machine. On the other hand, two mathematically equivalent summations that are not computationally equivalent may produce different results, as  above.

The assumption that all mathematically equivalent summations are also computationally equivalent is easy to make. However, this assumption can lead to rounding error \cite{goldberg91,higham93} and affect the accuracy of the summation result to the detriment of the overall calculation. The  need to find groupings and orderings that are accurate as well as efficient was the motivation for this study.  

\subsubsection{General computational examples}
In general, anything that has a natural representation as a rooted full binary tree (where non-sibling leaves may not be interchanged at will) is an example of a commutative non-associative product. 
NAND  gates are an example of a commutative, non-associative operation. 
A composable mean operation $\mu(x_1,x_2) = (x_1+x_2)/2$ is another example. 

Some operations are computationally non-associative but mathematically associative, like the summation example, and some operations are non-associative both computationally and mathematically. In the first case, the grouping and ordering does not matter mathematically, so one is free to choose the grouping and ordering that gives the most correct result with the most efficient execution. In the second case, there are not only computational but also mathematical constraints on the grouping choice.

\subsubsection{Some implications of non-associativity for computation}
In some sense,  selecting a grouping for a commutative non-associative product answers the question: ``In what manner should the product calculation take place?'', or ``How shall the calculation be partitioned?'', whereas selecting an ordering  answers: ``How shall we instantiate the chosen partition?"

For example, in a parallel calculation, a balanced tree is preferred for good load-balancing. However, there may not be control over the order in which parallel sub-products are calculated, and since the product is non-associative, this may affect the accuracy of the result.  

In the summation example, an obvious implication of non-associativity is rounding error and incorrectness. However, there may be a trade-off, in that the best forms of parenthesization and ordering for correctness may impair efficiency: in a perfectly load-balanced scheme, one may find oneself obliged to group very large summands with very small, thus greatly enhancing the likelihood of computational error.
In this paper, we provide mechanisms for addressing the trade-off between best grouping and best efficiency. 

\subsection{Mathematical phylogenetics}
Although our motivation was computational, these trees have also been studied in much detail in the field of mathematical phylogenetics.

Phylogenetics is the study of the appearance and development of evolutionary traits among species or individuals. These are expressed in ancestry trees, where a tree branches when a new trait appears in the evolutionary record. These trees are full binary trees, and indeed, they do express a certain kind of commutative non-associative binary product, where the product of a pair of organisms is their child. 

The shape of an ancestry tree gives important information about the divergence rates of species. In particular, the balance of a tree indicates the stability of a given genetic line, and likewise, its imbalance indicates the differences in the number of divergences that have occurred in branches of the tree; in other words, how likely is it that two genetic lines have diverged. Balance can be measured by looking at the number of leaf descendants of the two children of a node in the tree. 

This is essentially the same as the computational problem of finding balance in a calculation. However, just as in phylogenetics, perfect balance in calculation is not always attainable. Divide-and-conquer algorithms are known to be efficient in this sense, and are one of the classical methods of addressing problems that afford this kind of approach.

There are many references for mathematical work in this area, which date back many decades \cite{colless80, heard92, kirkpatrick93, colless95, mooers97, rosenberg19, coronado20, kersting21}. We refer to this work throughout this paper, and observe an interesting overlap between two distant scientific areas.

\subsection{Other examples}
There are many other examples of commutative non-associative products from many different fields. The children's game Rock-Paper-Scissors is an example, as is any tournament. Hanging mobiles could be considered a real-world example, with the balancing points as the products.

\section{Background} \label{background}
\subsection{Commutative non-associative products and binary trees}
\begin{definition}
	A \emph{commutative non-associative product} on $n$ terms is an fully parenthesized product with $n$ operands that employs some commutative non-associative operation. 
\end{definition}
\begin{definition}
Two commutative non-associative products are \emph{equivalent} if they differ only by some series of pairwise commutations within  a set of parentheses. The equivalence relation in use throughout this paper is thus pairwise commutativity, and all isomorphisms we discuss in this paper are in terms of this relation.
\end{definition}
    A commutative non-associative product is represented by a rooted full binary tree whose leaves are labeled. In this representation, leaves may be labeled by the terms in the product, and an internal (operator) node with two children represents the product of its children. 
    In this representation, two commutative non-associative products are \emph{equivalent} if one can be obtained from the other by a sequence of reversing the children of some set of nodes in their tree representations.
\begin{definition}
	The \emph{parenthetic form} of a fully parenthesized commutative non-associative product is the form the parentheses and the operators take, leaving the operands themselves undefined.
\end{definition}
\begin{definition}
    Two parenthetic forms are \emph{isomorphic} if one can be obtained from the other by some sequence of commuting abstract operands within parentheses. 
\end{definition}
    A parenthetic form is represented by a rooted full binary tree whose leaves are not labeled. 
    Two parenthetic forms are \emph{isomorphic} if one can be obtained from the other by a sequence of reversing the children of some set of nodes in their tree representations. This is the usual definition of tree isomorphism.
    Two equivalent commutative non-associative products must have isomorphic parenthetic forms. However, two commutative non-associative products with isomorphic parenthetic forms need not themselves be equivalent. 
    
    \begin{definition}
	\emph{Full binary trees} are trees in which every node has either zero or two children \cite{knuth97art1}. 
\end{definition}
    Full binary trees are equivalent to commutative non-associative products, where the internal nodes represent the product of their children. They can be used to represent any such product. 
    We treat commutative, non-associative products and leaf-labeled tree representations interchangeably, and also treat parenthetic forms and and their unlabeled tree representations interchangeably.

\section{SD-Trees} \label{sdtrees}
\subsection{SD-trees, S-nodes, and D-nodes}
\begin{definition}\label{def_SDtree}
	An \emph{SD-tree} is a rooted full binary tree in which an internal node is labeled $S$ if its two children have the \textbf{S}ame number of descendant leaf nodes, and $D$ if its two children have \textbf{D}ifferent numbers of descendant leaf nodes. 
\end{definition}

$S$-nodes are similar to \emph{symmetry vertices} \cite{coronado20, kersting21}, vertices whose two children subtrees have the same shape. However, the subtrees of $S$-nodes are not required to have the same shape, only the same number of leaf descendants. In the recent \cite{kersting21}, Kersting and Fischer discuss full binary trees in terms of symmetry vertices, which is closely related to this work on $S$-nodes. They differ in that the subtrees of $S$-nodes are not isomorphic, so these are different classes of trees. In \cite{kersting21}, the authors discuss certain connections between the two types of trees.

$D$-nodes are the same as the $J$-nodes with non-zero balance value discussed in \cite{colless80, rogers96} and in other references from mathematical phylogenetics. 

An SD-tree is a variant of a binary PQ-tree \cite{booth76}. The leaves of an SD-tree may be labeled or unlabeled, depending on whether an instantiated product or a parenthetic form is represented. 
The definition of isomorphism remains the same.  
\begin{lemma}\label{iso_SD}
	Two equivalent/isomorphic SD-trees with $n$ leaf nodes have the  same number of $S$-nodes. Both trees have the same number of $D$-nodes.
\end{lemma}
\begin{proof}
	Follows from the definition of equivalence, and pairwise commutativity.
\end{proof}
	The converse of Lemma~\ref{iso_SD} is not true. Two SD-trees with $n$ leaf nodes and with the same number  of $S$-nodes and $D$-nodes need not be isomorphic. For example, there are two non-isomorphic trees with $6$ leaves, $2$ $S$-nodes and $3$ $D$-nodes, and two non-isomorphic trees with $6$ leaves, $3$ $S$-nodes and $2$ $D$-nodes, as can be seen in Fig. \ref{fig:sdtrees6}. These are the smallest such examples.
\begin{lemma} \label{dnodes}
    Let $s(n)$ be the number of $S$-nodes in a binary tree having $n$ leaves. Then the number of $D$-nodes is $((n-1)-s(n))$. 
\end{lemma} 
\begin{proof}
    The number of interior nodes in a binary tree with $n$ leaves is $(n-1)$. $D$ nodes are all those interior nodes that are not $S$-nodes.
\end{proof}
\subsection{Motivation for the SD-tree data structure}
There are several reasons to consider this type of structure. In general, $S$-nodes are useful in considering automorphisms of these trees, whereas $D$-nodes are useful in assessing tree balance.
\subsubsection{Combinatorics}
When we count trees of a particular parenthetic form, we may start by counting permutations of the leaves. However, when a node has children with the same number of descendant leaf nodes, equivalent trees are counted twice, since the same set of leaf nodes may be assigned to either of the two children and produce equivalent trees. 
It therefore makes sense to differentiate between nodes with children having the same number of descendant leaf nodes and nodes with children having a different number of descendant leaf nodes. This motivates the SD-tree data structure. 

We formalize this in Section \ref{method}, and exploit it throughout the rest of the paper, using the SD-tree structure to count and classify commutative non-associative products. 

\subsubsection{Computational efficiency}
When processing tree-based data structures, efficiency is enhanced by the use of a balanced tree; for example, when executing a parallel algorithm or a logarithmic algorithm. The SD-tree structure explicitly identifies those sub-branches of a tree that are balanced in terms of the number of descendant leaves, and therefore can be useful in assessing overall efficiency. Note that a balanced $S$-node may have descendant nodes that are imbalanced, as in Fig. \ref{fig:treetypes}(\subref{fig:treetypes-b}).

Although a balanced tree may be most efficient, under certain circumstances use of such a structure can lead to increased error in computation. For example, rounding error may increase in those floating-point summations on a finite-precision machine under a pairwise divide-and-conquer (balanced) calculation, in which one element is very much larger than the others \cite{job20}. Use of the SD-tree data structure helps in the analysis of such calculations.

\subsubsection{Colless index of a tree}
The Colless index is widely studied in phylogenetics as a measure of tree balance but is less well known in general computational theory. $D$-nodes may be used to calculate the Colless index, and we do so in Section \ref{colless_dnodes}.

\subsubsection{Automorphism group of a SD-tree}
Tree automorphisms act upon a tree by transposing the left and right subtrees of some set of nodes, where the result is the same tree. We may choose to define SD-trees so that the left child of a node always has at least as many descendant leaf nodes as the right child. It may happen that the left child and the right child have the same number of descendants. If $s$ nodes have that characteristic, there will then be $2^s$ representations of a tree having $s$ $S$-nodes. In other words, if $s$ is the number of $S$ nodes in an SD-tree on $n$ leaves defined as above, then the order of the automorphism group of that tree is $2^s$.

\subsubsection{Connections to other mathematical areas}
The Takagi function is a nowhere-differentiable function from differential calculus \cite{takagi01, lagarias11, allaart11}. The sequence of $D$-nodes of divide-and-conquer trees on $n$ leaves is a series of dilations of the Takagi curve restricted to the dyadic rationals. 
\section{Preliminary definitions, notations  and lemmas}\label{prelims} \label{back_def_not_lem}
\subsection{A few useful definitions and notations}
\begin{definition}
    The $(k+1)$-bit \emph{binary representation} of $n=\sum\limits_{i=0}^{k}n_i \cdot 2^i$ is the sequence $n_{k} \dots n_1 n_0$, where $n_i \in \{0, 1\}$.
\end{definition}
\begin{notation}
    An integer $n$ may be expressed as 
    \begin{itemize}
        \item $n=(n_k \dots n_1 n_0)$, where $n_k=1$, the \emph{reduced binary representation} of $n$, used when discussing individual entries in the binary representation of $n$.
        \item $n=2^k+r$, where $0 \le r < 2^k$, used when discussing $k=\floor{\log_2(n)}$ (one less than the number of bits in the reduced binary representation of $n$) or $r=n\bmod 2^k$.
        \item $n=2^\ell \cdot d$, where $d$ is odd, used when discussing the highest power of $2$ dividing $n$.
        \item $n=2m$ or $2m+1$, used when discussing the parity of $n$.
    \end{itemize}
    Each of these forms of $n$ is used in this paper. We try to adhere to these variables and notation, for consistency and ease of reading.
\end{notation}
\begin{notation}
     $R(n)=\{\rho_i \mid n_{\rho_i} = 1\}$, the set of the positions where the entries in the binary expansion of $n$ are $1$. Thus, $n=\sum_{\rho_i \in R(n)}2^{\rho_i}$. 
     By convention, we set $\rho_{\omega(n)} >  \dots > \rho_2 > \rho_1$.
\end{notation}
\begin{notation}
    The number of $S$-nodes in a divide-and-conquer tree is denoted by $\sigma(n)$. The number of $D$-nodes in a divide-and-conquer tree is denoted by $\delta(n)$.
\end{notation}
\subsection{Some useful lemmas}
The lemmas in this section are simple and appear in many papers, but are used throughout this paper, so we state them here for reference. They appear in particular in Section \ref{danc_dnodes_section} on $D$-nodes of divide-and-conquer trees, in Section \ref{max_S_section} on trees having minimal $D$-nodes, and in Section \ref{colless_ladder_section} on trees having minimal $D$-nodes with base ladder trees.
\begin{definition}
    The \emph{weight} of $n$ is  $\omega(n)=\sum\limits_{i=0}^{k-1}n_i$, the number of $1$s in the binary expansion of $n$.
\end{definition}
Lemmas \ref{weightlemma_count} through \ref{weightlemma_addtwo} can be seen to be true through inspection.
\begin{lemma}\label{weightlemma_count}
      Let  $R(n)=\{\rho_i \mid n_{\rho_i} = 1\}$. Then $\lvert R(n)\rvert = \omega(n)$.
\end{lemma}
\begin{lemma}\label{weightlemma_odd}
    If $n$ is odd, then $\omega(n-1)=\omega(n)-1$. 
\end{lemma}
\begin{lemma}\label{weightlemma_even}
    If $n$ is even, then $\omega(n+1) = \omega(n)+1$.
\end{lemma}
\begin{lemma}\label{weightlemma_mult2}
    $\omega(2m) = \omega(m)$.
\end{lemma}
\begin{lemma} \label{weightlemma_2kless1}
    $\omega(2^k-1)=k$
\end{lemma}
\begin{lemma}\label{weightlemma_addone}
    If $n=2^k+r$, with $0 \le r < 2^k$, then $\omega(n) = \omega(r)+1$.
\end{lemma}
\begin{lemma} \label{weightlemma_addtwo}
    $\omega(a)+\omega(b)=\omega(a+b)$ if and only if no index $i$ is such that both $a_i$ and $b_i$ are $1$.
\end{lemma}
\begin{lemma} \label{weightlemma}
    Let $r$ be odd, where $\floor{\log_2(r)} <  k$. Then $\omega(2^{k+1}-r)=2+k-\omega(r)$.
\end{lemma}
\begin{proof}
    The binary representation of $2^{k+1}$ is $1\ 0...0$, where there are $k$ $0$s following the leading $1$. Let the binary representation of $r$ be $0\ 0\ r_{k-1}\dots\ r_1\ 1$, since $r$ is odd and $\floor{\log_2(r)} <  k$. Subtracting $r$ from $2^{k+1}$ gives $0\ 1\ !{r_{k-1}}\dots\ !{r_1}\ 1$, where $!r_i$ is the negation of $r_i$. The lemma follows.
\end{proof}
    The following lemma is well-known and may be found in many references.
    \begin{lemma}
        A full binary tree on $n$ leaves has $n-1$ internal nodes.
    \end{lemma}
    Because of pairwise commutativity, we may consider as a canonical form only trees in which the left child always has at least as many descendant leaf nodes as the right child. The trees in this paper will often be of that form. 
    \begin{notation}
        We use the symbol $\odot$ to represent the composition of trees, so if $T_L$ and $T_R$ are joined to become the left and right subtrees of a new tree, we denote this tree by $T_L \odot T_R$.
    \end{notation}
\subsection{Special tree forms} \label{specialforms}
\begin{figure}[h!tb]
	\tikzstyle{mynode}=[circle, draw]
	\tikzstyle{leaf}=[fill=gray!45]
	\tikzstyle{myarrow}=[]
	\centering
	\begin{subfigure}[b]{.22\textwidth}
	    \centering
	    \resizebox{0.9\textwidth}{!}{%
		\begin{tikzpicture}	
			\centering
			\node[mynode] (L1) at (0, 3) {\footnotesize\texttt{}};
			\node[mynode] (L2a) at (-1, 2) {\footnotesize\texttt{}};
			\node[mynode,leaf] (L2b) at (1, 2) {\footnotesize\texttt{}};
			\node[mynode] (L3a) at (-2, 1) {\footnotesize\texttt{}};
			\node[mynode,leaf] (L3b) at (0, 1) 
			{\footnotesize\texttt{}};
			\node[mynode] (L4a) at (-3, 0) {\footnotesize\texttt{}};
			\node[mynode,leaf] (L4b) at (-1, 0) 			{\footnotesize\texttt{}};
			\node[mynode] (L5a) at (-4, -1) {\footnotesize\texttt{}};
			\node[mynode,leaf] (L5b) at (-2, -1) 			{\footnotesize\texttt{}};
			\node[mynode, leaf] (L6a) at (-5, -2) {\footnotesize\texttt{}};
			\node[mynode,leaf] (L6b) at (-3, -2) 			{\footnotesize\texttt{}};
			\draw[myarrow] (L1) to node[midway,above=0pt, left=7pt]{} (L2a);
			\draw[myarrow] (L1) to node[midway,above=0pt, left=-2.5pt]{} (L2b);
			\draw[myarrow] (L2a) to node[midway,above=0pt, left=4pt]{} (L3a);
			\draw[myarrow] (L2a) to node[midway,above=0pt, right=-1pt]{} (L3b);
			\draw[myarrow] (L3a) to node[midway,above=0pt, left=0mm]{} (L4a);
			\draw[myarrow] (L3a) to node[midway,above=0pt, right=-2pt]{} (L4b);
			\draw[myarrow] (L4a) to node[midway,above=0pt, left=0mm]{} (L5a);
			\draw[myarrow] (L4a) to node[midway,above=0pt, right=-2pt]{} (L5b);
			\draw[myarrow] (L5a) to node[midway,above=0pt, left=0mm]{} (L6a);
			\draw[myarrow] (L5a) to node[midway,above=0pt, right=-2pt]{} (L6b);
		\end{tikzpicture}
		} 
		\caption{A ladder tree on $6$ leaves.}
		\label{fig:treetypes-a}
	\end{subfigure}\hfill%
	\begin{subfigure}[b]{.22\textwidth}
	    \centering
	    \resizebox{0.8\textwidth}{!}{%
		\begin{tikzpicture}	
			\centering
			\node[mynode] (L1) at (0, 3) {\footnotesize\texttt{}};
			\node[mynode] (L2a) at (-1, 2) {\footnotesize\texttt{}};
			\node[mynode] (L2b) at (1, 2) {\footnotesize\texttt{}};
			\node[mynode] (L3a) at (-2, 1) {\footnotesize\texttt{}};
			\node[mynode,leaf] (L3b) at (-1, 1) 
			{\footnotesize\texttt{}};
			\node[mynode] (L3c) at (1, 1) {\footnotesize\texttt{}};
			\node[mynode,leaf] (L3d) at (2, 1) {\footnotesize\texttt{}};
			\node[mynode, leaf] (L4a) at (-3, 0) {\footnotesize\texttt{}};
			\node[mynode,leaf] (L4b) at (-2, 0) {\footnotesize\texttt{}};
			\node[mynode,leaf] (L4e) at (1, 0) {\footnotesize\texttt{}};
			\node[mynode,leaf] (L4f) at (2, 0) {\footnotesize\texttt{}};
			\draw[myarrow] (L1) to node[midway,above=0pt, left=7pt]{} (L2a);
			\draw[myarrow] (L1) to node[midway,above=0pt, left=-2.5pt]{} (L2b);
			\draw[myarrow] (L2a) to node[midway,above=0pt, left=4pt]{} (L3a);
			\draw[myarrow] (L2a) to node[midway,above=0pt, right=-1pt]{} (L3b);
			\draw[myarrow] (L2b) to node[midway,above=0pt, left=-0.5pt]{} (L3c);
			\draw[myarrow] (L2b) to node[midway,above=0pt, left=-1pt]{} (L3d);
			\draw[myarrow] (L3a) to node[midway,above=0pt, left=0mm]{} (L4a);
			\draw[myarrow] (L3a) to node[midway,above=0pt, right=-2pt]{} (L4b);
			\draw[myarrow] (L3c) to node[midway,above=0pt, left=0mm]{} (L4e);
			\draw[myarrow] (L3c) to node[midway,above=0pt, right=-2pt]{} (L4f);
		\end{tikzpicture}
		} 
		\caption{A divide and conquer tree on $6$ leaves.}
		\label{fig:treetypes-b}
	\end{subfigure}\hfill%
	\begin{subfigure}[b]{.22\textwidth}
	    \centering
	    \resizebox{0.8\textwidth}{!}{%
		\begin{tikzpicture}	
			\centering
			\node[mynode] (L1) at (0, 3) {\footnotesize\texttt{}};
			\node[mynode] (L2a) at (-1, 2) {\footnotesize\texttt{}};
			\node[mynode] (L2b) at (1, 2) {\footnotesize\texttt{}};
			\node[mynode] (L3a) at (-2, 1) {\footnotesize\texttt{}};
			\node[mynode] (L3b) at (-1, 1)
			{\footnotesize\texttt{}};
			\node[mynode,leaf] (L3c) at (1, 1) {\footnotesize\texttt{}};
			\node[mynode,leaf] (L3d) at (2, 1) {\footnotesize\texttt{}};
			\node[mynode, leaf] (L4a) at (-3, 0) {\footnotesize\texttt{}};
			\node[mynode,leaf] (L4b) at (-2, 0) {\footnotesize\texttt{}};
			\node[mynode,leaf] (L4c) at (-1, 0) {\footnotesize\texttt{}};
			\node[mynode,leaf] (L4d) at (-0, 0) {\footnotesize\texttt{}};
			\draw[myarrow] (L1) to node[midway,above=0pt, left=7pt]{} (L2a);
			\draw[myarrow] (L1) to node[midway,above=0pt, left=-2.5pt]{} (L2b);
			\draw[myarrow] (L2a) to node[midway,above=0pt, left=4pt]{} (L3a);
			\draw[myarrow] (L2a) to node[midway,above=0pt, right=-1pt]{} (L3b);
			\draw[myarrow] (L2b) to node[midway,above=0pt, left=-0.5pt]{} (L3c);
			\draw[myarrow] (L2b) to node[midway,above=0pt, left=-1pt]{} (L3d);
			\draw[myarrow] (L3a) to node[midway,above=0pt, left=0mm]{} (L4a);
			\draw[myarrow] (L3a) to node[midway,above=0pt, right=-2pt]{} (L4b);
			\draw[myarrow] (L3b) to node[midway,above=0pt, left=0mm]{} (L4c);
			\draw[myarrow] (L3b) to node[midway,above=0pt, right=-2pt]{} (L4d);
		\end{tikzpicture}
		} 
		\caption{A complete full binary tree on 6 leaves.}
		\label{fig:treetypes-c}
	\end{subfigure}\hfill%
	\begin{subfigure}[b]{.22\textwidth}
	    \centering
	    \resizebox{.9\textwidth}{!}{%
		\begin{tikzpicture}	
			\centering
			\node[mynode] (L1) at (0, 3) {\footnotesize\texttt{}};
			\node[mynode] (L2a) at (-1, 2) {\footnotesize\texttt{}};
			\node[mynode] (L2b) at (1, 2) {\footnotesize\texttt{}};
			\node[mynode] (L3a) at (-2, 1) {\footnotesize\texttt{}};
			\node[mynode] (L3b) at (-1, 1)
			{\footnotesize\texttt{}};
			\node[mynode] (L3c) at (1, 1) {\footnotesize\texttt{}};
			\node[mynode] (L3d) at (2, 1) {\footnotesize\texttt{}};
			\node[mynode, leaf] (L4a) at (-3, 0) {\footnotesize\texttt{}};
			\node[mynode,leaf] (L4b) at (-2, 0) {\footnotesize\texttt{}};
			\node[mynode,leaf] (L4c) at (-1.25, 0) {\footnotesize\texttt{}};
			\node[mynode,leaf] (L4d) at (-.5, 0) {\footnotesize\texttt{}};
			\node[mynode,leaf] (L4e) at (0.5, 0) {\footnotesize\texttt{}};
			\node[mynode,leaf] (L4f) at (1.25, 0) {\footnotesize\texttt{}};
			\node[mynode,leaf] (L4g) at (2, 0) {\footnotesize\texttt{}};
			\node[mynode,leaf] (L4h) at (3, 0) {\footnotesize\texttt{}};
			\draw[myarrow] (L1) to node[midway,above=0pt, left=7pt]{} (L2a);
			\draw[myarrow] (L1) to node[midway,above=0pt, left=-2.5pt]{} (L2b);
			\draw[myarrow] (L2a) to node[midway,above=0pt, left=4pt]{} (L3a);
			\draw[myarrow] (L2a) to node[midway,above=0pt, right=-1pt]{} (L3b);
			\draw[myarrow] (L2b) to node[midway,above=0pt, left=-0.5pt]{} (L3c);
			\draw[myarrow] (L2b) to node[midway,above=0pt, left=-1pt]{} (L3d);
			\draw[myarrow] (L3a) to node[midway,above=0pt, left=0mm]{} (L4a);
			\draw[myarrow] (L3a) to node[midway,above=0pt, right=-2pt]{} (L4b);
			\draw[myarrow] (L3b) to node[midway,above=0pt, left=0mm]{} (L4c);
			\draw[myarrow] (L3b) to node[midway,above=0pt, right=-2pt]{} (L4d);
			\draw[myarrow] (L3c) to node[midway,above=0pt, left=0mm]{} (L4e);
			\draw[myarrow] (L3c) to node[midway,above=0pt, right=-2pt]{} (L4f);
			\draw[myarrow] (L3d) to node[midway,above=0pt, left=0mm]{} (L4g);
			\draw[myarrow] (L3d) to node[midway,above=0pt, right=-2pt]{} (L4h);
		\end{tikzpicture}
		}
		\caption{A perfect tree on $2^3=8$ leaves.}
		\label{fig:treetypes-d}
	\end{subfigure} %
	\mycaption{Examples of tree types: ladder, divide-and-conquer, complete full binary, perfect} {The divide-and-conquer tree on $6$ leaves shown in subfigure (\subref{fig:treetypes-b}) is very  similar to the complete full binary tree on $6$ leaves in subfigure (\subref{fig:treetypes-c}), differing in the placement of their lowest-level leaves and in the tree topology that results from this.}
	\label{fig:treetypes}		
\end{figure}
\begin{definition}
	A \emph{ladder tree} is a full binary tree in which every interior node has a leaf child and a child that is another ladder tree, with the exception of the lowest interior node, which has two leaf children. These correspond to serial calculations. An example is shown in Fig. \ref{fig:treetypes}(\subref{fig:treetypes-a}).
\end{definition}
\begin{definition}
    A \emph{divide-and-conquer tree} is a full binary tree in which the number of leaf descendants of the left and right children of any node differ by at most $1$. These are as nearly balanced as possible. These correspond to divide-and-conquer calculations. An example is shown in Fig. \ref{fig:treetypes}(\subref{fig:treetypes-b}).
\end{definition}
\begin{definition}
    A \emph{complete full binary tree} is a tree in which every node has either two or zero children, every level except the last is completely filled, and all nodes are as far to the left as possible. An example is shown in Fig. \ref{fig:treetypes}(\subref{fig:treetypes-c}).
\end{definition}
\begin{definition}
	A \emph{perfect binary tree} is a full binary tree in which all levels $i$ have the maximum number $2^i$ nodes. A perfect tree has $2^k$ leaves and $2^k-1$ interior nodes \cite{knuth97art1}. These are perfectly balanced trees. Perfect trees are divide-and-conquer trees, and are also complete full binary trees. An example is shown in Fig. \ref{fig:treetypes}(\subref{fig:treetypes-d}).
\end{definition}

\section{Characterizing trees in terms of S- and D-nodes} \label{s-nodes-counting}

All members of a class, and its parenthetic form as well, have the same number of $S$-nodes, by Lemma \ref{iso_SD}.
If this number is known or can be quantified, there is a simple formula to count the products that have that form.

This is especially useful when discussing sets of commutative non-associative products that share a particular parenthetic form, or else that have a known number of $S$-nodes. We discuss several cases of that nature in this section.
\subsection{Two known counting formulas} \label{known}
We start with two known formulas counting all parenthetic forms and commutative non-associative products on $n$ terms, without reference to $S$- or $D$-nodes. Forms are defined by the tree structure, and products are defined not only on the tree structure but also the named terms (or leaves).
The number of parenthetic forms increases much more  slowly than the total number of products. 

An equivalent observation to Proposition \ref{halfcat} is made by Bohl \cite{bohl06}. Proposition \ref{ineq_sum} is known to apply to leaf-labeled binary trees, and is shown by Stanley \cite{stanley97},  Callan \cite{callan09}  and Dale \cite{dale93}. Walters makes this observation on commutative, non-associative multiplication in OEIS entry \seqnum{A001147}  \cite{OEIS}.

\begin{proposition}\label{halfcat} \cite{bohl06}
	The number of non-isomorphic parenthetic forms with $n$ undefined terms is 
	$$
	\alpha(n)=\sum_{i=1}^{\floor{\frac{n}{2}}}\alpha(i)\alpha(n-i)
	$$
\end{proposition}
\begin{proposition}\label{ineq_sum}\cite{stanley97, callan09, dale93}
	The number of  computationally inequivalent commutative non-associative products on $n$ identified terms is
	$(2n-3)!!$
\end{proposition}

The sequence in Proposition \ref{halfcat}  is OEIS \seqnum{A000992}, the $n^\text{th}$ half-Catalan number.  
The sequence in Proposition \ref{ineq_sum} is OEIS \seqnum{A001147} \cite{OEIS}. There are many interpretations of these sequences. 

\subsection{Counting parenthetic forms} 
All isomorphic parenthetic forms have the same number of $S$-nodes. However, it is possible for two forms to have the same number of $S$ nodes and yet be non-isomorphic. 

In this section, we break out the complete set of parenthetic forms on $n$ terms into sets of parenthetic forms having the same number of $S$-nodes. 

\begin{theorem}\label{propkSnodes}
	The number of non-isomorphic leaf-unlabeled full binary trees (or parenthetic forms) with $n$ unlabeled leaf nodes (or summands) and $s$ $S$-nodes is  
	\begin{equation} \label{numforms_s}
	\theta(n,s) = 
	\begin{cases}
	\sum\limits_{j=1}^{\floor{\frac{n-1}{2}}}\sum\limits_{i=0}^s \theta(j,i)\theta(n-j,s-i), & \text{if $n$ is odd}; \\
	\sum\limits_{j=1}^{\floor{\frac{n-1}{2}}}\sum\limits_{i=0}^s \theta(j,i)\theta(n-j,s-i) + \sum\limits_{i=0}^{s-1} \theta(\frac{n}{2},i)\theta(\frac{n}{2},s-1-i), & \text{if $n$ is even}. \\
	\end{cases}
	\end{equation}
	\begin{equation*}
	\text{where } \theta(n,s)=
	\begin{cases}
	1, & \text{if } n=0 \text{ and } s=0 \text{ (the unique empty tree)};  \\
	1, & \text{if } n=1 \text{ and } s=0 \text{ (the unique single-node tree)};\\
	0, & \text{if } s\ge n>0.\\
	\end{cases}
	\end{equation*}
\end{theorem}
\begin{proof}
	At every stage, the first term $n$ in the calculation is reduced by half, so this recursion terminates, according to the rules for termination of $\theta(n,s)$ above.
	
	If $n$ is odd, the children of the top node cannot both have the same number of leaf descendants, so the top node cannot be an $S$-node. The formula follows by counting the sub-trees on the 2-partitions of $n$, where the left sub-tree has $i$ $S$-nodes  and the right has $s-i$ $S$-nodes, for  $0 \le i \le  s$, and the two counts are multiplied together.
	
	If $n$ is even,  the formula is as above, but we must also include the case where the top node is an S-node. In that case, both children have $\frac{n}{2}$ leaf nodes, and the summation is on sub-trees where there are  $\frac{n}{2}$ leaf nodes in each sub-tree. In this case, the root is one of the $S$-nodes, so the left sub-tree has $i$ $S$-nodes  and the right has $s-1-i$ $S$-nodes, for  $0 \le i \le  s-1$. 
\end{proof}
\begin{corollary} \label{rowsum}
	$\sum\limits_{i=1}^{n-1} \theta(n,i) = \alpha(n)$, the number of parenthetic forms on $n$ terms.
\end{corollary}
\begin{proof}
	Follows from Proposition \ref{halfcat} and Theorem \ref{propkSnodes}.
\end{proof}
	We calculate the number of non-isomorphic parenthetic forms up to $n=16$ in Table \ref{paren_nk}, as per Equation \ref{numforms_s}, in Theorem  \ref{propkSnodes}.
Equation \ref{numforms_s} is recursive and the proof is constructive, so one can more or less laboriously calculate all of the actual parenthetic forms on $n$ summands.

\subsection{Counting products using S-nodes} \label{method}
\begin{proposition} \label{sd_tree_enum}
	Let $T$ be an unlabeled SD-tree with $n$ leaf nodes, and let $s$ be the number of $S$-nodes of $T$. Then the number of inequivalent leaf-labeled  SD-trees that have form isomorphic to T is  $\frac{n!}{2^s}$. 
\end{proposition}
\begin{proof}
	There are $n!$ ways to order the $n$ leaf nodes of the tree $T$.  As above, we consider only SD-trees where the number of leaf descendants of the left child of a node is always less than or equal to the number of leaf descendants of the right child. 
	
	Call the set of left descendant leaf nodes of a node $d_\ell$, and the set of right descendant leaf nodes $d_r$. $|d_\ell|=|d_r|$ if and only if an ordering transposing $d_\ell$ and $d_r$ gives a tree that is the same as the original tree, up  to transposition of  the two subtrees. So  such subtrees are combinatorially counted twice. The number of trees that are the same as another given ordering is $2^s$, so the total number of non-equivalent trees isomorphic to T is $\frac{n!}{2^s}$.
\end{proof}
\begin{corollary}\label{equiv_sum}
	Let $P$ be a parenthetic form on $n$ operands, let $T_P$ be the unlabeled SD-tree representing $P$, and let $s$ be the number of $S$-nodes of $T_P$. Then the number of inequivalent commutative non-associative products having form represented by $P$ is $\frac{n!}{2^s}$. 
\end{corollary}

Corollary \ref{equiv_sum} means that when counting a set of commutative non-associative products sharing a parenthetic form, it suffices to have a formula for the number of $S$ nodes in the representative SD tree. We will follow this method throughout this paper to count commutative non-associative operations of certain interesting forms.

\subsection{Upper and lower bounds} \label{lu_bounds}
We present here upper and lower bounds for parenthetic forms and for the number of commutative non-associative products having a given form. All of these bounds are met, as will be discussed in Sections \ref{ladder_pf} and \ref{max_S_section}.
\begin{proposition} \label{upper_bound_s}
    The upper bound on $S$-nodes in a parenthetic form having $n$ leaves is $s(n)$, where $2^{s(n)}$ is the highest  power of  $2$  that  divides $n!$.
\end{proposition}
\begin{proof}
    This is true since $2^{s(n)}$ must divide $n!$, by Proposition \ref{sd_tree_enum}.
\end{proof}
\begin{corollary}\label{lower_bound}
	The lower bound for the number of inequivalent commutative non-associative products on $n$ variables having parenthetic form represented by a set of isomorphic SD-trees is
	$\frac{n!}{2^{s(n)}}$, where $2^{s(n)}$ is the highest  power of  $2$  that  divides $n!$. 
\end{corollary}
\begin{proof}
	Follows from Corollary \ref{equiv_sum} and Proposition \ref{upper_bound_s}.
\end{proof}
\begin{proposition} \label{lower_bound_s}
    The lower bound on $S$-nodes in a parenthetic form is $1$.
\end{proposition}
\begin{proof}
    This can be seen by considering that all nodes in an SD-tree must have zero or two children, so leaf nodes at the lowest level must come in pairs. Any tree must have at least one of these pairs. 
\end{proof}
\begin{corollary}\label{upper_bound}
	The upper bound for the number of inequivalent commutative non-associative products on $n$ variables on a class of isomorphic SD-trees is $\frac{n!}{2}$.  
\end{corollary}
\begin{proof}
Follows from Corollary \ref{equiv_sum} and Proposition \ref{lower_bound_s}. 
\end{proof}
	The sequence $2^{s(n)}$ from Proposition~\ref{upper_bound_s} is OEIS \seqnum{A060818}
	\cite{OEIS}, the largest power of $2$ that divides $n!$.
	The sequence $\frac{n!}{2^{s(n)}}$ from Proposition~\ref{lower_bound} is OEIS \seqnum{A049606}
	\cite{OEIS}, the largest odd divisor of $n!$.
	The upper bound sequence from Proposition~\ref{upper_bound} is OEIS \seqnum{A001710} 
	\cite{OEIS}
	, the number of even permutations on $n$ letters.

\section{A few special tree forms} \label{special}
In this section, we discuss special forms,  including the ladder forms that have exactly 1 $S$-node, the forms having exactly 2 $S$-nodes, and the forms having a maximal number of $S$-nodes. We emphasize the  important case of divide-and-conquer products, which  are widely used, and which are related to other mathematical constructs.

\subsection {Trees with one S-node: ladder products} \label{ladder_pf}
\begin{lemma}
	There is a unique parenthetic form on $n>1$ leaf nodes with exactly one $S$-node.
\end{lemma}
\begin{proof}
	Such a tree exists: if $n=2$, it is the tree with a root and two leaves; if not, it is the tree in which every internal node has a left child that is the unique parenthetic form with one $S$-node on $n-1$ leaf nodes and a right leaf child.
	
	It is unique: 
	Any subtree of an SD-tree must have at least one $S$-node, except when it consists of a single leaf node. This means one of the top branches of an SD-tree with exactly one $S$-node is a single leaf node. The conclusion follows by induction on the number of leaf nodes.
\end{proof}

This type of tree is called a \emph{ladder} or \emph{sequential}  or \emph{comb left} tree. The \emph{ladder product} instantiates the ladder tree, and is one in which the operation proceeds pairwise in the order in which the terms appear. For example, the ladder product on four operands is $(((a*b)*c)*d)$.

\begin{lemma} \label{laddercor}
	Ladder SD-trees are exactly those with one $S$-node, up to isomorphism.
\end{lemma}
\begin{proposition}\label{numladder}
	The number of computationally inequivalent commutative non-associative sequential products on $n$ variables is 
	$$
	\frac{n!}{2}
	$$
\end{proposition}
\begin{proof}
	Follows immediately from Corollary \ref{equiv_sum} and Lemma \ref{laddercor}.
\end{proof}
    Ladder trees meet the lower bound for $S$-nodes discussed in Section \ref{lu_bounds}. The number of ladder products meets the lower bound for the number of inequivalent commutative non-associative products on $n$ variables.  
\subsubsection{Computational applications of ladder products.}
	Ladder summation corresponds to the C language default of left-to-right associativity on summations with ungrouped summands \cite{C18}. 	
	By pairwise commutativity, the same result is guaranteed in C on an IEEE-754-compliant system upon transposition of the first two elements of an ungrouped summation,
	but is not guaranteed after any other transposition. 

\subsection{Trees with two S-nodes } \label{2_snodes}
\begin{proposition}\label{prop2Snodes}
	The number of parenthetic forms on $n \ge 1$ with exactly $2$ $S$-nodes is
	$$
	\theta(n,2) =
	\begin{cases}
	(m-1)^2, &   \text{if } n=2m+1; \\
	(m-1)(m-2), &   \text{if } n=2m. \\
	\end{cases}
	$$
\end{proposition}
\begin{proof}
	We proceed by induction, and observe that the proposition is  true by inspection for $n=1,2,3,4$. We also observe that  no tree having exactly 2 $S$-nodes can have an $S$-node root. This means that either 
	(a) one of the subtrees below the root has two $S$-nodes and the other has none, or else 
	(b) both subtrees have exactly one $S$-node.  
	\newline
	In case (a), the subtree with no $S$-nodes consists of only one leaf node. The other subtree has $n-1$ leaf nodes and $2$ $S$-nodes. So the number of trees  with such subtrees is $\theta(n-1,2)$. 
	\newline
	In case (b), each subtree has one $S$-node, so by Corollary \ref{laddercor}, it is the unique ladder tree of its size. Each subtree with $\ell$ leaf nodes has a sibling with $n-\ell$ leaf nodes, and there is one ladder tree at each of these sizes, so the number of such double ladder trees is $\floor{\frac{n}{2}}-1$ if $n$ is odd, and ${\frac{n}{2}}-2$ if $n$ is even (excluding when the siblings have an equal number of leaf nodes).
	\newline
	Putting (a) and (b) together, $\theta(n,2)=\theta(n-1,2)+\floor{\frac{n}{2}}-1$. Using the induction hypothesis,
	\begin{equation*}
	\theta(n,2) =
	\begin{cases}
	(m-1)(m-2)+(m-1) = (m-1)^2, &  \text{if } n=2m+1; \\
	(m-2)^2+(m-2)= (m-1)(m-2), &   \text{if } n=2m. \\
	\end{cases}
	\end{equation*}
\end{proof}

\subsection{Complete full binary trees} \label{cfb}

\begin{definition}
    A \emph{complete full binary tree} is a tree in which every node has either two or zero children, every level except the last is completely filled, and all nodes are as far to the left as possible. 
\end{definition}
    The complete full binary tree is similar to a divide-and-conquer tree, except in the spacing of its lowest level \cite{knuth97art1}. The lowest level in a complete full binary tree is filled from left to right, in contrast to the even spacing of lowest level leaf pairs in the divide-and-conquer tree.  

\begin{lemma} \label{cft_mod}
    Let $n=2^k+r$, with $0 \le r < 2^k$. The complete full binary tree on $n$ leaves has a perfect tree as one of its subtrees, and a tree with $(2^{k-1}+(r \bmod 2^{k-1}))$ leaves as the other subtree. In particular,
    \begin{itemize}
        \item if $r<2^{k-1}$, the left subtree has $(2^{k-1}+r)$ leaves, and the right subtree is a perfect tree with $2^{k-1} $ leaves.
        \item if $r=2^{k-1}$, the left subtree is a perfect tree with $2^k$ leaves, and the right subtree is a perfect tree with $2^{k-1} $ leaves.
        \item if $r>2^{k-1}$ the left subtree is a perfect tree with $2^k$ leaves, and the right subtree has $(2^{k-1}+(r \bmod 2^{k-1}))$ leaves.
    \end{itemize}
\end{lemma}
\begin{proof}
    If $r \le 2^{k-1}$, then the right subtree has $2^{k-1}$ leaves and is perfect, and the left has $n-2^{k-1}=2^k+r-2^{k-1}=2^{k-1}+r=2^{k-1}+(r\bmod 2^{k-1})$ leaves.  
    
    If $r > 2^{k-1}$, then the left subtree has $2^k$ leaves and is perfect,  and the right subtree has $n=(2^k+r)-2^k = r=2^{k-1}+(r \bmod 2^{k-1})$ leaves.
\end{proof}

\begin{theorem} \label{delta_cft}
    Let $n=2^\ell \cdot d$, where $d$ is odd. Then the number of $D$-nodes in a complete full binary tree with $n$ leaves is $\floor{\log_2(d)}$.
\end{theorem}
\begin{proof}
    Let $n= 2^k+r$, where $0 \le  r < 2^k$. 
    
    \emph{Case 1: $r=2^{k-1}$.} Then $n=2^{k-1}\cdot 3$, $d=3$ and $\floor{\log_2(d)}=1$. There is exactly $1$ $D$-node, which is the root, since the two subtrees are perfect, and contribute $0$ $D$-nodes.
    
    \emph{Case 2: $r\ne 2^{k-1}$.} By Lemma \ref{cft_mod}, one subtree is perfect, and the non-perfect subtree has $(2^{k-1}+(r\bmod 2^{k-1}))$ leaves. $(r \bmod 2^{k-1}) = ((2^{k+1}+r)) \bmod 2^{k-1}$, so without loss of generality, we may assume that $r < 2^{k-1}$. So $(r\bmod 2^{k-1})=r$. 
    
    $2^\ell$ divides $n$ and $2^k$, so must divide $r$. Let $r$ be $2^\ell \cdot b$, with $b$ odd. The non-perfect subtree has $(2^{k-1}+2^\ell b) = 2^\ell \cdot (2^{k-1-\ell}+ b)$ leaves, with $b<2^{k-1-\ell}$, so by induction, it contributes $\floor{\log_2(2^{k-1-\ell}+ b)}=(k-1-\ell)$ $D$-nodes.
    The root contributes $1$ as a $D$-node, so the total number of $D$ nodes is $(k-1-\ell)+1=k-\ell=\floor{\log_2(d)}$. 
\end{proof}
\begin{corollary}
         Let $n=2^\ell \cdot d$, where $d$ is odd. Then the number of $S$-nodes in a complete full binary tree with $n$ leaves is  $n-1-\floor{\log_2(d)}$.
\end{corollary}
\begin{proof}
    Follows from Lemma \ref{dnodes} and Theorem \ref{delta_cft}.
\end{proof}
\begin{proposition} \label{cfb_oeis}
    Let $n>0$ and let $n=2^\ell \cdot d$, where $d$ is odd. Then the number of $D$-nodes in a complete full binary tree is the number of binary digits that are the same in the binary expansions of $n$ and $(n-1)$. 
\end{proposition}
\begin{proof}
    The binary expansion of $n$ has $\ell$ trailing $0$s, with $1$ in the $\ell^\text{th}$ bit.The binary expansion of $(n-1)$ has $\ell$ trailing $1$s, with $0$ in the $\ell^\text{th}$ bit. There are then  $(k-\ell)$ binary digits that are the same in the binary expansions of $n$ and $(n-1)$. $d$ is odd, so this number is $\floor{\log_2(d)}$.
\end{proof}
    Proposition \ref{cfb_oeis} shows that the number of $D$-nodes in a complete full binary tree is OEIS \seqnum{A119387}
        \cite{OEIS}, 
        with an offset.

\subsection{Divide-and-conquer trees} \label{pairwise_cnaos}
In this section, we discuss evenly partitioned divide-and-conquer methods on commutative non-associative operations.
Divide-and-conquer methods proceed by dividing the set of $n$ terms into two subset, operating on the subsets, then operating on the two results. 
Evenly partitioned divide-and-conquer methods divide the $n$ terms into two subsets of size $\frac{n}{2}$ if $n$ is even, or into two subsets of size $\floor{\frac{n}{2}}$ and $\ceil{\frac{n}{2}}$ if $n$ is odd. Throughout this section we assume that a divide-and-conquer algorithm is evenly partitioned.

Evenly partitioned divide-and-conquer methods give rise to a full binary tree in which the number of leaf descendants of the left and right children of any node differ by at most $1$. The binary tree generated is thus as balanced as it can be. At every level, the children of a node are evenly or almost evenly divided between the left and the right branches.

We discuss in this section three different formulas for the number of $S$-nodes in an evenly partitioned divide-and-conquer tree. We also show the corresponding three formulas for the number of evenly partitioned divide-and-conquer products on $n$ terms, using the SD-tree-based approach from Corollary \ref{equiv_sum}. In particular, we show new non-recursive closed forms for each of these. 
We also provide several new formulas for $D$-nodes in these trees that we have not found in the literature.

The analysis of the number of $S$-nodes also gives rise to interesting correspondences with the Takagi function  \cite{takagi01}.

\begin{proposition}\label{pair_iso}
	All divide-and-conquer SD-trees on commutative non-associative operations on $n$ elements are isomorphic, and therefore have the same number of  $S$-nodes.
\end{proposition}
\begin{proof}
	A divide-and-conquer SD-tree is constructed by extending each node by two children. If a node has   $k$  leaf  descendants,  then its children must have $\floor{\frac{k}{2}}$ and $\ceil{\frac{k}{2}}$ leaf  descendants. These  can be transposed, by pairwise  commutativity. It follows  by Lemma \ref{iso_SD} that they have the  same  number of $S$-nodes.
\end{proof}
\begin{proposition}
	Any commutative non-associative product with tree representation isomorphic to a divide-and-conquer SD-tree is itself a divide-and-conquer product.
\end{proposition}
\begin{proof}
	Follows from the definition of divide-and-conquer.
\end{proof}

\subsubsection{Computational applications of divide-and-conquer trees}
Divide-and-conquer is used in many computational algorithms, due to its $\mathcal{O}(\log_2{}n)$ performance. For example, in context of floating-point summation, the evenly partitioned divide-and-conquer method is called pairwise (or cascade) summation. 
Pairwise summation is the default on ungrouped summands in NumPy \cite{numpysum} and in Julia \cite{julia}.
Pairwise summation is known to be fairly accurate, and in some cases is nearly as accurate as such gold-standard techniques as Kahan summation \cite{higham93}. 

As another example, the balanced ancestry trees seen in mathematical phylogenetics \cite{colless80, heard92, mooers97} are divide-and-conquer trees. The maximally balanced trees discussed at length in \cite{coronado20} and elsewhere are divide-and-conquer trees. 

\subsubsection{Counting divide-and-conquer S-nodes} \label{dandc_snodes}
The number of $S$-nodes in a divide-and-conquer tree is non-linear, and in fact is not even monotonic. We develop formulas for $\sigma(n)$, the number of $S$-nodes in an $n$-term divide-and-conquer tree here. One of these formulas is in non-recursive closed form.

\begin{theorem}\label{exp_recursive}
	The number of $S$-nodes in a divide-and-conquer SD-tree with $n$ leaf nodes is
	\begin{equation}\label{recursive_A268289}
		\sigma(n)=
		\begin{cases}
			2\sigma(m)+1, & \text{if}\ n=2m; \\
			\sigma(m)+\sigma(m+1), & \text{if}\ n=2m+1;\\
			0 & \text{if}\ n=1.
		\end{cases}
	\end{equation}
\end{theorem}
\begin{proof}
	All subtrees of a divide-and-conquer SD-tree are themselves divide-and-conquer SD-trees. We proceed via induction. 
	
	If $n=2m$, then the divide-and-conquer product divides the $2m$ leaf nodes evenly between the two subtrees, and the root node is an $S$-node. So the number of $S$-nodes in the SD-tree is equal to $1+2\sigma(m)$, by the induction hypothesis.
	
	If $n=2m+1$, then the divide-and-conquer product divides the $2m + 1$ leaf nodes into two subtrees having $m$ and $m+1$ nodes, and the root node is not an $S$-node. So the number of $S$-nodes in the SD-tree is equal to $\sigma(m)+\sigma(m+1)$, by the induction hypothesis.
\end{proof}
Theorem \ref{exp_recursive} shows that $\sigma(n)$ is another interpretation of sequence OEIS \seqnum{A268289}	\cite{OEIS}, the cumulative deficient binary digit sum.
The recursive Equation
(\ref{recursive_A268289}) is the formula referenced by Sloane for OEIS \seqnum{A268289}, offset by 1.
The closed forms shown in Theorem \ref{exp_closed} are new for that sequence. 
\begin{proposition}\label{prop_snodes_half_n}
    The number of $S$-nodes in a divide-and-conquer SD-tree with $n$ leaf nodes is greater than or equal to $\floor{\frac{n}{2}}$.
\end{proposition}
\begin{proof}
    Follows by induction from Theorem \ref{exp_recursive}.
\end{proof}
\begin{theorem}\label{exp_closed}
	Let the binary decomposition of $n$ be $n_k \dots n_1 n_0$, where $n_k=1$. Then the number of $S$-nodes in a divide-and-conquer SD-tree with $n$ leaf nodes  is
	\begin{equation}\label{closed_clear_A268289}
			\sigma(n)=\sum_{i=0}^{\floor{\log_2(n)}} \beta_i(n) \text{, where } \beta_i(n)= \begin{cases}
			2^i-(n\bmod 2^i), & \text{ if } n_i=0;\\
			(n\bmod 2^i), & \text{ if } n_i=1.\\
	\end{cases}	  
\end{equation}  
An explicit form is
	\begin{equation}\label{closed_A268289}
		\sigma(n)=\sum_{i=0}^{\floor{\log_2(n)}}  \left[ \left(\left(\floor{\frac{n}{2^i}}+1\right)\bmod 2\right)\times 2^i+ (-1)^{\left(\left(\floor{\frac{n}{2^i}}+1\right)\bmod 2\right)}\times (n \bmod  2^i)\right]
	\end{equation}
\end{theorem}
\begin{proof}
	Because every SD-tree is a full complete binary tree, there are $2^i$ nodes at level  $i$ except for the bottom level. At each level, the descendant leaf nodes are almost evenly divided between the nodes: at the $i^\text{th}$ level, each of  the $2^i$ nodes has at least $\floor{\frac{n}{2^i}}$ leaf descendants, and $(n\bmod 2^i)$ of these nodes have an additional leaf descendant. 
	
	Consider the binary representation of the number $n$ where the bits are ordered from least significant to most significant.  
	If the $i^\text{th}$ bit of  $n$ is $0$, then $\floor{\frac{n}{2^i}}$ is even and $2^i-(n\bmod 2^i)$ nodes have an even number of  descendants. If the $i^\text{th}$ bit of  $n$ is $1$, then $\floor{\frac{n}{2^i}}$ is odd, and $(n\bmod 2^i)$ nodes have an even number of  descendants. A node in a divide-and-conquer tree is an $S$-node if it has an even number of leaf descendants. So $\beta_i(n)$, the number of $S$-nodes at level $i$, is
	\begin{equation}\label{eqA}
		\begin{cases}
			2^i-(n\bmod 2^i), & \text{ if the}\ i^\text{th} \text{ bit  of $n$ is }  0;\\
			(n\bmod 2^i), & \text{ if the}\ i^\text{th} \text{ bit of $n$ is }  1.\\
		\end{cases}	  
	\end{equation}  
	Equation \ref{closed_clear_A268289} is obtained by summing across the levels.
		
	We observe the following two equations: 
	\begin{equation}\label{eqB}
		\left(\floor{\frac{n}{2^i}}+1\right)\bmod 2=
		\begin{cases}
			$1$,& \text{ if the}\ i^\text{th} \text{ bit of $n$ is }  0;\\ 
			$0$,& \text{ if the}\ i^\text{th} \text{ bit of $n$ is }  1.\\ 
		\end{cases}
	\end{equation}	
	\begin{equation}\label{eqC}
		(-1)^{\left(\left(\floor{\frac{n}{2^i}}+1\right)\bmod 2\right)}=
		\begin{cases}
			$-1$,& \text{ if the}\ i^\text{th} \text{ bit of $n$ is }  0;\\ 
			$ 1$,& \text{ if the}\ i^\text{th} \text{ bit of $n$ is }  1.\\ 
		\end{cases}
	\end{equation}	
	Putting  Equations (\ref{eqA}), (\ref{eqB}),  and (\ref{eqC}) together, we see that the number of $S$-nodes at level $i$ is
	$$
		\left(\left(\floor{\frac{n}{2^i}}+1\right)\bmod 2\right)
		\times 2^i+ (-1)^{\left(\left(\floor{\frac{n}{2^i}}+1\right)\bmod 2\right)}\times (n \bmod  2^i)
	$$
	Equation \ref{closed_A268289} is obtained by summing across the levels.
\end{proof}
    Hwang provides an extensive analysis of solutions to divide-and-conquer recurrences in \cite{hwang17}. Equations (\ref{closed_clear_A268289}) and (\ref{closed_A268289}) perhaps might be obtained from Equation (\ref{recursive_A268289}) using the methods in that paper. Instead, in Theorem \ref{exp_closed}, we prove Equation (\ref{closed_A268289}) directly by analyzing $S$-nodes.

	Equation (\ref{closed_A268289}) is almost entirely in terms of bit operations and shifts, so is easy to calculate. $\floor{\log_2(n)}$ is one less than the number of bits in $n$. The expression $(\floor{\frac{n}{2^i}}+1)\bmod 2)$ is just the negation of the $i^\text{th}$ bit of $n$.
\subsubsection{Counting divide-and-conquer D-nodes}\label{danc_dnodes_section}

We state here formulas for the number of $D$-nodes in a divide-and-conquer form with $n$ leaves. These are similar to Theorems \ref{exp_recursive} and \ref{exp_closed} for the number of $S$-nodes. The number of $D$-nodes corresponds to sequence OEIS \seqnum{A296062}
\cite{OEIS}
. 
This fact is noted in \cite{coronado20}.

\begin{theorem}\label{exp_D_recursive} \cite{coronado20}
	The number of $D$-nodes in a divide-and-conquer SD-tree with $n$ leaf nodes is
	$$
		\delta(n)=
		\begin{cases}
			2\delta(m), & \text{if}\ n=2m; \\
			\delta(m)+\delta(m+1)+1, & \text{if}\ n=2m+1;\\
			0 & \text{if}\ n=1.
		\end{cases}
	$$
\end{theorem}
\begin{proof}
    Follows from Theorem \ref{exp_recursive} and Lemma \ref{dnodes}.
\end{proof}
\begin{corollary}\label{cor_dnodes_half_n}
    The number of $D$-nodes in a divide-and-conquer SD-tree with $n$ leaf nodes is less than $\floor{\frac{n}{2}}$.
\end{corollary}
\begin{proof}
    Follows by induction from Theorem \ref{exp_D_recursive}, after observing that it is true for $n=1$.
\end{proof}
\begin{corollary}\label{cor_dnodes_ineq_odd_n}
    If $n$ is odd, then $1 \le \delta(n) \le \floor{\frac{n}{2}}$.
\end{corollary}
\begin{proof}
    Follows from Theorem  \ref{exp_D_recursive} and Corollary \ref{cor_dnodes_half_n}.
\end{proof}
\begin{corollary} \label{cor_dnodes2k}
         $\delta(2^k)=0$.
\end{corollary}
\begin{proof}
    Follows from repeated application of Theorem \ref{exp_D_recursive}.
\end{proof}

The following Corollary~\ref{yet_another_dnodes_prop} applies to $n$ when $n$ is not a power of $2$. 

\begin{corollary} \label{yet_another_dnodes_prop}
    Let $n=2^k+r$, where $0<r<2^k$, and let $\rho_1$ be the position of the smallest non-$0$ bit in $n$. The number of $D$-nodes in a divide-and-conquer SD-tree with $n$ leaf nodes is
    $$
    \delta(n)=\frac{1}{2}(\delta(n-1)+\delta(n+1))+1-\rho_1
    $$
\end{corollary}
\begin{proof}
    If $r=2d+1$ is odd:  
    \begin{flalign*}
        \delta(n) &= \delta(2(2^{k-1}+d)+1)\\
        &= \delta(2^{k-1}+d)+\delta(2^{k-1}+d+1)+1&\text{ by Theorem \ref{exp_D_recursive}}&\\
        &= \frac{1}{2}(2\cdot \delta(2^{k-1}+d)+2\cdot \delta(2^{k-1}+d+1)) +1\\
        &= \frac{1}{2}(\delta(2^k+2d)+\delta(2^k+2d+2)) +1 &\text{ by Theorem \ref{exp_D_recursive}}\\
        &= \frac{1}{2}(\delta(n-1)+\delta(n+1)) +1-0 \\
        &= \frac{1}{2}(\delta(n-1)+\delta(n+1)) +1 -\rho_1 &\text{ since $n$ is odd}\\
    \end{flalign*}
    Now let $r=2d$ be even. First, we use what has been proved for $n$ odd:
    \begin{equation} \label{1steq_69}
        \delta(n-1) 
        = \frac{1}{2}(\delta(n-2)+\delta(n)) +1 
        = \frac{1}{2}\delta(n-2)+\frac{1}{2}\delta(n) +1 
    \end{equation}
    \begin{equation}\label{2ndeq_69}
        \delta(n+1) 
        = \frac{1}{2}(\delta(n)+\delta(n+2)) +1 
        = \frac{1}{2}\delta(n+2) +\frac{1}{2}\delta(n)+1 
    \end{equation}
    Equations \ref{1steq_69} and \ref{2ndeq_69} give:
    $$
        \delta(n-1)+\delta(n+1) 
        = \delta(n) +\frac{1}{2}(\delta(n-2)+\delta(n+2)) +2 
    $$
    Thus
 \begin{equation} \label{n_plusminus_2}
    \delta(n-1)+\delta(n+1)-2-\delta(n) = \frac{1}{2}(\delta(n-2)+\delta(n+2)) 
 \end{equation}
 Now we apply induction on $k$, noting that the smallest non-$0$ bit in $\frac{n}{2}=(2^{k-1}+d)$ is ($\rho_1-1$):
   \begin{flalign*}
        \delta(n) &= \delta(2(2^{k-1}+d)) \\
         &= 2\delta(2^{k-1}+d)& \text{ by Theorem \ref{exp_D_recursive}}&\\
         &= 2\cdot (\frac{1}{2}[\delta(2^{k-1}+d-1)+\delta(2^{k-1}+d+1)]+1-(\rho_1-1))& \text{ by induction}\\
         &=  \frac{1}{2}[2\cdot\delta(2^{k-1}+d-1)+2\cdot\delta(2^{k-1}+d+1)]+2-(2\cdot\rho_1-2) \\
         &=  \frac{1}{2}[\delta(2^k+2d-2)+\delta(2^k+2d+2)]+2-(2\cdot\rho_1-2)& \text{ by Theorem \ref{exp_D_recursive}}\\
         &=  \frac{1}{2}(\delta(n-2)+\delta(n+2))+4-2\cdot\rho_1 \\
         &=  \delta(n-1)+\delta(n+1)-2-\delta(n)+4-2\cdot\rho_1& \text{  by Equation \ref{n_plusminus_2}}\\
         &=  \delta(n-1)+\delta(n+1)-\delta(n)+2-2\cdot\rho_1 
    \end{flalign*}
So
    $$
        2\cdot \delta(n) 
        = \delta(n-1)+\delta(n+1)+2-2\cdot\rho_1
    $$
    and
$$
    \delta(n) 
        = \frac{1}{2}(\delta(n-1)+\delta(n+1))+1-\rho_1
$$
\end{proof}
\begin{theorem}\label{exp_D_closed}
	Let the binary decomposition of $n$ be $n_k \dots n_1 n_0$, where $n_k=1$. Then the number of $D$-nodes in a divide-and-conquer SD-tree with $n$ leaf nodes  is
	$$
			\delta(n)=\sum_{i=0}^{\floor{\log_2(n)}-1} \lambda_i(n) \text{, where } \lambda_i(n)= 
			\begin{cases}
			    (n\bmod 2^i), & \text{ if } n_i=0;\\
			    2^i-(n\bmod 2^i), & \text{ if } n_i=1.\\
	        \end{cases}	  
    $$ 
An explicit form is
	$$
		\delta(n)=\sum_{i=0}^{\floor{\log_2(n)}-1}  n_i\cdot 2^i + (-1)^{n_i}\cdot (n \bmod 2^i)
	$$
\end{theorem}
\begin{proof}
    Follows from Theorem \ref{exp_closed}, the fact that there are $2^i$ interior nodes at the $i^\text{th}$ level of a divide-and-conquer form and the fact that all nodes at the $\floor{\log_2(n)}$ level are either leaves or $S$-nodes.
\end{proof}
\begin{lemma} \label{symmetricity}
    $\delta(2^{k+1}-r) = \delta(2^k+r)$, when $0 \le r \le 2^k$.
\end{lemma}
\begin{proof}
    Follows from Theorem \ref{exp_D_recursive} using induction.
\end{proof}
\begin{lemma} \label{lambdadiff}
    Let $n=2^k+r$, with $r$ odd and $0 < r < 2^k$, and let $n$ have as its binary expansion $n_k n_{k-1} \dots n_0$. Let $0 < i<k$. Then
    $$
        \lambda_i(n)-\lambda_i(n-1)  =
        \begin{cases}
 			    1, & \text{ if }\ n_i=0;\\
			    -1, & \text{ if }\ n_i=1.\\
        \end{cases}
    $$
\end{lemma}
\begin{proof}
    For $0<i<k$, all bits of $r$ and $(r-1)$ besides the $0^\text{th}$ are the same, since $r$ is odd. Thus,
    $$
        \lambda_i(n)-\lambda_i(n-1)  =
        \begin{cases}
 			    (n \bmod 2^i)-((n-1) \bmod 2^i)=1, & \text{ if }\ n_i=0;\\
			    (2^i-(n\bmod 2^i))-(2^i-((n-1)\bmod 2^i))=-1, & \text{ if }\ n_i=1.\\
        \end{cases}
    $$
\end{proof}
\begin{lemma} \label{dnodes_diff_even}
    Let $n>0$ be even. Then
    $$
        \delta(n+1) = \delta(n) + (\floor{\log_2(n)}-2\cdot\omega(n)+2)
    $$
\end{lemma}
\begin{proof}
    Let $n=2^k+r$, with $0\le r <  2^k$, and let $n$ have as its binary expansion $n_k n_{k-1} \dots n_0$. $n$ is even, so $\floor{\log_2(n)}=\floor{\log_2(n+1)}$.
    \begin{flalign*}
        \delta(n+1)-\delta(n)
        &= \sum_{i=0}^{k-1} \lambda_i(n+1)-\sum_{i=0}^{k-1} \lambda_i(n)& \text{ by Theorem \ref{exp_D_closed}}&\\
        &= \sum_{i=0}^{k-1} (\lambda_i(n+1)- \lambda_i(n))\\
        &= (\lambda_0(n+1)- \lambda_0(n)) + \sum_{i=1}^{k-1} (\lambda_i(n+1)- \lambda_i(n))\\
    \end{flalign*}
    $n_0=0$, so $\lambda_0(n+1)-\lambda_0(n)=(2^0-0)-0=1$.
    \newline
    There are $\omega(n)-1=(\omega(n+1)-2)$ $1$s  in the $1^\text{st}$ through $(k-1)^\text{th}$ bits of $(n+1)$.
    \newline
    Thus there are $((k-1)-(\omega(n+1)-2))$ $0$s in $1^\text{st}$ through $(k-1)^\text{th}$ bits of $(n+1)$. So 
    \begin{flalign*}
        \delta(n+1)-\delta(n)
        &= 1 + -1\cdot [\omega(n+1)-2] + 1\cdot  [(k-1)-(\omega(n+1)-2)]&\text{ by Lemma \ref{lambdadiff}}&\\
        &= k-2\cdot\omega(n+1)+4 \\
        &= k-2\cdot(\omega(n)+1)+4&\text{ by Lemma \ref{weightlemma_even}}&\\
        &= k-2\cdot\omega(n)+2\\
        &= \floor{\log_2(n)}-2\cdot\omega(n)+2
    \end{flalign*}
\end{proof}
\begin{lemma} \label{dnodes_diff_odd}
    Let $n$ be odd. Then
    $$
        \delta(n+1) = \delta(n)+(\floor{\log_2(n)}-2\cdot\omega(n)+2)
    $$
    where $\delta(1)=0$.
\end{lemma}
\begin{proof}
    Let $n=2^k+r$, with $0\le r <  2^k$. We proceed by moving from $n$ and $n+1$ to the symmetric $m=2^{k+1}-(r+1)$ (which is even) and $m+1=2^{k+1}-r$ (which is odd), by Lemma \ref{symmetricity}. We then apply Lemma \ref{dnodes_diff_even}, and move back to $n$ and $n+1$, again by Lemma \ref{symmetricity}. 
    
    \begin{flalign*}
        \delta(n+1) 
        &= \delta(m)  &\text{  by Lemma \ref{symmetricity}}&\\
        &= \delta(m+1)-(\floor{\log_2(m)}-2\cdot\omega(m)+2) &\text{ by Lemma \ref{dnodes_diff_even}}\\
        &= \delta(m+1)-(k-2\cdot\omega(m)+2) \\
        &= \delta(m+1)-(k-2\cdot(\omega(m+1)-1)+2)&\text{ by Lemma \ref{weightlemma_even}}\\
        &= \delta(2^k+r)-(k-2\cdot\omega(m+1)+4)&\text{ by Lemma \ref{symmetricity}}\\
        &= \delta(n)-(k-2\cdot\omega(m+1)+4)\\
        &= \delta(n)-(k-2\cdot(2+k-\omega(r))+4)&\text{ by Lemma \ref{weightlemma}}\\
        &= \delta(n)-(k-2\cdot(2+k-(\omega(n)-1))+4)&\text{ by Lemma \ref{weightlemma_addone}}\\
        &= \delta(n)+(k-2\cdot\omega(n)+2)\\
        &= \delta(n)+(\floor{\log_2(n)}-2\cdot\omega(n)+2)
    \end{flalign*}
\end{proof}
We then have a general recurrence relation for $\delta(n)$.
\begin{theorem} \label{dnodes_another_recurrence}
    Let $n>0$. Then
    $$
        \delta(n+1) 
        = \delta(n) + (\floor{\log_2(n)}-2\cdot\omega(n)+2)
    $$
    where $\delta(1)=0$.
\end{theorem}
\begin{proof}
    Follows from Lemmas \ref{dnodes_diff_even} and \ref{dnodes_diff_odd}.
\end{proof}
Repeated application of Theorem \ref{exp_D_recursive} (for even sub-sums) and Lemma \ref{dnodes_diff_even} (for odd) gives rise to Theorem \ref{another_explicit_D}, another explicit formula for the number of $D$-nodes in a divide-and-conquer tree.
\begin{theorem} \label{another_explicit_D}
    Let $n>0$ have as its binary expansion $n_k n_{k-1} \dots n_0$, where $n_k=1$ (so $k=\floor{\log_2(n)}$). Then
    $$
        \delta(n) = \sum_{i=0}^{k-1}2^i \cdot n_i \cdot \left[(k-i)-2\cdot \sum_{j=i}^k n_i+4\right]
    $$
    or alternatively
    $$
        \delta(n) = \sum_{i=0}^{k-1}2^i \cdot n_i \cdot \left[(k-i)-2\cdot \omega\left(\floor{\frac{n}{2^i}}\right)+4\right]
    $$
\end{theorem}
\begin{proof}
    Proof by induction, noting that this is true for $n=1$ and 
    \newline
    If $n$ is even, $n/2$ has as its binary expansion $n_k n_{k-1} \dots n_1$, so 
    \begin{flalign*}
        \delta(n/2)
        &=\sum_{i=1}^{k-1}2^{i-1} \cdot n_i \cdot \left[(k-i)-2\cdot \sum_{j=i}^k n_i+4\right]& \text{ by induction}&\\
        &=\frac{1}{2}\sum_{i=1}^{k-1}2^i \cdot n_i \cdot \left[(k-i)-2\cdot \sum_{j=i}^k n_i+4\right]
    \end{flalign*}
    \begin{flalign*}
        \delta(n)
        &=2\delta(n/2) \text{, by Theorem \ref{exp_D_recursive}}&\\
        &=2\cdot \frac{1}{2}\sum_{i=1}^{k-1}2^i \cdot n_i \cdot \left[(k-i)-2\cdot \sum_{j=i}^k n_i+4\right] \\
        &=\sum_{i=0}^{k-1}2^i \cdot n_i \cdot \left[(k-i)-2\cdot \sum_{j=i}^k n_i+4\right]& \text{ since $n_0=0$}&
    \end{flalign*}
    \newline
    If $n$ is odd, $n$ has as its binary expansion $n_k n_{k-1} \dots n_1 1$, and $n-1$ has as its binary expansion $n_k n_{k-1} \dots n_1 0$.
    \begin{flalign*}
        \delta(n)
        &=\delta(n-1) + (\floor{\log_2(n-1)}-2\cdot\omega(n-1)+2)& \text{ by Theorem \ref{dnodes_another_recurrence}}&\\
        &=\delta(n-1) + (k-2\cdot\omega(n-1)+2) \\
        &=\sum_{i=0}^{k-1}2^i \cdot n_i \cdot \left[(k-i)-2\cdot \sum_{j=i}^k n_i+4\right] + (k-2\cdot\omega(n-1)+2)& \text{ by induction}\\
        &=\sum_{i=1}^{k-1}2^i \cdot n_i \cdot \left[(k-i)-2\cdot \sum_{j=i}^k n_i+4\right] + (k-2\cdot\omega(n-1)+2)\\
        &=\sum_{i=1}^{k-1}2^i \cdot n_i \cdot \left[(k-i)-2\cdot \sum_{j=i}^k n_i+4\right] + (k-2\cdot(\omega(n)-1)+2)& \text{ by Lemma \ref{weightlemma_odd}}\\
        &=\sum_{i=1}^{k-1}2^i \cdot n_i \cdot \left[(k-i)-2\cdot \sum_{j=i}^k n_i+4\right] + (k-2\cdot\omega(n)+4)\\
        &=\sum_{i=1}^{k-1}2^i \cdot n_i \cdot \left[(k-i)-2\cdot \sum_{j=i}^k n_i+4\right] + 2^0 \cdot 1  \cdot  ((k-0)-2\cdot\omega(n)+4) \\
        &=\sum_{i=1}^{k-1}2^i \cdot n_i \cdot \left[(k-i)-2\cdot \sum_{j=i}^k n_i+4\right] + 2^0 \cdot n_0  \cdot  ((k-0)-2\cdot\omega(n)+4) \\
        &=\sum_{i=0}^{k-1}2^i \cdot n_i \cdot \left[(k-i)-2\cdot \sum_{j=i}^k n_i+4\right]
    \end{flalign*}
\end{proof}
\subsubsection{Divide-and-conquer D-nodes and the Takagi function} \label{takagi_section} 
The Takagi function is a widely-studied self-similar nowhere-differentiable function on $[0,1]$, identified by Takagi in 1901 \cite{takagi01}, with connections to many areas of mathematics, including number theory, combinatorics, probability theory and analysis. Lagarias published a survey on the Takagi function \cite{lagarias11}, as have Allaart and Kawamura \cite{allaart11}.  

The Takagi function is closely related to the number of $D$-nodes in a divide-and-conquer tree. OEIS \seqnum{A268289}, which is the number of $S$-nodes in a divide-and-conquer parenthetic form, is related to the Takagi function through its relationship to  the number of $D$-nodes in a divide-and-conquer product \cite{ lagarias11, allaart11, baruchel19}.
Fig. \ref{takagicurve} compares a graph of the Takagi function to graphs counting $D$-nodes.

\begin{definition} \label{takagi_newdef} \cite{takagi01, lagarias11}
    Let an integer $r$ have as its binary expansion $r_{k-1} r_{k-2} \dots r_0$, where all bits $r_i=0 \text{, for }  i \notin \{0, \dots k-1\}$. The \emph{Takagi function} on dyadic rationals is defined as
	$$
        \tau\left(\frac{r}{2^k}\right) 
        = \sum_{i=1}^\infty\frac{\ell_i(r)}{2^i} 
        \text{, where } 
        \ell_{i+1}(r)=
        \begin{cases}
		    \sum\limits_{j=1}^i r_{k-j}, & \text{ if } r_{k-(i+1)} = 0,\\
		    i-\sum\limits_{j=1}^i r_{k-j}, & \text{ if } r_{k-(i+1)} = 1.\\
	  \end{cases}
	$$
The original definition of the Takagi function on dyadic rationals from \cite{takagi01} (quoted in \cite{lagarias11}), has been recast here to be consistent with our notation and numbering.
\end{definition}

The following Theorem \ref{takagi2} gives a simple formula for the dilations of the Takagi function on dyadic rationals in [0,1], in terms of the $D$-nodes of a divide-and-conquer tree. This is illustrated in Fig. \ref{takagicurve}.
Theorem \ref{takagi2} is essentially the same as Corollary 4 in \cite{coronado20} and to 3.3 in \cite{baruchel19_2}, but is stated here in terms of $D$-nodes and Takagi's original definition. It is also quickly derivable from Equation 4.7 in \cite{allaart11}, credited to Kr\"uppel \cite{kruppel07}, which seems to be the original source.

\begin{theorem} \label{takagi2} \cite{coronado20, allaart11,  baruchel19, kruppel07}
    Let $\delta(n)$ be the number of $D$-nodes in a divide-and-conquer tree on $n=2^k+r$ leaves, with $0 \le r \le 2^k$ and with the binary expansion of $r$ $r_{k-1} r_{k-2} \dots r_0$. Let $\tau(x)$ be the Takagi function on $x \text{, with }0 \le x \le 1$. Then 
    $$
        \tau \left( \frac{r}{2^k} \right) = \frac{\delta(2^k+r)}{2^k}
    $$
\end{theorem}
\begin{figure}[H]
    \centering
    \begin{subfigure}{0.2\textwidth} 
        \includegraphics[scale=0.2]{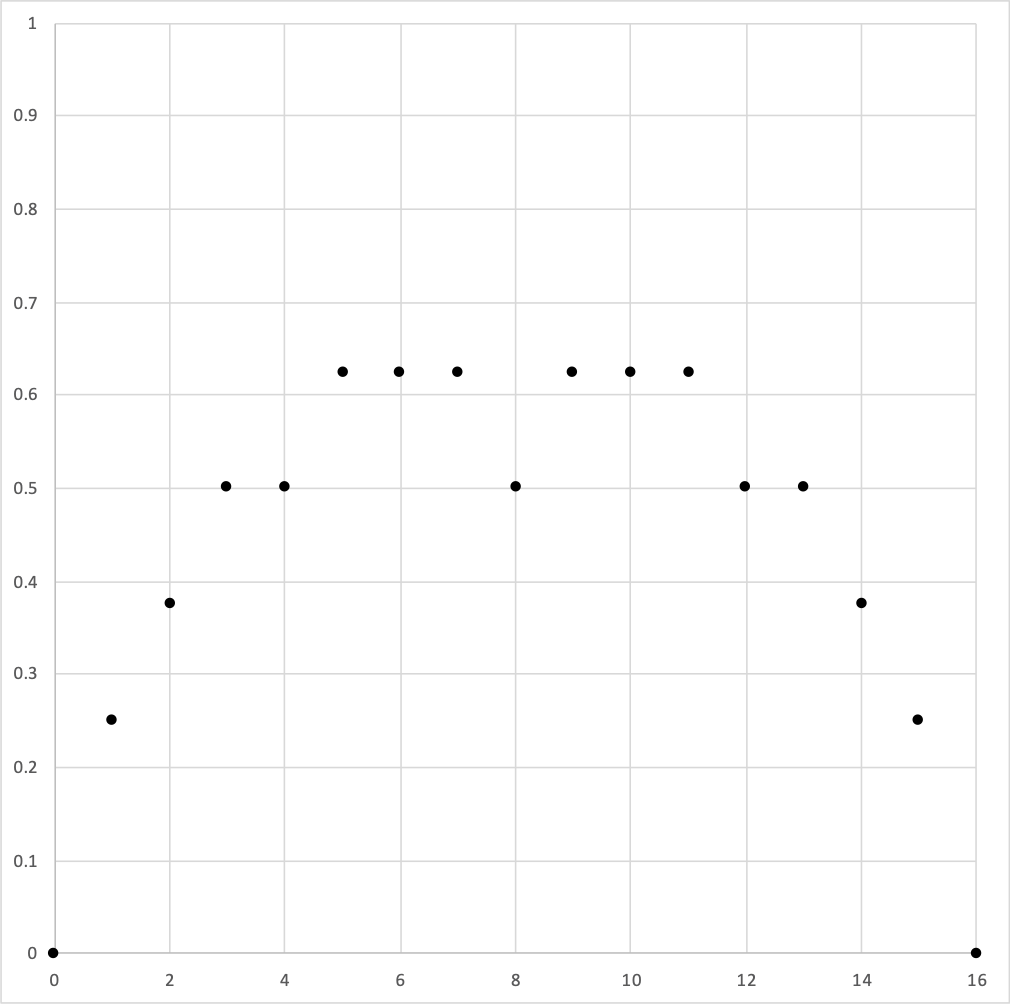}
        \caption{$y=\frac{\delta(16+x)}{16}$}
        \label{delta16}
    \end{subfigure}
    \hspace{1em}
    \begin{subfigure}{0.2\textwidth}
        \includegraphics[scale=0.2]{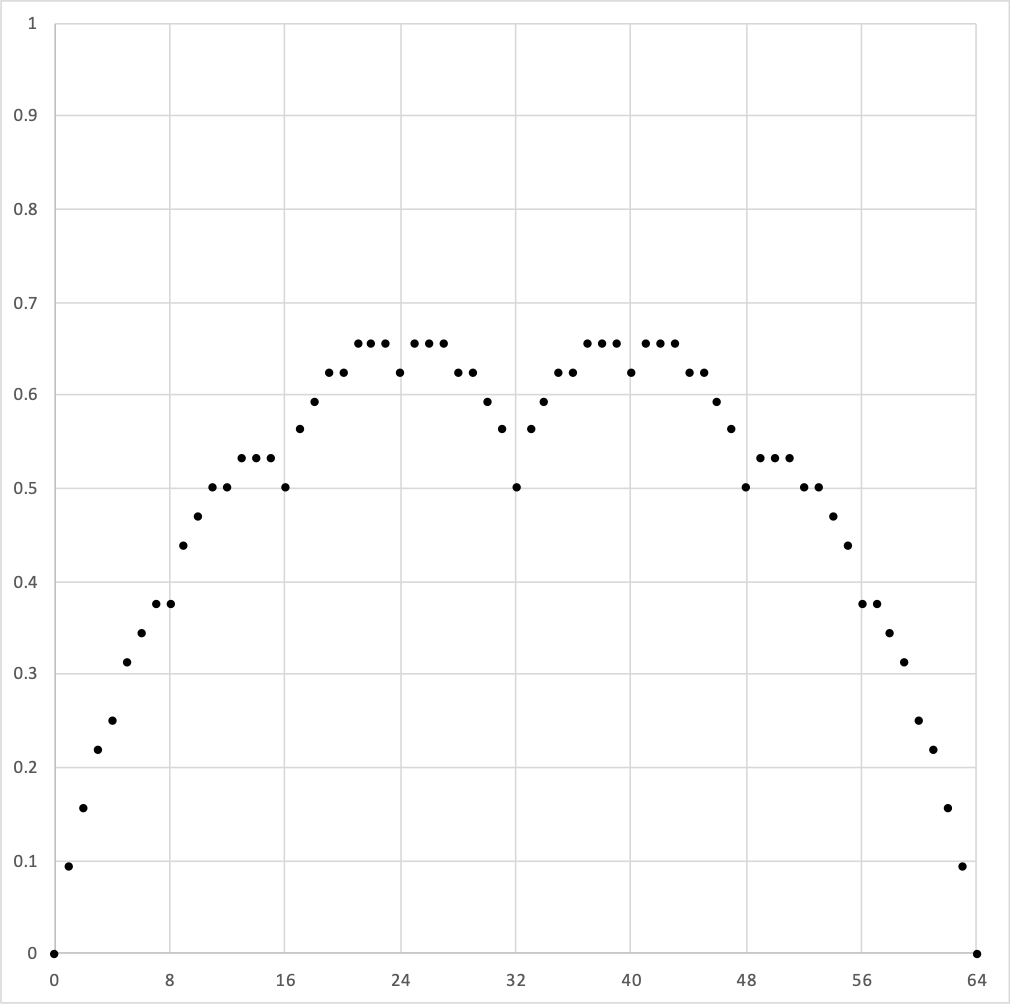}
        \caption{$y=\frac{\delta(64+x)}{64}$}
        \label{delta64}
    \end{subfigure}
    \hspace{1em}
    \begin{subfigure}{0.2\textwidth}
        \includegraphics[scale=0.2]{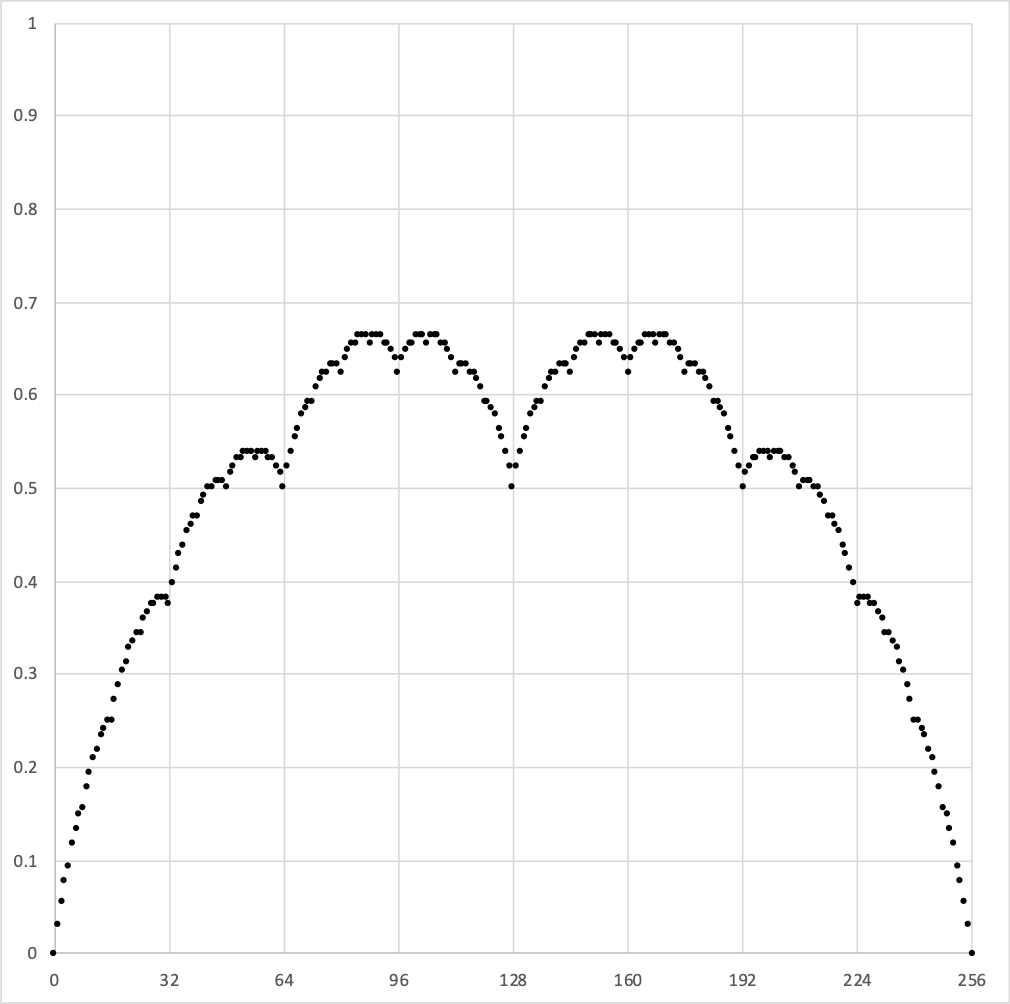}
        \caption{$y=\frac{\delta(256+x)}{256}$}
        \label{delta256}
    \end{subfigure}
    \hspace{1em}
    \begin{subfigure}{0.2\textwidth}
        \includegraphics[scale=0.33]{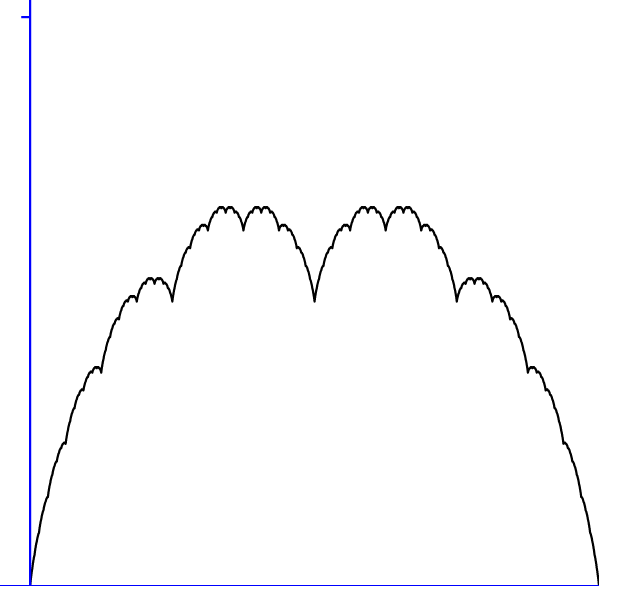}
        \caption{Takagi curve.}
        \label{takagifig}
    \end{subfigure}
    \mycaption{Divide-and-conquer dilations of the Takagi function on the dyadic rationals}
    {Subfigures (\subref{delta16}), (\subref{delta64}) and (\subref{delta256}) show examples of the dilations $y=\frac{\delta(2^k+x)}{2^k}=\tau\left(\frac{x}{2^k}\right)$  from Theorem \ref{takagi2}, where $\delta(n)$ is the number of $D$-nodes on a divide-and-conquer tree on $n$ leaves. Here, $k=4,6,\text{ and }8$, and $x$ is an integer with $0 \le x \le 2^k$. These may be visually compared to subfigure (\subref{takagifig}), showing
    the continuous, self-similar, nowhere-differentiable Takagi (blancmange) curve on $[0,1]$. (The blancmange curve image  in subfigure (\subref{takagifig}) is taken from Wiki Commons.)}
    \label{takagicurve}
\end{figure}

Since Theorem~\ref{takagi2} so closely relates the Takagi function to the number of $D$-nodes in a divide-and-conquer tree, all theorems shown earlier in this paper that count $D$-nodes of such trees give new identities on the Takagi function. These theorems include Corollary \ref{yet_another_dnodes_prop} and Theorems \ref{exp_D_closed}, \ref{dnodes_another_recurrence}, and \ref{another_explicit_D}. We also provide here in Theorem~\ref{another_takagiX} a new formula for the Takagi function. 

\begin{theorem} \label{another_takagiX}
Let $\ell_{i+1}$ be as in Definition \ref{takagi_newdef}, and let $\omega(r)$ be the weight of $r$. Then
    $$
        2^k \cdot \tau\left(\frac{r}{2^k}\right) 
        = \sum_{i=1}^k2^{k-i}\ell_i(r) + \omega(r) 
	$$
\end{theorem}
\begin{proof}
    \begin{flalign*}
        2^k \cdot \tau(\frac{r}{2^k})
        &= 2^k \cdot\sum_{i=1}^\infty\frac{\ell_i(r)}{2^i}&
        \text{ by Definition \ref{takagi_newdef}}&\\
        &= \sum_{i=1}^k 2^{k-i}\ell_i(r) + \sum_{i=k+1}^\infty\frac{\ell_{i}(r)}{2^{i-k}}\\
        &= \sum_{i=1}^k 2^{k-i}\ell_i(r) + \sum_{i=k+1}^\infty\frac{\omega(r)}{2^{i-k}}&
        \text{ by Definition \ref{takagi_newdef}}\\
        &= \sum_{i=1}^k 2^{k-i}\ell_i(r) + \omega(r)\cdot \sum_{i=1}^\infty\frac{1}{2^i}\\
        &= \sum_{i=1}^k 2^{k-i}\ell_i(r) + \omega(r)\\
    \end{flalign*}
\end{proof}

\subsubsection{Counting divide-and-conquer  products}

We complete this section by counting the commutative non-associative divide-and-conquer products, using the results from Section \ref{dandc_snodes} on the number of $S$-nodes in such products. The application of Corollary \ref{equiv_sum} is simple; however, this result is worth stating explicitly in the context of this paper, since the motivating example of finite-precision floating-point summation is sensitive to ordering of variables \cite{job20}, making the different products relevant. 

\begin{proposition}\label{pairclosednew}
	The number of computationally inequivalent divide-and-conquer  products on $n$ terms is 
	$$
	\rho(n)=\frac{n!}{2^{\sigma(n)}}
	$$
	where 
	$\sigma(n)$ is one of the forms in Theorems \ref{exp_recursive} or \ref{exp_closed}.
\end{proposition}
\begin{proof}
	Follows directly from Corollary \ref{equiv_sum} and Theorems \ref{exp_recursive} and \ref{exp_closed}.
\end{proof}
Counting the divide-and-conquer  products on $n$ elements is equivalent to counting the tournaments on $n$ teams. This is the classical formulation of this problem, and gives rise to the sequence OEIS \seqnum{A096351} \cite{OEIS}.
Proposition \ref{pairclosednew} gives new recursive and closed formulas for OEIS \seqnum{A096351}.
For completeness and comparison, we present Proposition \ref{pairrecursiveold}, also giving a formula for this sequence, which was shown by David in \cite{david88}. 

\begin{proposition}\label{pairrecursiveold} \cite{david88}
	The number of computationally inequivalent divide-and-conquer products on $n$ variables is 
	$$
	\rho(n)=
	\begin{cases}
	\frac{1}{2}\binom{2m} {m}\rho(m)^2, & \text{if}\ n=2m; \\
	\binom{2m+1}{m}\rho(m)\rho(m+1), & \text{if}\ n=2m+1.
	\end{cases}
	$$ 
\end{proposition}

\section {MinD trees: trees having a minimum number of D-nodes} \label{max_S_section}
Trees having a minimal number of $D$-nodes are of interest when considering computational balance. In this section, we constructively characterize such trees. This construction is really quite simple. 

It is convenient to do this by considering $S$-nodes, using the fact that the trees on $n$ leaves with minimal $D$-nodes are exactly the trees on $n$ leaves with maximal $S$-nodes. 
It follows from Corollary \ref{equiv_sum} that $\sigma$, the number of $S$-nodes in a form on $n$ terms, must be less than or equal to the maximum power of $2$ that divides $n$. We use this fact to deduce a maximum number of $S$-nodes on a parenthetic form, to show that there is a parenthetic form that meets this maximum, and to list all forms that meet the maximum. 
\begin{lemma}\label{pow2_2k}
	The largest power of $2$ that divides $2^k!$ is  $2^{(2^k-1)}$.
\end{lemma}
\begin{proof}
	Let $2^{\beta(2^k)}$ be the largest power of $2$ that divides $2^k!$. For  $1\le i \le k$ there are $\frac{2^k}{2^i}=2^{k-i}$ numbers in $\{1,2,\dots,2^k\}$ divisible by $2^i$. Each of these adds one to  $\beta(2^k)$. So $\beta(2^k)=\sum\limits_{i=1}^{k}2^{k-i}=\sum\limits_{i=0}^{k-1}2^{i}=2^k-1$.
\end{proof}
\begin{lemma}\label{pow2_m}
	Let $n=2^k+r$, where $0  < r <2^k$. 
	The largest power of $2$ that divides $\frac{n!}{(2^k)!}=n(n-1)\cdots (2^k+1)$ is the same as the largest power of $2$ that divides $(n-2^k)!=r!$.
\end{lemma}
\begin{proof}
	$2^k$ does  not divide any of $\{2^k+1, 2^k+2, \dots ,2^k+r=n\}$, so largest power of of $2$ that divides $\frac{n!}{(2^k)!}=n(n-1)\cdots (2^k+1)$ is the same as the largest power of $2$ that divides $(n-2^k)!$.
\end{proof}	
\begin{lemma}\label{pow2_2km}
	Let $n=2^k+r$, where $0  < r <2^k$. The largest power of $2$ that divides $n!$ is the largest power of $2$ that divides $2^k!$ plus the largest power of $2$ that divides $(n-2^k)!=r!$
\end{lemma}
\begin{proof}
	Follows from Lemmas \ref{pow2_2k}  and \ref{pow2_m}.
\end{proof}	
\begin{lemma} \label{pfb}
	A full binary tree on $2^k$ terms has the maximal number $2^k-1$ $S$-nodes if and only if it is perfect.
\end{lemma}
\begin{proof}
	A binary tree has $2^{k}$ leaf nodes and $2^k - 1$  interior nodes. It is perfect if and only if each of the interior nodes is an $S$-node \cite{zou19}. So it is perfect if and only if has the maximum number of $S$-nodes.
\end{proof}

\begin{notation}
    Denote by $T_{max}(i)$ a parenthetic form or tree on $i$ terms having a maximal number of $S$-nodes. 
\end{notation}
\begin{notation}
    Denote by $T_L \odot T_R$ the tree formed by joining to a root node a left subtree $T_L$ and a right subtree $T_R$. 
\end{notation}
\begin{proposition}\label{lowerboundmet}
	Let $n=2^k+r$ , where $0 \le r < 2^k$. 	Then $T_{max}(2^k)\odot T_{max}(r)$
	has the maximal number of $S$-nodes possible on a form on $n$ terms.
\end{proposition}
\begin{proof}
    If $r=0$, $T_{max}(n)$ is the perfect full binary tree on $2^k$ nodes and has the maximal number of $S$-nodes, by Lemma \ref{pfb}.

    If $r>0$, the top node of $T_{max}(2^k)\odot T_{max}(r)$ is not an $S$-node  and the problem reduces to finding the largest powers of 2 that divide $2^k$ and $r$. The conclusion follows by Lemma \ref{pow2_2km}.
\end{proof}
\begin{corollary}\label{max_s_nodes}
    Let $n=2^k+r$ , where $0 \le r < 2^k$. The
	maximal number of $S$-nodes possible on a tree or form on $n$ terms is $n-\omega(n)$, and $T_{max}$ has $n-\omega(n)$ $S$-nodes.
\end{corollary}
\begin{proof}
    If $r=0$, $T_{max}(n)$ is the perfect tree with $n$ leaves, by Lemma \ref{pow2_m}, and has $2^k-1$ $S$-nodes, and we are done.
    
    Let $r>0$. By Proposition \ref{lowerboundmet},
    the maximal number of $S$-nodes in $T_{max}(n)$ is the number of $S$-nodes in $(T_{max}(2^k)\odot T_{max}(r))$. 
    
    We proceed via induction on the number of leaves.
    $r \ne 2^k$, so the root is not an $S$-node and the number of $S$-nodes in $(T_{max}(2^k)\odot T_{max}(r))$ is the sum of the number of $S$-nodes in $(T_{max}(2^k)$ and $(T_{max}(r))$.
    $T_{max}(2^k)$ has $2^k-1$ $S$-nodes, by Lemma \ref{pow2_2k}, and $T_{max}(r)$ has $r-\omega(r)$  $S$-nodes, by induction. So the number of $S$-nodes in $(T_{max}(2^k)\odot T_{max}(r))$ is
    \begin{flalign*}
        2^k-1+(r-\omega(r))
        &=(2^k+r)-(\omega(r)+1)\\
        &=n-\omega(n)& \text{ by Lemma \ref{weightlemma_addone}}&
    \end{flalign*} 
  \end{proof}
	The sequence $\sigma(n)$, the maximum number of $S$-nodes in a full binary tree with $n$ leaf nodes, is OEIS \seqnum{A011371}
	\cite{OEIS}, the exponent of the highest power of $2$ dividing $n!$.
	
	The following is a construction theorem on full binary trees with a maximum number of $S$-nodes (i.e., a minimum number of $D$-nodes). All trees with a minimum number of $D$-nodes are of this form. This theorem was independently shown in the recent \cite{kersting21}.
\begin{theorem} \label{max_tree_s}
    Let $n$ have binary expansion $(n_{k-1}\text{ }n_{k-2} \dots n_0)$ with weight $\omega(n)$, and let $R(n)=\{\rho \mid n_\rho = 1\}$. Let $T_{\omega(n)}$ be a tree with $\omega(n)$ leaves, with its leaves replaced by the $|\omega(n)|$ perfect trees on $2^{n_\rho}$ leaves, where $\rho$ ranges over $R(n)$. 
    Then
    a tree $T$ with $n$ leaves has the maximum number of $S$-nodes, and hence the minumal number of $D$-nodes, if and only if $T$ is of the form of $T_{\omega(n)}$.
\end{theorem}
\begin{proof}
    Let $T$ be of the form of $T_{\omega(n)}$.
    Each of its $R(n)$ perfect subtrees has $2^{n_\rho}-1$ interior $S$-nodes, where $\rho \in R(n)$, and none of the nodes exterior to the perfect subtrees are $S$-nodes. So the number of $S$-nodes is 
    $$\sum_{\rho \in R(n)} (2^{n_\rho}-1) = \sum_{\rho \in R(n)} 2^{n_\rho}-\omega(n) = n-\omega(n)$$
    which is the maximum possible, by Corollary \ref{max_s_nodes}.
        
    Now, let $T$ on $n$ leaves have the maximum $(n-\omega(n))$ $S$-nodes. Let its left subtree have $\ell$ leaves and its right have $r$ leaves, so $\ell+r=n$, and proceed by induction to show that $T$ has form $T_{\omega(n)}$.
        
    If $\ell=r$, then each  subtree has the same number of leaves,  so  $n$ is even, $\ell=r=\frac{n}{2}$, and    $\omega(\ell)=\omega(r)=\omega(n)$.
    The number of $S$-nodes is $1+(\ell-\omega(\ell))+(r-\omega(r))=n-(2\omega(n)-1)$, since the root is an $S$-node, and by the induction hypothesis, since the ``leaves'' of the subtrees  are perfect trees. This means that $2\omega(n)-1 = \omega(n)$, so $\omega(n)=1$, $n$ is a power of 2, and by Lemma \ref{pfb}, $T$ is perfect and has the form of $T_{\omega(n)}$.
        
    If $\ell \ne r$, then the root is not an $S$-node and the total number of $S$-nodes is $(\ell-\omega(\ell))+(r-\omega(r))=n-(\omega(\ell)+\omega(r))$. So $(\omega(\ell)+\omega(r))=\omega(n)$. By Lemma \ref{weightlemma_addtwo}, there is no index $i$ where $\ell_i$ and $r_i$ are both $1$. So the sets of left and right perfect trees are disjoint, and by induction, $T$ has the form of $T_{\omega(n)}$.
\end{proof}
\begin{definition}
    A \emph{MinD tree} on $n$ leaves is a binary full tree with a minimal number of $D$-nodes for that $n$.
\end{definition}
\begin{definition}
    The \emph{base tree} on a  MinD tree $T$ on $n$ leaves is the leaf-unlabeled tree on $\omega(n)$ leaves whose interior nodes are the $D$-nodes of $T$.
\end{definition}
\begin{corollary} \label{min_dnodes}
    A MinD tree on $n$ leaves has $(\omega(n)-1)$ $D$-nodes.
\end{corollary}
\begin{proof}
    Such a tree is of the form described in Theorem \ref{max_tree_s}. It has $(n-\omega(n))$ $S$-nodes, by Corollary \ref{max_s_nodes}, so has $(n-1-(n-\omega(n))=\omega(n)-1$ $D$-nodes, by Lemma \ref{dnodes}.
\end{proof}
\begin{corollary} \label{nummax}
    There are exactly $(2\cdot \omega(n)-3)!!$ MinD trees on $n$ leaves.
\end{corollary}
\begin{proof}
    MinD trees are completely characterized in Theorem \ref{max_tree_s}. Since these trees are formed essentially by labeling the leaves of a leaf-unlabeled tree with different (leaf-unlabeled) perfect trees, counting them is equivalent to counting inequivalent products with $\omega(n)$ terms. There are $(2\cdot \omega(n)-3)!!$ of these, by Proposition \ref{ineq_sum}.
\end{proof}
The parenthetic form characterized in Theorem \ref{max_tree_s} meets the upper bound for $S$-nodes discussed in Section \ref{lu_bounds}. The number of products instantiating this parenthetic form meets the lower bound for the number of inequivalent commutative non-associative products on $n$ variables. 
    
A complete example of the $15$ MinD trees on $27$ leaves is given in Fig. \ref{fig:minD_27}. This figure illustrates the process of constructing trees with a minimum number of $D$-nodes.

\section{D-nodes and measures of tree balance} \label{colless_dnodes}
The Colless index, defined in 1980 \cite{colless80}, is a measure of the balance of a full binary tree. It is a widely used measure in the field of phylogenics \cite{kirkpatrick93, mooers97, coronado20}.

The Colless index is calculated by assigning a value to each interior node expressing the difference in size of its two subtrees and adding these: if $\ell_L$ and $\ell_R$ are the number of leaves in the left and right subtrees of an interior node, the node's value is $\lvert \ell_L - \ell_R\rvert$, and the Colless index is the sum of these values. This gives some sense of the imbalance of a tree in terms of the number of leaves in each branch.

The relationship of $S$- and $D$-nodes to the Colless index is clear. The $S$-nodes have value 0. The $D$-nodes can be used to calculate the Colless index. The smallest value that a $D$-node can take is $1$, and indeed, all the $D$-nodes in divide-and-conquer trees have value $1$. The largest value a $D$-node can take is $(n-2)$.

\begin{definition} \label{def_colless}
    Let $D$ be the set of $D$-nodes on the tree $T$, and let $t(i) \in D$. Let the right subtree  of $t(i)$ have $\ell_t(i)$ leaves and the left subtree of $t(i)$ have $r_t(i)$ leaves. The \emph{Colless index} is defined as 
    $$
        \sum\limits_{t(i)\in  D}\lvert \ell_t(i)-r_t(i)\rvert
    $$
\end{definition}
Since  by definition, the two children of an $S$-node have the  same number of leaves, the $S$-nodes contribute nothing to the Colless index, and we  do not consider them.

\subsection{Trees with maximal Colless index: ladder trees} \label{subsec_colless_ladder}
Each node in a ladder  tree on $n$ leaves has one child a leaf, and the other a smaller ladder tree on $(n-1)$ leaves, with the exception of the lowest interior node, which has as children two leaves. Thus, all the interior nodes are $D$-nodes, except for the lowest one. The following proposition has been observed many times by Heard, Colless and others \cite{heard92, colless95, coronado20}. Ladder trees were shown by Mir et. al \cite{mir18} to have maximal Colless index. 
\begin{proposition} \label{colless_ladder} \cite{colless95}
    The Colless index of a ladder tree $T$ having $n$ leaves is
    $$
        \frac{(n-1)(n-2)}{2}
    $$
\end{proposition}
\begin{proof}
    This is seen from the fact that there are $(n-1)$ interior nodes and the $i^\text{th}$ interior node contributes $((n-i)-1)$ to the index. This can also be shown with a quick induction, since one of the children of an interior node with $j$ leaves is a ladder tree with $(j-1)$ leaves, and the other is a leaf.
\end{proof}
\begin{proposition}\cite{mir18}
    The Colless index on a ladder tree is maximal.
\end{proposition}
\subsection{Two types of tree with minimal Colless index}
Coronado et. al \cite{coronado20} showed that divide-and-conquer trees have minimal Colless index. They also showed that divide-and-conquer trees are not unique in having minimal index, and characterized all such binary full trees.
\subsubsection{The Colless index on divide-and-conquer trees} \label{colless_dandc}
The number of $D$-nodes in a divide-and-conquer tree, discussed in Section \ref{pairwise_cnaos}, is closely related to the Colless index on such a tree. 

\begin{proposition}\cite{coronado20} \label{minimumcolless}
    The Colless index on a divide-and-conquer tree on $n=2^k+r$ leaves is minimal for trees on  $n$ leaves.
\end{proposition}
\begin{proposition}
    The Colless index on a divide-and-conquer tree is the number of $D$-nodes in that tree.
\end{proposition}
\begin{proof}
    Each $D$-node in a divide-and-conquer tree has right and left subtrees such that the number of leaf nodes of each differ by $1$. Thus, the Colless index of such a tree is the number of $D$-nodes in the tree.
\end{proof}

\begin{theorem} \cite{coronado20}
    The Colless index on a divide-and-conquer tree with $n=2^k+r$ leaves is $2^k \cdot \tau\left(\frac{r}{2^k}\right)$, where $\tau$ is the Takagi function. 
\end{theorem}

\subsubsection{The Colless index on complete full binary trees} \label{colless_cfb}
\begin{theorem}
    The Colless index on a complete full binary tree with $n=2^k+r$ leaves is minimal, hence is $\delta(n)$, where $\delta(n)$ is the number of $D$-nodes in a divide-and-conquer tree with $n$ leaves.
\end{theorem}
\begin{proof}
    There are two cases: $r=2^{k-1}$ and $r \ne 2^{k-1}$.
    \newline\newline
    \emph{Case 1: $r=2^{k-1}$}.
    We observe that by Theorem \ref{exp_D_recursive}, the number of $D$-nodes on a tree with $2^k+2^{k-1}=2^{k-1}\cdot 3$ leaves is $2^{k-1}\cdot \delta(3) = 2^{k-1}\cdot (\delta(2)+\delta(1))=2^{k-1}$.
       
    In a complete full binary tree with $2^k+2^{k-1}$ leaves, the left subtree is a perfect tree with $2^k$ leaves, and the right is a perfect tree with $2^{k-1}$ leaves, by Lemma \ref{cft_mod}. This makes a contribution to the index total of $2^{k-1} = 2^{k-1}-(r \bmod 2^{k-1})$. Neither the left nor the right subtrees contribute any further $D$-nodes, so the Colless index is exactly $2^{k-1}$. 
    \newline\newline
    \emph{Case 2: $r \ne 2^{k-1}$}. 
    We proceed by induction, and use the characterization of $D$-nodes in a divide-and-conquer tree given in Theorem \ref{exp_D_closed}.
     
    If $r<2^{k-1}$, the $i^\text{th}$ bit of $n$ is $0$, the left subtree has $2^{k-1}+(r \bmod 2^{k-1})$ leaves, and the right is a perfect tree with $2^{k-1}$ leaves, by Lemma \ref{cft_mod}. This makes a contribution to the index total of $r \bmod 2^{k-1} = n\bmod 2^{k-1}$.  

    If $r>2^{k-1}$, the $i^\text{th}$ bit of $n$ is $1$, the left subtree is a perfect tree with $2^k$ leaves, and the right has $2^{k-1}+(r \bmod 2^{k-1})$ leaves, by Lemma \ref{cft_mod}. This makes a contribution to the index total of $2^{k-1}-(r \bmod 2^{k-1})= 2^{k-i}-(n\bmod 2^{k-1})$. 
    
    In either case, after the highest-level, we need only look for further contribution at a tree having $2^{k-1}+(r \bmod 2^{k-1})$ leaves, since the other subtree is perfect and contributes no further $D$-nodes. Without loss of generality, we may assume that $r<2^{k-1}$, so the tree has $2^{k-1}+r$ leaves. By induction, the number of leaves on the subtree is 	
    $$
			\sum_{i=0}^{k-2} \lambda_i(n) \text{, where } \lambda_i(n)= 
			\begin{cases}
			    (n\bmod 2^i), & \text{ if the}\ i^\text{th} \text{ bit  of $n$ is }  0;\\
			    2^i-(n\bmod 2^i), & \text{ if the}\ i^\text{th} \text{ bit of $n$ is }  1.\\
	        \end{cases}	  
$$
Adding the highest-level contribution to this gives us the formula for $D$-nodes given in Theorem \ref{exp_D_closed}.
\end{proof}

\begin{corollary}
    The Colless index on a complete full binary tree with $n=2^k+r$ leaves is $2^k \cdot \tau\left(\frac{r}{2^k}\right)$, where $\tau$ is the Takagi function.
\end{corollary}
\subsection{The Colless index on MinD trees} \label{colless_max}
We consider $n$ in its binary composition $n_{k-1}n_{k-2} \dots n_1n_0$, with $\omega(n)$ $1$s, and define $R(n)=\{\rho \mid n_\rho = 1\}$. A MinD tree on $n$ nodes is then constructed by forming a tree with $\omega(n)$ leaves, and attaching perfect trees with $2^{n_\rho}$ leaves in place of the leaves of $T$, as discussed in Section \ref{max_S_section} and shown in Fig. \ref{fig:minD_27}. 
The $D$-nodes are then all the nodes of the base tree that has $\omega(n)$ leaves. All MinD trees on $n$ leaves are of this form. 

In the propositions throughout Section \ref{colless_max}, dealing with MinD trees, we assume that
\begin{itemize}
    \item $n=\sum\limits_{1}^{\omega(n)}2^{\rho_i}$, where $\rho_{\omega(n)} > \dots > \rho_2 > \rho_1$.
    \item The base tree $T$ has $\omega(n)$ leaves.
    \item $\pi$ is a permutation on $\{\rho_1, \dots, \rho_{\omega(n)}\}$.
    \item The perfect tree ``leaves'' are arranged left-to-right depth-first, from the perfect tree having $2^{\pi(\rho_{\omega(n)})}$ leaves as the leftmost ``leaf'', to the perfect tree having $2^{\pi(\rho_1)}$ as  the rightmost ``leaf''.
    \item $T$ has children $T_L$ and $T_R$, where $T_L$ has $\ell_L$ leaves and $T_R$ has $\ell_R=\omega(n)-\ell_L$ leaves, numbered as above. 
\end{itemize} 

Because the base trees can be any of the $(2 \cdot \omega(n)-3)!!$ binary full trees on $\omega(n)$ leaves, it is difficult to do a complete characterization of the Colless index for trees with the minimal number of $D$-nodes. 
To do this, we need to be able to define not only the base tree itself, but also the placement of the perfect tree "leaves" at the leaf levels of the base tree, since the Colless index will be highly dependent upon the number of leaves in each perfect subtree. 

However, we are able to show the maximum and minimum Colless indices that MinD trees can take. Both the maximum and minimum have ladder trees as their base tree. The tree having maximum Colless index arranges the perfect tree "leaves" in an ascending order, and the tree having minimum Colless index arranges the perfect tree "leaves" in an descending order. 
All MinD trees have quite small normalized Colless index, as might be expected, so may be considered to be reasonably well balanced.

\subsubsection{Basic result on MinD trees}
\begin{theorem} \label{colless_general_base}
    Let $c_{(T,\pi)}$ be the Colless index of a MinD tree. Then 
    $$
        c_{(T,\pi)}(n) = \lvert \sum_{i=1}^{\ell_L}2^{\pi(\rho_i)}-\sum_{j=\ell_L+1}^{\omega(n)}2^{\pi(\rho_j)} \rvert +  c_{(T_L,\pi)}(\ell_L)+c_{(T_R,\pi)}(\ell_R)
   $$
\end{theorem}
\begin{proof}
    Follows directly from Definition \ref{def_colless} of the Colless index and the definition of $\pi$.
\end{proof}

Theorem \ref{colless_general_base} becomes a recursive formula when the base trees of the left and right subtrees $T_L$ and $T_R$ are of the same form as the base tree $T$. This is true of all the special cases we have discussed above: the ladder trees, the divide-and-conquer trees, and the complete full binary trees. We discuss these special cases below.

\subsubsection{Base tree is a ladder tree} \label{colless_ladder_section}
In the propositions throughout Section \ref{colless_ladder_section}, dealing specifically with ladder base trees, we assume that
\begin{itemize}
    \item $n=\sum\limits_{1}^{\omega(n)}2^{\rho_i}$, where $\rho_{\omega(n)} > \dots > \rho_2 > \rho_1$.
    \item The base tree is a ladder tree with $\omega(n)$ leaves.
    \item The perfect tree ``leaves'' are arranged so that the perfect tree having $2^{\pi(\rho_{\omega(n)})}$ leaves is on the top rung, down to the perfect trees having $2^{\pi(\rho_2)}$ and $2^{\pi(\rho_1)}$ leaves are on the two bottom rungs. (This follows from the previous assumptions.)
\end{itemize} 
Two examples of this type of tree are illustrated in Fig. \ref{fig:colless_ascdesc}.
\begin{proposition}\label{general_ladder_formula}
    Let $c_{\pi}$ be the Colless index of a MinD tree on $n$ leaves, with a ladder base tree. Then 
    $$
        c_{\pi}(n) = \lvert n-2^{\pi(\rho_{\omega(n)})+1}\rvert +c_\pi(n-2^{\pi(\rho_{\omega(n)})})
    $$
\end{proposition}
\begin{proof}
    Follows from Theorem \ref{colless_general_base}, since one of the top-level subtrees is a perfect tree with Colless index $0$, and the other is also a ladder tree. The left top-level subtree has Colless index $\lvert n-2^{\pi(\rho_{\omega(n)})} - 2^{\pi(\rho_{\omega(n)})} \rvert$. The  right has Colless index  $c_\pi(n-2^{\pi(\rho_{\omega(n)})})$.
\end{proof}
\begin{figure}[h!tb]
	\tikzstyle{mynode}=[circle, draw]
	\tikzstyle{leaf}=[fill=gray!45]
	\tikzstyle{perfectnode}=[rectangle, draw]
	\tikzstyle{dnode}=[regular polygon sides=4, minimum size=0.6cm, draw]
	\tikzstyle{perfect}=[fill=gray!20]
		\tikzstyle{leaf}=[fill=white]
	\tikzstyle{myarrow}=[]
	\centering
	\begin{subfigure}[b]{.47\textwidth}
		\centering
	    \resizebox{.75\textwidth}{!}{%
		\begin{tikzpicture}	
			\centering
			\node[dnode, label={right:$D_{\omega(n)-1}$}] (L1) at (0, 3) {\footnotesize\texttt{}};
			\node[dnode, label={right:$D_j$}] (L2a) at (1.5, 1.5) {\footnotesize\texttt{}};
			\node[perfectnode,perfect] (L2b) at (-1, 2) {\footnotesize\texttt{P$2^{\rho_{\omega(n)}}$}};
			\node[dnode, label={right:$D_1$}] (L3a) at (3, 0) {\footnotesize\texttt{}};
			\node[perfectnode,perfect] (L3b) at (0.5, 0.5) 		{\footnotesize\texttt{P$2^{\rho_j}$}};
			\node[perfectnode,perfect] (L4a) at (4, -1) {\footnotesize\texttt{P$2^{\rho_1}$}};
			\node[perfectnode,perfect] (L4b) at (2, -1) {\footnotesize\texttt{P$2^{\rho_2}$}};
			\draw[myarrow, dashed] (L1) to node[midway,above=0pt, left=7pt]{} (L2a);
			\draw[myarrow] (L1) to node[midway,above=0pt, left=-2.5pt]{} (L2b);
			\draw[myarrow, dashed] (L2a) to node[midway,above=0pt, left=4pt]{} (L3a);
			\draw[myarrow] (L2a) to node[midway,above=0pt, right=-1pt]{} (L3b);
			\draw[myarrow] (L3a) to node[midway,above=0pt, left=4pt]{} (L4a);
			\draw[myarrow] (L3a) to node[midway,above=0pt, right=-1pt]{} (L4b);
		\end{tikzpicture}
		} 
		\caption{Descending ladder base tree. The perfect tree ``leaves'' are arranged from largest to smallest, from the top down. The $D$-node $D_j$ contributes $2^{\rho_j}-(\sum_{i< j}2^{\rho_i})$ to the Colless index $c_{desc}(n)$.}
		\label{fig:colless_descending}
	\end{subfigure}\hfill%
	\begin{subfigure}[b]{.47\textwidth}
		\centering
	    \resizebox{.75\textwidth}{!}{%
		\begin{tikzpicture}	
			\centering
			\node[dnode, label={right:$D_1$}] (L1) at (0, 3) {\footnotesize\texttt{}};
			\node[dnode, label={right:$D_j$}] (L2a) at (1.5, 1.5) {\footnotesize\texttt{}};
			\node[perfectnode,perfect] (L2b) at (-1, 2) {\footnotesize\texttt{P$2^{\rho_1}$}};
			\node[dnode, label={right:$D_{\omega(n)-1}$}] (L3a) at (3, 0) {\footnotesize\texttt{}};
			\node[perfectnode,perfect] (L3b) at (0.5, 0.5) 	{\footnotesize\texttt{P$2^{\rho_j}$}};
			\node[perfectnode,perfect] (L4a) at (4, -1) {\footnotesize\texttt{P$2^{\rho_{\omega(n)}}$}};
			\node[perfectnode,perfect] (L4b) at (2, -1)             {\footnotesize\texttt{P$2^{\rho_{(\omega(n)-1)}}$}};
			\draw[myarrow, dashed] (L1) to node[midway,above=0pt, left=7pt]{} (L2a);
			\draw[myarrow] (L1) to node[midway,above=0pt, left=-2.5pt]{} (L2b);
			\draw[myarrow, dashed] (L2a) to node[midway,above=0pt, left=4pt]{} (L3a);
			\draw[myarrow] (L2a) to node[midway,above=0pt, right=-1pt]{} (L3b);
			\draw[myarrow] (L3a) to node[midway,above=0pt, left=4pt]{} (L4a);
			\draw[myarrow] (L3a) to node[midway,above=0pt, right=-1pt]{} (L4b);
		\end{tikzpicture}
		} 
		\caption{Ascending ladder base tree. The perfect tree ``leaves'' are arranged from smallest to largest, from the top down. The $D$-node $D_j$ contributes $(\sum_{i>j}2^{\rho_i})-2^{\rho_j}$ to the Colless index $c_{asc}(n)$.}
		\label{fig:colless_ascending}
	\end{subfigure}%
	\mycaption{MinD trees on $n$ leaves with descending and ascending ladder base trees} {$n=\sum_{1}^{\omega(n)}2^{\rho_i}$, where $\rho_{\omega(n)} > \dots > \rho_2 > \rho_1$, and there are $(\omega(n)-1)$ $D$-nodes. Each rectangular node represents a perfect tree on $2^j$ leaves, labeled as $P2^j$. Since these subtrees are perfect, they contribute no $D$-nodes.}
	\label{fig:colless_ascdesc}		
\end{figure}
\paragraph{Exponents descending.}
In this subsection, we arrange the powers of two in descending order; in other words, the permutation $\pi$ is the identity. We thus express the exponents leaving off the identity permutation. So
\begin{itemize}
    \item $n=\sum\limits_{1}^{\omega(n)}2^{\rho_i}$, where $\rho_{\omega(n)} > \dots > \rho_2 > \rho_1$.
    \item The base tree is a ladder tree with $\omega(n)$ leaves.
    \item The perfect tree ``leaves'' are arranged so that the perfect tree having $2^{\rho_{\omega(n)}}$ leaves is on the top rung, down to the perfect trees having $2^{\rho_2}$ and $2^{\rho_1}$ leaves are on the two bottom rungs. 
\end{itemize} 
This type of tree is illustrated in Fig. \ref{fig:colless_descending}.
\begin{lemma} \label{desc_build_lemma}
    Let $T$ on $n=2^k+r$ leaves  with $0\le r < 2^k$ be described as above. Then
    $$
        c_{desc}(n) = 2^k-r+c_{desc}(r)
    $$
\end{lemma}
\begin{proof}
    Follows from Proposition \ref{general_ladder_formula}.
\end{proof}
\begin{theorem} \label{laddercolless_desc}
    Let $T$ be a MinD tree on $n$ leaves with a ladder  base tree, where the perfect tree ``leaves'' are arranged in descending order, i.e., $\pi$ is the identity. Then the Colless index $c_{desc}(n)$ on this tree is 
    $$
        c_{desc}(n)=
        \begin{cases}
            2c_{desc}(m), &\text{if } n=2m;\\            2c_{desc}(m)-\omega(m)+2^{\rho_2}, &\text{if } n=2m+1.
        \end{cases}
    $$
\end{theorem} 
\begin{proof}
    Proceed by induction, and note that the statement is true for $n=2$ and $n=3$. 
    
    \emph{Case 1: $n=2m.$} In this case, $2m=n=2^k+r = 2^k+2t$, for some $t$.
    \begin{flalign*}
        c_{desc}(n) = c_{desc}(2m)
        &=2^k-2t + c_{desc}(2t)\\
        &= 2^k-2t + 2c_{desc}(t)& \text{ by induction}&\\
        &= 2 \cdot (2^{k-1}-t + c_{desc}(t))\\
        &= 2 \cdot c_{desc}(2^{k-1}+t)&  \text{ by Proposition \ref{desc_build_lemma}}\\
        &= 2 \cdot c_{desc}(m)
    \end{flalign*}

    \emph{Case 2: $n=2m+1.$} In this case, $m=\sum\limits_{i=2}^{\omega(n)}2^{\rho_i-1}$. $\rho_1=0$ so $2^{\rho_1}=1$. The tree on $n$ leaves may be constructed from the tree on $2m$ leaves by taking the rightmost of the twin leaves at the bottom of the base ladder tree, and splitting it into two children leaves with the left one the perfect tree on $2^{\rho_2}$ leaves and the right the perfect tree with $1$ leaf. Each of the $\omega(m)=\omega(2m)=\omega(n)-1$ subsums of the tree on $2m$  leaves is decremented by $1$ in calculating the Colless index for $n$. Also, the contribution of these children leaves to the Colless index is $2^{\rho_2}-1$. So 
    \begin{flalign*}
        c_{desc}(n) 
        &=c_{desc}(2m)-(\omega(n)-1)+(2^{\rho_2}-1) \\
        &=c_{desc}(2m)-\omega(2m)+(2^{\rho_2}-1)& \text{ by Lemma \ref{weightlemma_odd}}&\\
        &= 2c_{desc}(m)-\omega(m)+2^{\rho_2}& \text{ by Case 1 and Lemma \ref{weightlemma_mult2}}
    \end{flalign*}
\end{proof}
\begin{corollary} \label{desc_list}
The following are true of $c_{desc}$:
    \begin{enumerate}[label=(\alph*)]
        \item $c_{desc}(2^k)=0$
        \item $c_{desc}(2^k+1)=2^k-1$
        \item $c_{desc}(2^k+2^j)=2^k-2^j$
        \item $c_{desc}(2^k+2^j+1)=2^k-2$
        \item $c_{desc}(2^k-1)=k-1$, when $k>0$
        \item $c_{desc}(2^k-2^j)=2^j \cdot (k-j-1)$
    \end{enumerate}
\end{corollary}
\begin{proof}
    Follows directly from Theorem \ref{laddercolless_desc}.  
\end{proof}
\begin{proposition} \label{new_colless_desc_prop}
    Let $$n=2^k+n_{k-1}\cdot 2^{k-1} + d$$
    where $n_{k-1}=0$ or $1$, and where $0<d<2^{k-1}$. Let the perfect tree ``leaves'' be arranged in descending order. Then the Colless index of this tree is 
    $$
        c_{desc}(n) = c_{desc}(2^{k-1}+d) + (2^{k-1} - n_{k-1}\cdot d)
    $$
\end{proposition}
\begin{proof}
If $n$ is even, let $d=2e$.
\begin{flalign*}
    c_{desc}(n) 
    &= c_{desc}(2^k+n_{k-1}\cdot 2^{k-1} + 2e)& \\
    &= 2\cdot c_{desc}(2^{k-1}+n_{k-1}\cdot 2^{k-2} + e)& \text{ by Theorem \ref{laddercolless_desc}}&\\
    &= 2(c_{desc}(2^{k-2}+e) + (2^{k-2} - n_{k-1}\cdot e))& \text{ by induction hypothesis}&\\
    &= 2\cdot c_{desc}(2^{k-2}+e) + (2^{k-1} - n_{k-1}\cdot 2e)\\
    &= c_{desc}(2^{k-1}+2e) + (2^{k-1} - n_{k-1}\cdot 2e)& \text{ by Theorem \ref{laddercolless_desc}}&\\
    &= c_{desc}(2^{k-1}+d) + (2^{k-1} - n_{k-1}\cdot d))
    \end{flalign*}
If $n$ is odd, let $d=2e+1$.\newline
By Theorem \ref{laddercolless_desc},
\begin{flalign*}
    c_{desc}(n) 
    &= c_{desc}(2^k+n_{k-1}\cdot 2^{k-1}+2e+1) \\
    &= c_{desc}(2\cdot (2^{k-1}+n_{k-1}\cdot 2^{k-2}+e)+1) \\
    &= 2\cdot c_{desc}(2^{k-1}+n_{k-1}\cdot 2^{k-2}+e)- \omega(2^{k-1}+n_{k-1}\cdot 2^{k-2}+e) +2^{\rho_2}
\end{flalign*}
Induction gives
\begin{flalign*}
    c_{desc}(n) 
    &= 2\cdot [c_{desc}(2^{k-2}+ e)+(2^{k-2} - n_{k-1}\cdot e)] - \omega(2^{k-1}+n_{k-1}\cdot 2^{k-2}+e) +2^{\rho_2}\\
    &= 2\cdot c_{desc}(2^{k-2}+ e)+2\cdot (2^{k-2} - n_{k-1}\cdot e) - \omega(2^{k-1}+n_{k-1}\cdot 2^{k-2}+e) +2^{\rho_2}\\
    &= 2\cdot c_{desc}(2^{k-2}+ e)+(2^{k-1} - n_{k-1}\cdot 2e) - \omega(2^{k-1}+n_{k-1}\cdot 2^{k-2}+e) +2^{\rho_2}\\
    &= 2\cdot c_{desc}(2^{k-2}+ e) - \omega(2^{k-1}+n_{k-1}\cdot 2^{k-2}+e) +2^{\rho_2}+(2^{k-1} - n_{k-1}\cdot 2e)
\end{flalign*}
Again, application of Theorem \ref{laddercolless_desc} gives
\begin{flalign*}
    c_{desc}(n) 
    &= c_{desc}(2^{k-1}+ 2e) - \omega(2^{k-1}+n_{k-1}\cdot 2^{k-2}+e) +2^{\rho_2}+(2^{k-1} - n_{k-1}\cdot 2e)
\end{flalign*}
$e<2^{k-2}$, so $2e<2^{k-1}$, and $n_{k-1}=0$ or $1$, so
\begin{flalign*}
    c_{desc}(n) 
    &= c_{desc}(2^{k-1}+ 2e)- (\omega(2^{k-1}+ e) + n_{k-1}) +2^{\rho_2}+(2^{k-1} - n_{k-1}\cdot 2e)\\
    &= c_{desc}(2^{k-1}+ 2e)- (\omega(2^{k-1}+ 2e) + n_{k-1}) +2^{\rho_2}+(2^{k-1} - n_{k-1}\cdot 2e)  
\end{flalign*}
A sequence of arithmetic operations gives
\begin{flalign*}
    c_{desc}(n) 
    &= c_{desc}(2^{k-1}+ 2e)- \omega(2^{k-1}+ 2e) + 2^{\rho_2}+(2^{k-1} - n_{k-1}\cdot 2e -n_{k-1})  \\
    &= c_{desc}(2^{k-1}+ 2e)- \omega(2^{k-1}+ 2e) + 2^{\rho_2}+(2^{k-1} - n_{k-1}\cdot (2e +1))  \\
    &= c_{desc}(2^{k-1}+ 2e)- \omega(2^{k-1}+ 2e) + 2^{\rho_2}+(2^{k-1} - n_{k-1}\cdot d)  
\end{flalign*}
By application of the weight lemma Lemma \ref{weightlemma_mult2}:
\begin{flalign*}
    c_{desc}(n) 
    &= c_{desc}(2^{k-1}+ 2e)- \omega(2^{k-2}+ e) + 2^{\rho_2}+(2^{k-1} - n_{k-1}\cdot d) 
\end{flalign*}
and the proof is completed by two applications of Theorem \ref{laddercolless_desc}.
\begin{flalign*}
    c_{desc}(n) 
    &= 2\cdot c_{desc}(2^{k-2}+ e)- \omega(2^{k-2}+ e) + 2^{\rho_2}+(2^{k-1} - n_{k-1}\cdot d) \\
    &= c_{desc}(2^{k-1}+ 2e+1) + (2^{k-1} - n_{k-1}\cdot d)  \\
    &= c_{desc}(2^{k-1}+ d) + (2^{k-1} - n_{k-1}\cdot d)
\end{flalign*}
\end{proof}
\begin{corollary} 
        Let $T$ be a MinD tree on $n$ leaves with a ladder  base tree, where the perfect tree ``leaves'' are arranged in descending order. Let $n=2^k+r$, where $r=n_{k-1}\cdot 2^{k-1} + d$, $n_{k-1}=0$ or $1$, and $0<d<2^{k-1}$. Then  
        $$
        c_{desc}(2^k+r) \le c_{desc}(2^k+1)
        $$
\end{corollary}
\begin{proof}
    Proof by induction on $k$. Note that this is true for $k=0,1$. 
    
    \begin{flalign*}
        c_{desc}(2^k+r)
        &=  c_{desc}(2^{k-1}+r) + (2^{k-1} - n_{k-1}\cdot r)& \text{ by Proposition \ref{new_colless_desc_prop}}&\\
        &\le  c_{desc}(2^{k-1}+1) + (2^{k-1} - n_{k-1}\cdot r)& \text{ by induction hypothesis}\\
        &=  (2^{k-1}-1) + (2^{k-1} - n_{k-1}\cdot r)& \text{ by Corollary \ref{desc_list}}\\
        &=  2^k-1 - n_{k-1}\cdot r  \\
        &\le  2^k-1   \\
       &= c_{desc}(2^k+1)& \text{ by Corollary \ref{desc_list}}
    \end{flalign*}
\end{proof}
\begin{corollary} \label{mincolcdesc}
        The Colless index $c_{desc}(n)$ is the minimum possible Colless index $\delta(n)$ if and only if $n=2^k$ or $n=2^{k+1}-2^j$. The values for these are:
    $$c_{desc}(2^k)=0=\delta(2^k)$$
    $$c_{desc}(2^{k+1}-2^j) = 2^j \cdot (k-j) = \delta(2^{k+1}-2^j)$$.
\end{corollary}
\begin{proof}
    The minimum possible Colless index is $\delta(n)$ by Proposition \ref{minimumcolless}.  The equalities for $c_{desc}$  follow from Corollary \ref{desc_list}. The equalities for $\delta$ follow from repeated application of Theorem \ref{exp_D_recursive}. So if $n=2^k$ or $n=2^{k+1}-2^j$, the Colless index $c_{desc}(n)$ is minimal.
    
    For the converse, we assume that the Colless index $c_{desc}(n)$ is minimal; i.e., $c_{desc}(n)=\delta(n)$, and proceed by induction on $k$ to show that $n=2^k$ or $n=2^{k+1}-2^j$, for some $j<k$. We assume that $n=2^k+r$,  where $r=n_{k-1}\cdot 2^{k-1} + d$, $n_{k-1}=0$ or $1$, and $0<d<2^{k-1}$, as in Proposition \ref{new_colless_desc_prop}. 
    
    We start by observing that the proposition is true for $k=1$, i.e. $n=2^1$ or $n=2^1-2^0$. Now let $k>1$, and assume $n \ne 2^k$. 
    \begin{flalign*}
        c_{desc}(n) &= c_{desc}(n-2^k) + (2^{k-1} - n_{k-1}\cdot d)& \text{ by Proposition \ref{new_colless_desc_prop}}&\\
        \delta(n) &= \delta(n-2^k) + (n_{k-1}2^{k-1} + (-1)^{n_{k-1}}d)& \text{ by Theorem \ref{exp_D_closed}}
    \end{flalign*}
    $c_{desc}(n)=\delta(n)$, so
    \begin{equation}\label{odd_cdesc}
    c_{desc}(n-2^k) + (2^{k-1} - n_{k-1}\cdot d)=\delta(n-2^k) + (n_{k-1}2^{k-1} + (-1)^{n_{k-1}}d)
    \end{equation}
    If $n_{k-1}=0$, this equation becomes
    $$c_{desc}(n-2^k) + 2^{k-1} =\delta(n-2^k) + d$$
    $c_{desc}(n-2^k) \ge \delta(n-2^k)$ and $2^{k-1} > d$, so 
    $$c_{desc}(n-2^k) + 2^{k-1} >\delta(n-2^k) + d$$
    which is a contradiction. Thus, $n_{k-1}=1$, and Equation (\ref{odd_cdesc}) becomes
    $$
    c_{desc}(n-2^k) + (2^{k-1} - d)=\delta(n-2^k) + (2^{k-1} - d)
    $$
    so $c_{desc}(n-2^k)=\delta(n-2^k)$, and by the induction hypothesis, $n-2^k=2^i-2^j$ or $n-2^k=2^i$. 
    
    So $n=2^k+2^i-2^j$ or $n=2^k+2^i$. 
    
    If $n=2^k+2^i-2^j$, $2^{k-1} \le  2^i-2^j< 2^k$, so $i=k$ and $n=2^k+2^k-2^j=2^{k+1}-2^j$.

    If $n=2^k+2^i$, $i=k-1$ since $n_{k-1}=1$, so  $n=2^k+2^{k-1}=2^{k+1}-2^{k-1}$.
\end{proof}

The MinD trees built on descending ladder trees where $n=2^k-2^j$ are complete full binary trees, so are not a new type of tree having minimal Colless index.

\paragraph{Exponents ascending.} In this subsection, we arrange the powers of two in descending order; i.e., $\pi({\omega(n)}) < \dots < \pi(\rho_2) < \pi(\rho_1)$. In other words, the permutation $\pi$ reverses the order of the exponents. So
\begin{itemize}
    \item $n=\sum\limits_{1}^{\omega(n)}2^{\rho_i}$, where $\rho_{\omega(n)} > \dots > \rho_2 > \rho_1$.
    \item The base tree is a ladder tree with $\omega(n)$ leaves.
    \item The perfect tree ``leaves'' are arranged so that the perfect tree having $2^{\rho_1}$ leaves is on the top rung, down to the perfect trees having $2^{\rho_{\omega(n)-1}}$ and $2^{\rho_{\omega(n)}}$ leaves are on the two bottom rungs. 
\end{itemize}
This type of tree is illustrated in Fig. \ref{fig:colless_ascdesc}(\subref{fig:colless_ascending}).
\begin{lemma} \label{asc_build_lemma}
    Let $T$ be a MinD tree on $n$ leaves with a ladder  base tree, where the perfect tree ``leaves'' are arranged in ascending order. Let $n=2^k+r$ leaves  with $0\le r < 2^k$. Then the Colless index $c_{asc}(n)$ on this tree is
    $$
        c_{asc}(n) = n-2^{(\rho_1)+1}+c_{asc}(n-2^{\rho_1})
    $$
\end{lemma}
\begin{proof}
    Follows from Proposition \ref{general_ladder_formula}.
\end{proof}
\begin{theorem} \label{laddercolless_asc}
    Let $T$ be a MinD tree on $n$ leaves with a ladder  base tree, where the perfect tree ``leaves'' are arranged in ascending order, with the perfect tree having $2^{\rho_1}$ leaves on the top rung, continuing down the rungs to the perfect trees having $2^{\rho_{\omega(n)-1}}$ and $2^{\rho_{\omega(n)}}$ leaves on the two bottom rungs. Let $n=\sum\limits_{1}^{\omega(n)}2^{\rho_i}$, where $\rho_{\omega(n)} > \dots > \rho_2 > \rho_1$. Then the Colless index $c_{asc}$ of this tree is 
    $$
        c_{asc}(n)=
        \begin{cases}
            2 \cdot c_{asc}(m), &\text{if } n=2m;\\ 2 \cdot c_{asc}(m)+2m-1, &\text{if } n=2m+1.
        \end{cases}
    $$
\end{theorem} 
\begin{proof}
    Proof by induction.
    \newline
If $n=2m$, then $\rho_1>0$ and $\rho_1-1$ is the position of the smallest non-$0$ bit of $m$.
    \begin{flalign*}
    c_{asc}(n) = c_{asc}(2m) &= 2m-2^{(\rho_1)+1}+c_{asc}(2m-2^{\rho_1})&\text{ by Lemma \ref{asc_build_lemma}}&\\
    &= 2m-2^{(\rho_1)+1}+c_{asc}(2(m-2^{(\rho_1)-1}))&\text{ since $\rho_1>0$}\\
    &= 2m-2^{(\rho_1)+1}+2\cdot c_{asc}(m-2^{\rho_1})&\text{ by induction}\\
    &= 2 \cdot(m-2^{(\rho_1)+1}+ c_{asc}(m-2^{\rho_1})\\
    &= 2 \cdot c_{asc}(m)&\text{ by Lemma \ref{asc_build_lemma}}
    \end{flalign*}
    If $n=2m+1$, 
    \begin{flalign*}
        c_{asc}(n) 
        &= n-2+c_{asc}(n-1)&\text{ by Lemma \ref{asc_build_lemma}}&\\
        &= 2m-1+c_{asc}(2m)\\
        &= c_{asc}(2m)+2m-1\\
        &= 2 \cdot c_{asc}(m)+2m-1 &\text{ by the proof for $n$ even}
    \end{flalign*}
\end{proof}
\begin{corollary}\label{c_asc_closed}
    Let the binary decomposition of $n$ be $n_k \dots n_1 n_0$, where $n_k=1$. Then 
    $$
        c_{asc}(n) = \sum_{i=0}^{k-1} 2^i n_i \left(\frac{n-(n \mod 2^i)}{2^i}-2\right)
    $$
\end{corollary}
\begin{proof}
    Proof by induction. This is true for $n=1,2,3$ by inspection. Let $n=2m$ or $2m+1$, depending on whether $n$ is even or odd, so the binary decomposition of $m$ is $n_{k-1} \dots n_2 n_1$.
    
    If $n=2m$, then 
    \begin{align*}
        c_{asc}(n) 
        &= 2\cdot c_{asc}(m)& \text{ by Theorem \ref{laddercolless_asc}}\\
        &= 2\cdot \sum_{i=0}^{k-2} 2^i n_{i+1} \left(\frac{m-(m \mod 2^i)}{2^i}-2\right)&\text{ by induction hypothesis}\\
        &= \sum_{i=0}^{k-2} 2^{i+1} n_{i+1} \left(\frac{m-(m \mod 2^i)}{2^i}-2\right)\\
        &= \sum_{i=0}^{k-2} 2^{i+1} n_{i+1} \left(\frac{n-(n \mod 2^{i+1})}{2^{i+1}}-2\right)\\
        &= \sum_{i=1}^{k-1} 2^i n_i \left(\frac{n-(n \mod 2^i)}{2^i}-2\right)\\
        &= \sum_{i=0}^{k-1} 2^i n_i \left(\frac{n-(n \mod 2^i)}{2^i}-2\right)&\text{ since $n_0=0$}
    \end{align*}
    If $n=2m+1$, then
    \begin{align*}
        c_{asc}(n) 
        &= 2\cdot c_{asc}(m) +2m -1& \text{ by Theorem \ref{laddercolless_asc}}\\
        &= 2\cdot c_{asc}(m) +n-2\\
        &= \sum_{i=1}^{k-1} 2^i n_i \left(\frac{n-(n \mod 2^i)}{2^i}-2\right) +n-2& \text{ by the above proof}\\
        &= \sum_{i=1}^{k-1} 2^i n_i \left(\frac{n-(n \mod 2^i)}{2^i}-2\right) + \left(1\cdot 1\left(\frac{n-0}{1}-2\right)\right) \\
        &= \sum_{i=0}^{k-1} 2^i n_i \left(\frac{n-(n \mod 2^i)}{2^i}-2\right)& \text{ since $n_0=1$}
    \end{align*}
\end{proof}
\begin{corollary} \label{asc_list}
The following are true of $c_{asc}$:
    \begin{enumerate}[label=(\alph*)]
        \item $c_{asc}(2^k)=0$
        \item $c_{asc}(2^k+1)=2^k-1$
        \item $c_{asc}(2^k+2^j)=2^k-2^j$
        \item $c_{asc}(2^k+2^{k-1})=2^{k-1}$
        \item $c_{asc}(2^k-1)=2^k(k-1)-3\cdot (2^{k-1}-1)$, when $k>1$
        \item $c_{asc}(2^k-2^j)=2^j \cdot(2^{k-j}(k-j-1)-3\cdot (2^{k-j-1}-1))$
        \item $c_{asc}(2^k-2^{k-2})=2^{k-2}$, when $k>1$
        \item $c_{asc}(2^k-2^{k-2}+1)=2^k-1$, when $k>2$
    \end{enumerate}
\end{corollary}
\begin{proof}
    Follows directly from Theorem \ref{laddercolless_asc}.  Statement (f) is obtained by solving the recurrence.
\end{proof}
\begin{lemma} \label{casc_lemma}
        Let $T$ be a MinD tree on $n$ leaves with a ladder  base tree, where the perfect tree ``leaves'' are arranged in ascending order, with the perfect tree having $2^{\rho_1}$ leaves on the top rung, continuing down the rungs to the perfect trees having $2^{\rho_{\omega(n)-1}}$ and $2^{\rho_{\omega(n)}}$ leaves on the two bottom rungs. Let $n=\sum\limits_{1}^{\omega(n)}2^{\rho_i}$, where $\omega(n)>1$ and $\rho_{\omega(n)} > \dots > \rho_2 > \rho_1$. Then
        $$
        c_{asc}(n)=c_{asc}(n-2^{\rho_{\omega(n)}})+(\omega(n)-1)\cdot 2^{\rho_{\omega(n)}} - 2^{\rho_{\omega(n)-1}}
        $$
\end{lemma}
\begin{proof}
    The base tree on $n-2^{\rho_{\omega(n)}}$ leaves having $\omega(n)-1$ perfect tree ``leaves'' in ascending order has Colless index $c_{asc}(n-2^{\rho_{\omega(n)}})$. The base tree having $\omega(n)$ perfect tree ``leaves'' in ascending order is constructed by replacing the bottom perfect tree ``leaf'' on $2^{\rho_{\omega(n)-1}}$ leaves with an interior node having children the perfect trees with $2^{\rho_{\omega(n)-1}}$ and $2^{\rho_{\omega(n)}}$ leaves.  Since the tree  is ascending, the Colless index of the $i^{th}$ $D$-node of this tree, counting from the bottom, where $i>1$ is then $2^{\rho_{\omega(n)}}+C(2^{\rho_{\omega(n)-i}})$, so
    \begin{align*}
        c_{asc}(n) &=2^{\rho_{\omega(n)}}-2^{\rho_{\omega(n)-1}} + \sum\limits_{i=2}^{\omega(n)-1}(2^{\rho_{\omega(n)}}+C(2^{\rho_{\omega(n)-i}}))\\
        &=2^{\rho_{\omega(n)}}-2^{\rho_{\omega(n)-1}} + \sum\limits_{i=2}^{\omega(n)-1}2^{\rho_{\omega(n)}}+\sum\limits_{i=2}^{\omega(n)-1}C(2^{\rho_{\omega(n)-i}})\\
        &=\sum\limits_{i=2}^{\omega(n)-1}C(2^{\rho_{\omega(n)-i}})+(\omega(n)-1)\cdot2^{\rho_{\omega(n)}}-2^{\rho_{\omega(n)-1}}\\
        &= c_{asc}(n-2^{\rho_{\omega(n)}})+(\omega(n)-1)\cdot 2^{\rho_{\omega(n)}}-2^{\rho_{\omega(n)-1}}  
    \end{align*}
\end{proof}
\begin{corollary} \label{max_casc}
        Let $n=2^k+r$, where $0<r<2^k$. Let the base tree be a ladder tree where the perfect tree``leaves'' are arranged in ascending order. Then  
        $$
        c_{asc}(2^k+r) \le c_{asc}(2^{k+1}-1)
        $$
\end{corollary}
\begin{proof}
\begin{flalign*}
    c_{asc}(2^k+r) 
    &= (\omega(2^k+r)-1)\cdot 2^k - 2^{\rho_{\omega(n)-1}} + c_{asc}(r)& \text{ by Lemma \ref{casc_lemma}}&\\
    &= ((\omega(r)+1)-1)\cdot 2^k - 2^{\rho_{\omega(n)-1}} + c_{asc}(r)& \text{ by Lemma \ref{weightlemma_addone}}\\
    &= \omega(r)\cdot 2^k - 2^{\rho_{\omega(n)-1}} + c_{asc}(r)\\
    &= \omega(r)\cdot 2^k - 2^{\rho_{\omega(r)}} + c_{asc}(r)& \text{ by Lemma \ref{weightlemma_addone}}\\
\end{flalign*}
\begin{flalign*}
    c_{asc}(2^{k+1}-1) 
    &= c_{asc}(\sum\limits_{i=0}^k 2^i)\\
    &= c_{asc}((\sum\limits_{i=0}^{k-1} 2^i))+(\omega(\sum\limits_{i=0}^k 2^i)-1)\cdot 2^k - 2^{k-1}& \text{ by Lemma \ref{casc_lemma}}&\\
    &= c_{asc}(2^k-1)+k\cdot 2^k - 2^{k-1}\\
    &= c_{asc}(2^k-1)+(\omega(r)\cdot 2^k +(k-\omega(r))\cdot 2^k) - 2^{k-1}\\
    &= c_{asc}(2^k-1)+\omega(r)\cdot 2^k +(2\cdot (k-\omega(r))-1)\cdot 2^{k-1} \\
\end{flalign*}
$c_{asc}(r) \le c_{asc}(2^k-1)$, by the induction hypothesis.\newline
$-2^{\rho_{\omega(r)}} \le -1  \le (2\cdot (k-\omega(r))-1)$. \newline
So 
\begin{flalign*} 
c_{asc}(2^k+r)
&=c_{asc}((2^k+r)-2^k)+\omega((2^k+r)-1)\cdot 2^k - 2^{\rho_{(\omega(2^k+r))-1}}&\\
&=c_{asc}(r)+\omega(r)\cdot 2^k - 2^{\rho_{\omega(r)}}& \text{ by Lemma \ref{weightlemma_addone}}\\
&\le c_{asc}(2^k-1)+\omega(r)\cdot 2^k +(2\cdot (k-\omega(r))-1)\cdot 2^{k-1} \\
&= c_{asc}(2^{k+1}-1) & \text{ by Lemma \ref{casc_lemma}}
\end{flalign*}
\end{proof}
\begin{corollary} \label{mincolcasc}
    The Colless index $c_{asc}(n)$ is the minimum possible Colless index $\delta(n)$ on trees with $n$ leaves if and only if $n=2^k$ or $n=2^k+2^{k-1}$. The values for these are:
    $$c_{asc}(2^k)=0=\delta(2^k)$$
    $$c_{asc}(2^k+2^{k-1}) = 2^{k-1} = \delta(2^k+2^{k-1})$$.
\end{corollary}
\begin{proof}
    The minimum possible Colless index is $\delta(n)$ by Proposition \ref{minimumcolless}.  
    The equalities for $c_{asc}$  follow from Corollary \ref{asc_list}. The equalities for $\delta$ follow from repeated application of Theorem \ref{exp_D_recursive}. So if $n=2^k$ or $n=2^k+2^{k-1}$, the Colless index $c_{asc}(n)$ is minimal.
    
    For the converse, we assume that the Colless index $c_{asc}(n)$ is minimal; i.e., $c_{asc}(n)=\delta(n)$, and proceed by induction on $k$ to show that $n=2^k$ or $n=2^k+2^{k-1}$. We assume that $n=2^k+r$, with $0\le r<2^k$, and show that $r=0$ or $r=2^{k-1}$. We start by observing that this is true for $k=1$, i.e. $n=2^1+2^0=3$. Now let $k>1$.
    
        If $n=2m$  is even, then 
    \begin{flalign*}
        2 \cdot c_{asc}(m) 
        &= c_{asc}(2m)& \text{ by Theorem \ref{laddercolless_asc}}&\\
        &= \delta(2m)& \text{ by assumption}\\
        &= 2 \cdot \delta(m)& \text{ by Theorem \ref{exp_D_recursive}}
    \end{flalign*}
    $c_{asc}(m)=\delta(m)$, so $m=2^i \text{ or } (2^i+2^{i+1})$, for some $i$, by induction, so $n=2^{i+1} \text{ or } (2^{i+1}+2^{i})$.
    
    If $n$ is odd, then $n=2^k+2d+1$ for some $k,d \ge 0$. Then
    \begin{flalign*}
        c_{asc}(n) 
        &= c_{asc}(2^k+2d+1) \\
        &= 2c_{asc}(2^{k-1}+d) + 2\cdot(2^{k-1}+d)-1& \text{ by Theorem \ref{laddercolless_asc}}&\\
        &= c_{asc}(2^k+2d) + (2^k+2d)-1& \text{ by Theorem \ref{laddercolless_asc}}
    \end{flalign*}
and
    \begin{flalign*}
        \delta(n) 
        &= \delta(2^k+2d+1) \\
        &= \delta(2^k+2d)+\floor{\log_2(2^k+2d)}-2\omega(2^k+2d)+2& \text{ by Theorem \ref{dnodes_another_recurrence}}&\\
        &= \delta(2^k+2d)+k-2\omega(2^k+2d)+2 \\
    \end{flalign*}
By assumption, $c_{asc}(n)=\delta(n)$, so
$$ 
c_{asc}(2^k+2d) + (2^k+2d)-1 = \delta(2^k+2d)+k-2\omega(2^k+2d)+2
$$
so
$$ 
c_{asc}(2^k+2d) + [(2^k+2d) +2\omega(2^k+2d) -k -3] = \delta(2^k+2d)
$$
We show first that $d=0$. Assume not, and $0<d<2^{k-1}$. Then 
    \begin{flalign*}
        2^k+2d +2\omega(2^k+2d) -k -3 
        &= 2^k+2d +2\omega(2^{k-1}+d) -k -3&\text{ by Lemma \ref{weightlemma_mult2}}&\\
        &\ge 2^k+2d +2\cdot 2 -k -3&\text{ since $\omega(d) \ge 1$}\\
        &= 2^k+2d -k +1\\
        &> 2^k -k &\text{ since $d>0$}\\
        &> 0& \text{ for all }k \ge 0.
    \end{flalign*}
So $c_{asc}(2^k+2d) < \delta(2^k+2d)$, which cannot be, since $\delta(2^k+2d)$ is the minimum possible Colless index on a tree with $2^k+2d$ leaves. So $d=0$, and $n=2^k+1$.
    \begin{flalign*}
        \delta(2^k+1)
        &=\delta(2^k)+\floor{\log_2(2^k)}-2\cdot \omega(2^k)+2& \text{ by Theorem \ref{dnodes_another_recurrence}}&\\
        &=0+\floor{\log_2(2^k)}-2\cdot \omega(2^k)+2& \text{ by $k$ applications of Theorem \ref{exp_D_recursive}}\\
        &=0+k-2\cdot 1+2\\
        &= k\\
  c_{asc}(2^k+1)&=2^k-1& \text{ by Corollary \ref{asc_list}}
\end{flalign*}
By assumption, $c_{asc}(2^k+1)=\delta(2^k+1)$, so $2^k-1=k$, so $k=1$, and $n=3=2^1+2^0$.
\end{proof}
\begin{corollary}
    $c_{asc}$ and $c_{desc}$ coincide if and  only if $n=2^k+2^j$ or $n=2^k$, and they both coincide with the minimal Colless index if and only if $n=2^k+2^{k-1}$ or $n=2^k$.
\end{corollary}
\begin{proof}
    Follows from Corollaries \ref{desc_list}, \ref{mincolcdesc}, \ref{asc_list}, and \ref{mincolcasc}.
\end{proof}

The MinD trees built on ascending ladder trees where $n=2^k+2^{k-1}$ are complete full binary trees, so are not a new type of tree having minimal Colless index.

\subsubsection{Comparing Colless indices on MinD trees}
In this section, we show that a MinD tree that has a base ladder tree and is descending meets the lower bound on Colless indices for MinD trees, and a MinD tree that has a base ladder tree and is ascending meets the upper bound on Colless indices for MinD trees. Fig. \ref{fig:colless_minD_27} gives an example of the Colless indices of all MinD trees on $27$ leaves, to give an idea of what these lower and upper bound trees look like, in relation to all MinD trees on $27$ leaves. The next series of lemmas shows how to construct MinD trees from smaller MinD trees. 
\input{figs/fig_colless27_MinD}
\begin{lemma} \label{bijection_2m}
    Let $m=\sum_{i=1}^{\omega(m)}2^{\rho_i}$. Let $M_1$ be the set of MinD trees on $m$ leaves, and let $M_2$ be the set of MinD trees on $2m$ leaves. There is then a bijection $\beta:T_m \rightarrow T_{2m}$ between $M_1$ and $M_2$, where $\beta$ leaves the base tree of $T_m$ fixed and replaces the perfect tree ``leaves'' $P2^i$ in $T_m$ with perfect tree ``leaves'' $P2^{i+1}$.
\end{lemma}
\begin{proof}
    The weights of $m$ and $2m$ are the same, by Lemma \ref{weightlemma_mult2}, so the set of base trees of $M_1$ and $M_2$ are the same, being the set of all trees having $\omega(m)$ leaves. Replacing the perfect tree ``leaves'' $P2^i$ in $T_m$ with perfect tree leaves $P2^{i+1}$ gives a tree in $M_2$, and conversely, replacing the perfect tree ``leaves'' $P2^i$ in $T_{2m}$ with perfect tree leaves $P2^{i-1}$ gives a tree in $M_1$.
\end{proof}
\begin{proposition} \label{twice_arb_colless}
    Let $T_m$ be a MinD tree on $m$ leaves, and let $\beta(T_m)$ be its image under the bijection in Lemma \ref{bijection_2m}. Then $c(\beta(T_m))= 2\cdot c(T_m)$. 
\end{proposition}
\begin{proof}
    Follows from the definition of the Colless index and Lemma \ref{bijection_2m}.
\end{proof}
There is no such easy bijection for $n=2m+1$. This can be seen by considering that for odd $n$, we are not only taking the original base tree and multiplying leaves by $2$, leaving $\omega(n)=\omega(2m)=\omega(m)$, as in Proposition \ref{twice_arb_colless},  but we are also adding a leaf to an even $2m$, thus we have $\omega(n)=\omega(2m+1)=\omega (2m)+1=\omega(m)+1$. So there are many $T_{2m+1}$ built from each $T_m$. 

However, we can say something about how the Colless index grows as we add new perfect trees. Proposition \ref{lemma_newleaftree} discusses how a new MinD tree may be constructed by adding a leaf, and Proposition \ref{prop_cn_from_cm} calculates the difference in the Colless indices of the two trees. Fig. \ref{fig:colless_arb} illustrates the process of adding perfect trees to a MinD tree, and will be referred to throughout this section.
\begin{proposition} \label{lemma_newleaftree}
    Let $T$ be a leaf-labelled full binary tree on $m$ leaves, and let $x$ be a new leaf to be added. The new full binary tree may be built either 1) by replacing a child $y$ of an interior node with a new interior node, whose children are the new leaf $x$ and the replaced child $y$, or 2) by joining the tree $T$ and the leaf $x$ as children of a new root node (so the entire tree $T$ is $y$). 
\end{proposition}
\begin{proof}
    The nodes of full binary trees have either $0$ or $2$ children. 
    Thus, the new tree cannot be formed by adding the leaf as a child to a leaf node, since that would produce an interior node with only $1$ child. Likewise, the new leaf cannot simply be appended to an interior node, since that would produce an interior node with $3$ children. 
    The other two alternatives are those stated in the lemma. Both satisfy the condition on the number of children of a node of a full binary tree, so both are possibilities in adding $x$ to $T$.
\end{proof}

Without loss of generality, this tree may be formed so that the leaf $x$ is the rightmost element in the new tree (so the last node counted) and $y$ is its left sibling. The arrangement of the new tree so that $x$ is rightmost follows from commutativity of children of nodes. We do this purely for ease of explication, since the two orderings of sibling children of a node are equivalent. 

The next series of propositions lead to an induction in Theorem \ref{th_bounds_MinD}, proving bounds on the Colless indices of MinD trees. This induction goes from a MinD tree on $n-2^{\rho_1}$ leaves to a MinD tree on $n$ leaves. Thus, $2^{\rho_1} \le 2^{\rho_i}$ for all perfect trees in the tree with $n$ leaves. It may seem more natural in an induction to add a perfect tree that is larger than the already-existing perfect trees, rather than smaller, as we do here. However, proceeding in this manner allows one to easily identify the Colless index from the tree with $n-2^{\rho_1}$ leaves in the calculation on the tree with $n$ leaves.
\begin{figure}[h!tb]
	\tikzstyle{mynode}=[circle, draw]
	\tikzstyle{dnode}=[regular polygon sides=4, minimum size=0.6cm, draw]
	\tikzstyle{perfectnode}=[rectangle, draw, fill=gray!20]
	\tikzstyle{blanknode}=[rectangle]
	\tikzstyle{trinode}=[draw=black, regular polygon, regular polygon sides=3,inner sep=.5pt]
	\tikzstyle{myarrow}=[]
	\centering
	\begin{subfigure}[b]{.47\textwidth}
		\centering
	    \resizebox{.5\textwidth}{!}{%
		\begin{tikzpicture}	
			\centering
			\node[dnode, label={right:$D_j$}] (L1) at (0, 3) {\footnotesize\texttt{}};
			\node[trinode] (L2a) at (-1, 2) {\footnotesize\texttt{$T_j$}};
			\node[dnode, label={right:$D_2$}] (L2b) at (2, 1) {\footnotesize\texttt{}};
			\node[trinode] (L3a) at (1, 0) 		{\footnotesize\texttt{$T_2$}};
			\node[trinode] (L3b) at (3,0) {\footnotesize\texttt{$T_1$}};
			\node[blanknode] (L4a) at (2, -1) {\footnotesize\texttt{}};
			\draw[myarrow] (L1) to node[midway,above=0pt, left=7pt]{} (L2a);
			\draw[myarrow, dashed] (L1) to node[midway,above=0pt, left=-2.5pt]{} (L2b);
			\draw[myarrow] (L2b) to node[midway,above=0pt, left=4pt]{} (L3a);
			\draw[myarrow] (L2b) to node[midway,above=0pt, right=-1pt]{} (L3b);
		\end{tikzpicture}
		} 
		\caption{The original tree $T'$ on $n-2^{\rho_1}=\sum_{2}^{\omega(n)}2^{\rho_i}$ leaves.There are $\omega(n)-2$ $D$-nodes, but some may be in the trees $T_i$, so $j-1 \le \omega(n)-2$.
		}
		\label{fig:colless_arb_1}
	\end{subfigure}\hfill%
	\begin{subfigure}[b]{.47\textwidth}
		\centering
	    \resizebox{.6\textwidth}{!}{%
		\begin{tikzpicture}	
			\centering
			\node[dnode, label={right:$D_j$}] (L1) at (0, 3) {\footnotesize\texttt{}};
			\node[trinode] (L2a) at (-1, 2) {\footnotesize\texttt{$T_j$}};
			\node[dnode, label={right:$D_2$}] (L2b) at (2, 1) {\footnotesize\texttt{}};
			\node[trinode] (L3a) at (1, 0) 		{\footnotesize\texttt{$T_2$}};
			\node[dnode, label={right:$D_1$}] (L3b) at (3, 0) {\footnotesize\texttt{}};
			\node[trinode] (L4a) at (2, -1) {\footnotesize\texttt{$T_1$}};
			\node[perfectnode] (L4b) at (4, -1) {\footnotesize\texttt{P$2^{\rho_1}$}};
			\draw[myarrow] (L1) to node[midway,above=0pt, left=7pt]{} (L2a);
			\draw[myarrow, dashed] (L1) to node[midway,above=0pt, left=-2.5pt]{} (L2b);
			\draw[myarrow] (L2b) to node[midway,above=0pt, left=4pt]{} (L3a);
			\draw[myarrow] (L2b) to node[midway,above=0pt, right=-1pt]{} (L3b);
			\draw[myarrow] (L3b) to node[midway,above=0pt, left=4pt]{} (L4a);
			\draw[myarrow] (L3b) to node[midway,above=0pt, right=-1pt]{} (L4b);
		\end{tikzpicture}
		} 
		\caption{The new tree $T$ on $n=\sum_{1}^{\omega(n)}2^{\rho_i}$ leaves. There are $\omega(n)-1$ $D$-nodes, but some may be in the trees $T_i$, so $j \le \omega(n)-1$.
		}
		\label{fig:colless_arb_2}
	\end{subfigure}\hfill%
	\mycaption{Construction of a MinD tree on $n$ leaves from a MinD tree on $n-2^{\rho_1}$ leaves} {The tree in (\subref{fig:colless_arb_2})  is obtained from the tree in (\subref{fig:colless_arb_1}) by inserting a new $D$-node where tree $T_1$ had been, making $T_1$ its left child, and adding the perfect tree on $2^{\rho_1}$ leaves as its right child. The $T_i$ are the left subtrees of the visible $D$-nodes between P$2^{\rho_1}$ and the root, and the $T_i$ may or may not themselves contain $D$-nodes. The shaded rectangular node represents a perfect tree on $2^{\rho_1}$ leaves, which contributes no $D$-nodes. }
	\label{fig:colless_arb}		
\end{figure}
\begin{proposition} \label{prop_cn_from_cm}
    Let $n=\sum_{1}^{\omega(n)}2^{\rho_i}$, where $\rho_{\omega(n)} > \dots > \rho_2 > \rho_1$.  Let $T'$ be a MinD tree on $(n-2^{\rho_1})$ leaves. Let $T$ be a MinD tree on $n$ leaves formed by affixing the perfect tree on $2^{\rho_1}$ leaves to $T'$ as a rightmost child, as discussed in Proposition \ref{lemma_newleaftree} and as shown in Fig. \ref{fig:colless_arb}. Let the $\{T_j\}$ be the left children of the $j$ internal nodes that are ancestors of the newly added $2^{\rho_1}$, and let $\lvert T_i \rvert$ be the number of descendant leaves in $T_i$. Then $T$ has Colless index
    \begin{equation}
        c_{T}(n)=c_{T'}(n-2^{\rho_1})+(\lvert T_1 \rvert  + 2^{\rho_1}\sum_{i =2}^j f(i)) - 2^{\rho_1}
    \end{equation}
    where 
    $$
        f(i) =
        \begin{cases}
		    -1, & \text{ if } T_i-\sum\limits_{0<h<i}T_h>0,\\
		    1, & \text{ if } T_i-\sum\limits_{0<h<i}T_h<0.
	  \end{cases}
    $$
\end{proposition}
\begin{proof}
    Let $T_{(n-2^{\rho_1})}$ be as illustrated in subfigure \ref{fig:colless_arb_1} and let $T_n$ be as illustrated in subfigure \ref{fig:colless_arb_2}. Let $t_i=\lvert T_i \rvert$, the number of leaves in $T_i$. 
    The Colless index of $T'$ is 
    \begin{equation} \label{eq_colless_t_prime}
        c_{T'}(n-2^{\rho_1})=\sum_{i=2}^j\lvert T_i-\sum_{0<h<i}T_h\rvert + \sum_{i=1}^j c_{T_i}(t_i)
    \end{equation}
    The Colless index of $T$ is 
    \begin{equation}  \label{eq_colless_t}
        c_{T}(n)=\sum_{i=1}^j\lvert T_i-(2^{\rho_1}+\sum_{0<h<i}T_h\rvert) + \sum_{i=1}^j c_{T_i}(t_i)
    \end{equation}.
    All $t_i$ are multiples of $2^{\rho_2}$, since they are sums of powers of $2$ greater than or equal to $2^{\rho_2}$. $(T_i-\sum_{0<h<i}T_h)$ is positive if and only if $(T_i-(2^{\rho_1}+\sum_{0<h<i}T_h))$ is positive, when $i>1$, since $2^{\rho_1} < 2^{\rho_2}$. So 
    $$
    \lvert T_i-(2^{\rho_1}+\sum_{0<h<i}T_h)\rvert = \lvert T_i-\sum_{0<h<i}T_h\rvert+2^{\rho_1}f(i)
    $$
    and 
    \begin{align*}
        c_{T}(n)
        &=\sum_{i=1}^j\left\lvert T_i-(2^{\rho_1}+\sum_{0<h<i}T_h)\right\rvert + \sum_{i=1}^j c_{T_i}(t_i)&\text{ by Equation \ref{eq_colless_t}}\\
        &=\sum_{i=2}^j\left\lvert T_i-(2^{\rho_1}+\sum_{0<h<i}T_h)\right\rvert +\lvert T_1-2^{\rho_1} \rvert + \sum_{i=1}^j c_{T_i}(t_i)\\
        &=\sum_{i=2}^j\left(\left\lvert T_i-\sum_{0<h<i}T_h\right\rvert+2^{\rho_1}f(i)\right) +\lvert T_1-2^{\rho_1} \rvert+ \sum_{i=1}^j c_{T_i}(t_i)&\text{ by definition of $f(i)$}\\
        &=\sum_{i=2}^j\left\lvert T_i-\sum_{0<h<i}T_h\right\rvert +\lvert T_1-2^{\rho_1} \rvert+ \sum_{i=1}^j c_{T_i}(t_i)+\sum_{i=2}^j 2^{\rho_1}f(i)\\
        &=\sum_{i=2}^j\left\lvert T_i-\sum_{0<h<i}T_h\right\rvert + \lvert T_1 \rvert - 2^{\rho_1}+\sum_{i=1}^j c_{T_i}(t_i) + \sum_{i=2}^j 2^{\rho_1}f(i)& \text{ since $\lvert T_1 \rvert>2^{\rho_1}$}\\
        &=\left(\sum_{i=2}^j\left\lvert T_i-\sum_{0<h<i}T_h\right\rvert+ \sum_{i=1}^j c_{T_i}(t_i)\right)  + \lvert T_1 \rvert+ 2^{\rho_1}\sum_{i=2}^j f(i) - 2^{\rho_1}\\
        &= c_{T'}(n-2^{\rho_1}) + \lvert T_1 \rvert + 2^{\rho_1}\sum_{i=2}^j f(i) - 2^{\rho_1}&\text{ by Equation \ref{eq_colless_t_prime}}
    \end{align*}
\end{proof}
\begin{corollary} \label{prop_cn_from_cm_desc}
    Let $n=\sum_{1}^{\omega(n)}2^{\rho_i}$, where $\rho_{\omega(n)} > \dots > \rho_2 > \rho_1$. Then
    $$
        c_{desc}(n) = c_{desc}(n-2^{\rho_1})+(2^{\rho_2}-2^{\rho_1}(\omega(n)-2))-2^{\rho_1}
    $$
\end{corollary}
\begin{proof}
    In a descending tree, $\lvert T_1\rvert = 2^{\rho_2}$, and $j=\omega(n)-1$. Also, in a descending MinD tree, $T_i-\sum_{0<h<i}T_h>0$ for all $i$, so $f(i)=-1$ for all $i$. The corollary then follows directly from Proposition \ref{prop_cn_from_cm}.
\end{proof}
\begin{corollary} \label{prop_cn_from_cm_asc}
    Let $n=\sum_{1}^{\omega(n)}2^{\rho_i}$, where $\rho_{\omega(n)} > \dots > \rho_2 > \rho_1$. Then
    $$
        c_{asc}(n) = c_{asc}(n-2^{\rho_1})+(n-2^{\rho_1}) -2^{\rho_1}
    $$
\end{corollary}
\begin{proof}
    This is Lemma \ref{asc_build_lemma}, but it also follows directly from Proposition \ref{prop_cn_from_cm}: in an ascending MinD tree, $\lvert T_1\rvert = n-2^{\rho_1}$, and $j=1$, so $\sum_{i=2}^1 f(i) = 0$. 
\end{proof}
\begin{theorem} \label{th_bounds_MinD}
    Let $n=\sum_{1}^{\omega(n)}2^{\rho_i}$, where $\rho_{\omega(n)} > \dots > \rho_2 > \rho_1$. Let $T$ be a MinD tree on $n$ leaves, and let $c_T(n)$ be its Colless index. Then 
    $$
        c_{desc}(n) \le c_T(n) \le c_{asc}(n)
    $$
\end{theorem}
\begin{proof}
The theorem is trivially true of $n=2^k$, since in that case $c_{desc}(n) = c_T(n) = c_{asc}(n)=0$. So we may assume that $n \ne 2^k$, and the perfect tree P$2^{\rho_1}$ has a left sibling subtree, as shown in Fig. \ref{fig:colless_arb}(\subref{fig:colless_arb_2}).
Let $T_1$ be this left sibling subtree of P$2^{\rho_1}$ (so $\lvert T_1 \rvert \ge 2^{\rho_2}$), and let $t_1$ be the number of leaves in $T_1$. Let $f(i)$ be as in Proposition \ref{prop_cn_from_cm}.

Let $H$ be such that $t_1 =\sum_{\rho_i \in H} 2^{\rho_i}$, and let $\lvert H \rvert=h$. Then 
\begin{equation}\label{j_lemma}
    j-1 \le (\omega(n)-h)-1
\end{equation}
\begin{align} \label{ineq_1}
    2^{\rho_2}  -2^{\rho_1}(\omega(n)-2)
    &\le \lvert T_1 \rvert  -2^{\rho_1}(\omega(n)-2)& \text{ since $\lvert T_1 \rvert \ge 2^{\rho_2}$} \nonumber\\
    &\le \lvert T_1 \rvert  - 2^{\rho_1}(j-1)&\text{ by Inequality \ref{j_lemma}, since $h \ge 1$}\nonumber\\
    &= \lvert T_1 \rvert  + 2^{\rho_1}\sum_{i =2}^j (-1)\nonumber\\
    &\le \lvert T_1 \rvert  + 2^{\rho_1}\sum_{i =2}^j f(i)& \text{ since $f(i) \ge -1$}
\end{align}
\begin{align}\label{ineq_2}
    \lvert T_1 \rvert  + 2^{\rho_1}\sum_{i =2}^j f(i))
    &= t_1 + 2^{\rho_1}\sum_{i =2}^j f(i))\nonumber\\
    &= \sum_{\rho_i \in H} 2^{\rho_i} + 2^{\rho_1}\sum_{i =2}^j f(i))\nonumber\\
    &\le \sum_{\rho_i \in H} 2^{\rho_i} + 2^{\rho_1}(j-1)& \text{ since $f(i) \le 1$}\nonumber\\
    &\le \sum_{\rho_i \in H} 2^{\rho_i} + 2^{\rho_1}((\omega(n)-h)-1)&\text{ by Inequality \ref{j_lemma}}\nonumber\\
    &= \sum_{\rho_i \in H} 2^{\rho_i} + 2^{\rho_1}(\omega(n)-h)- 2^{\rho_1}\nonumber\\
    &= \sum_{\rho_i \in H} 2^{\rho_i} + (2^{\rho_1}\sum_{\rho_i \notin H}1)- 2^{\rho_1}\nonumber\\
    &= \sum_{\rho_i \in H} 2^{\rho_i} + \sum_{\rho_i \notin H} 2^{\rho_1} - 2^{\rho_1}\nonumber\\
    &\le \sum_{\rho_i \in H} 2^{\rho_i} + \sum_{\rho_i \notin H} 2^{\rho_i} - 2^{\rho_1}\nonumber\\
    &= (\sum_{\rho_i \in H} 2^{\rho_i} + \sum_{\rho_i \notin H} 2^{\rho_i}) - 2^{\rho_1}\nonumber\\
    &= n - 2^{\rho_1}
\end{align}
Inequalities \ref{ineq_1} and \ref{ineq_2} give
\begin{equation} \label{cT_ineq}
    (2^{\rho_2}-2^{\rho_1}(\omega(n)-2)) 
    \le (\lvert T_1 \rvert  + 2^{\rho_1}\sum_{i =2}^j f(i))
    \le (n-2^{\rho_1})
\end{equation}
We now proceed by induction on $n$. The statement is true for $n=1,2$, since the Colless index of these perfect trees is $0$. It is also true for any $n=2^k$, since these MinD trees are also perfect, and have Colless index $0$. So by induction,
\begin{equation} \label{eq_ind_c}
    c_{desc}(n-2^{\rho_1}) \le c_T(n-2^{\rho_1}) \le c_{asc}(n-2^{\rho_1})
\end{equation}
By putting Equations \ref{cT_ineq} and \ref{eq_ind_c} together and subtracting $2^{\rho_1}$ from each part of the inequality, we get
\begin{align*}
    c_{desc}(n-2^{\rho_1})+(2^{\rho_2}-2^{\rho_1}(\omega(n)-2))-2^{\rho_1} 
    &\le c_T(n-2^{\rho_1}) + (\lvert T_1 \rvert  + 2^{\rho_1}\sum_{i =2}^j f(i)) -2^{\rho_1}\\
    &\le c_{asc}(n-2^{\rho_1}) +(n-2^{\rho_1}) -2^{\rho_1}
\end{align*}
which by Proposition \ref{prop_cn_from_cm} and Corollaries \ref{prop_cn_from_cm_desc} and \ref{prop_cn_from_cm_asc} is the same as 
$$
    c_{desc}(n) \le c_T(n) \le c_{asc}(n)
$$
\end{proof}
\paragraph{Normalized Colless index.}
We have shown that all MinD trees on $n$ leaves have Colless index between $c_{desc}(n)$ (that of a MinD tree having a descending ladder base tree) and $c_{asc}(n)$ (that of a MinD tree having an ascending ladder base tree). We now compare the Colless index of  MinD trees with general trees. 

In \cite{coronado20}, the divide-and-conquer tree on $n$ leaves was shown to have the least possible Colless index on trees with $n$ leaves. In \cite{mir18}, Mir, Rotger and Rosselló established that the ladder tree has the greatest possible Colless index. These two facts enable us to normalize the Colless index of any tree or type of tree. 

We show in this section that MinD trees have very low normalized Colless index, so are close to the minimal and are efficient.
\begin{lemma}\label{lemma_norm_colless}
    Let $c_{max}(n)$ be the maximal Colless index on a tree on $n$ leaves. Then the normalized Colless index $N(c_T(n))$ for a tree $T$ on $n$ leaves is 
    $$
        \frac{c_T(n)-\delta(n)}{c_{max}(n)-\delta(n)}
    $$
\end{lemma}
\begin{proof}
    The greatest possible Colless index on a tree on $n$ leaves is $c_L(n)$, which is shown in \cite{mir18}, and the least is $\delta(n)$, shown in \cite{coronado20}. Normalizing $c_T(n)$ against the greatest and least possible Colless indices gives the lemma.
\end{proof}
\begin{lemma}\label{lemma_normc_nbound}
    The normalized Colless index $N(c_T(n))$ is  not defined for  $n=1,2,3$ and is defined for all integers $n\ge 4$.
\end{lemma}
\begin{proof}
    $c_{max}(n)=\delta(n)=0$ for $n=1,2$, and $c_{max}(n)=\delta(n)=1$  for $n=3$. So $c_{max}(n)-\delta(n)=0$  for $n=1,2,3$, and $N(c_T(n))$ is  not defined.
    
    $c_{max}(n)=\frac{(n-1)(n-2)}{2}>\frac{n}{2}$, for $n\ge 4$, and $\delta(n)<\frac{n}{2}$, by Corollary \ref{cor_dnodes_half_n}, so $c_{max}(n)-\delta(n)>0$ for $n\ge 4$, and $N(c_T(n))$ is defined.
\end{proof}
\begin{example}\label{example1_colless_norm}
    The normalized $c_{asc}(2^k+1)$ is
    $$
    \frac{2^{k-1}\cdot 2\; -(k+1)}{2^{k-1}\cdot(2^k-1)-k}
    $$
    The maximum Colless index on a tree with $2^k+1$ leaves is $\frac{2^k   \cdot (2^k-1)}{2} = 2^{k-1}(2^k-1) $, by Proposition \ref{colless_ladder}  \cite{mir18}. The minimum Colless index is the number of $D$-nodes on a divide-and-conquer tree with $2^k+1$ leaves, which is $k$, by Lemma \ref{dnodes_diff_even} and Theorem \ref{exp_D_recursive}. The value of $c_{asc}(2^k+1)$ is $2^k-1$, by Corollary \ref{asc_list}. So the denominator of the normalized Colless index is ${2^{k-1}(2^k-1)-k}$, and the numerator is $2^k-1-k$. The result follows by algebraic manipulation.
\end{example}
\begin{example}\label{example2_colless_norm}
The normalized $c_{asc}(2^k-1)$ is
    $$
    \frac
    {2^{k-1}\cdot (2k-5)+(4-k)}
    {2^{k-1}\cdot(2^k-5)+(4-k)}
    $$
    The value of $c_{asc}(2^k-1)$ is
    $2^k(k-1)-3\cdot (2^{k-1}-1)$, by Corollary \ref{asc_list}. The minimum Colless index is the number of $D$-nodes on a divide-and-conquer tree with $2^k-1$ leaves, which is $k-1$, by Lemma \ref{dnodes_diff_even} and Theorem \ref{exp_D_recursive}. The maximum Colless index on a tree with $2^k-1$ leaves is $\frac{(2^k-2)   \cdot (2^k-3)}{2} = (2^{k-1}-1)(2^k-3) $, by Proposition \ref{colless_ladder}  \cite{mir18}. So the numerator of the normalized Colless index is $2^k(k-1)-3\cdot (2^{k-1}-1)-(k-1)$, and the denominator is ${(2^{k-1}-1)(2^k-3)-(k-1)}$. The result follows by algebraic manipulation.
\end{example}
Examples \ref{example1_colless_norm} and \ref{example2_colless_norm} suggest possible upper bounds for the normalized Colless function on a MinD tree in terms of $k=\floor{\log_2(n)}$. We show one such bound in Theorem \ref{theorem_upperbound_Colless_minD}.
\begin{lemma} \label{lemma_cmax4}
    Let $c_{max}(n)$ be the maximal Colless index of  a  tree on $n>1$ leaves, and let  $m=\floor{\frac{n}{2}}$. Then $$c_{max}(n)\ge4\cdot c_{max}(m)$$ 
    with equality  if and only if $n=2$.
\end{lemma}
\begin{proof}
If  $n=2m$ is even,
    \begin{align*}
        c_{max}(n) 
        &= \frac{(n-1)(n-2)}{2}& \text{ by Proposition \ref{colless_ladder}} \\
        &= \frac{(2m-1)(2m-2)}{2} \\
        &= (m-1)(2m-1) \\
        &\ge (m-1)(2m-4)& \text{since $m>0$}\\
        &= \frac{4\cdot (m-1)(m-2)}{2} \\
        &= 4\cdot c_{max}(m)
    \end{align*}
If  $n=2m+1$ is odd,
    \begin{align*}
        c_{max}(n) 
        &= \frac{(n-1)(n-2)}{2}& \text{ by Proposition \ref{colless_ladder}} \\
        &= \frac{((2m+1)-1)((2m+1)-2)}{2} \\
        &= (m)(2m-1) \\
        &> (m-1)(2m-1)  & \text{since $m>\frac{1}{2}$}\\
        &\ge (m-1)(2m-4) & \text{since $m>0$}\\
        &= \frac{4\cdot (m-1)(m-2)}{2} \\
        &= 4\cdot c_{max}(m)
    \end{align*}
\end{proof}
\begin{theorem}\label{theorem_upperbound_Colless_minD}
    Let $T$  be a MinD tree with $n \ge 4$ leaves, and let $c_T(n)$ be its Colless index. Then
    $$
    N(c_T(n))<\frac{2\cdot \floor{\log_2(n)}}{n}
    $$
\end{theorem}
\begin{proof}
    It suffices to prove the theorem for $T_{asc}$, the ascending ladder MinD tree on $n$ leaves, since by Theorem \ref{th_bounds_MinD}, all MinD trees on $n$ leaves have Colless index less than or equal to the Colless index of $T_{asc}$. We  proceed by induction.
    
    By Lemma \ref{lemma_norm_colless}, the Colless index of $T_{asc}$ is 
    $$
        \frac{c_{asc}(n)-\delta(n)}{c_{max}(n)-\delta(n)}
    $$
    If $n=2m$ is even, then the normalized Colless index of  $T$ is 
    \begin{align*}
        N(c_{asc}(n))  
        &= \frac{c_{asc}(n)-2\delta(m)}{c_{max}(n)-2\delta(m)}& \text{ by Theorem \ref{exp_D_recursive}}\\
        &= \frac{2c_{asc}(m)-2\delta(m)}{c_{max}(n)-2\delta(m)}&  \text{ by Theorem \ref{laddercolless_asc}} \\
        &\le \frac{2c_{asc}(m)-2\delta(m)}{4\cdot c_{max}(m)-2\delta(m)}&  \text{ by Lemma \ref{lemma_cmax4}} \\
        &\le \frac{2c_{asc}(m)-2\delta(m)}{4\cdot c_{max}(m)-4\delta(m)} \\
        &= \frac{N(c_{asc}(m))}{2} & \text{ by Theorem \ref{exp_D_recursive}}\\
        &\le \frac{2\cdot \floor{\log_2(m)}}{2m}\cdot  & \text{ by induction hypothesis}\\
        &= \frac{2\cdot (\floor{\log_2(n)}-1)}{2m}\\
        &< \frac{2\cdot \floor{\log_2(n)}}{n}
    \end{align*}
    If {n=2m+1} is odd, 
    \begin{align*}
        c_{asc}(n)-\delta(n)
        &= 2\cdot c_{asc}(m)+2m- 1-\delta(n)& \text{ by Theorem \ref{laddercolless_asc}}\\
        &\le 2\cdot c_{asc}(m)+2m- 1-1& \text{ by Corollary \ref{cor_dnodes_ineq_odd_n}}\\
        &= 2\cdot (c_{asc}(m)+m- 1)
    \end{align*}
and
\begin{align*}
    c_{max}(n)-\delta(n)  
    &= \frac{(n-1)(n-2)}{2}-\delta(n) \\
    &= \frac{(2m+1-1)(2m+1-2)}{2}-\delta(n) \\
    &= m(2m-1)-\delta(n) \\
    &\ge m(2m-1)-\floor{\frac{n}{2}}& \text{ by Corollary \ref{cor_dnodes_ineq_odd_n}}\\
    &= m(2m-1)-m \\
    &= 2m^2-2m 
\end{align*}
One or the other of the inequalities above must be a strict inequality for $n>3$, and by Lemma \ref{lemma_normc_nbound}, $n\ge 4$. So
\begin{align*}
    \frac{c_{asc}(n)-\delta(n)}{c_{max}(n)-\delta(n)} 
    &< \frac{2\cdot (c_{asc}(m)+m- 1)}{2m^2-2m}\\
    &= \frac{(c_{asc}(m)+m- 1)}{m^2-m}\\
    &\le \frac{2\cdot\frac{\floor{log_2(m)}}{m}+m- 1}{m^2-m}&\text{ by induction}\\
    &= \frac{2\cdot\floor{log_2(m)}+m^2- m}{m^3-m^2}
\end{align*}
Let  $k=\floor{log_2(n)}$. So $\floor{log_2(m)}=k-1$, and
\begin{align*}
    \frac{c_{asc}(n)-\delta(n)}{c_{max}(n)-\delta(n)} -\frac{2k}{n}
    &< \frac{2\cdot\floor{log_2(m)}+m^2-m}{m^3-m^2}-\frac{2k}{(2m+1)}\\
    &= \frac{2(k-1)+m^2-m}{m^3-m^2}-\frac{2k}{n}\\
    &= \frac{n\cdot (2(k-1)+m^2- m)}{(2m+1)(m^3-m^2)}-\frac{2k(m^3-m^2)}{(2m+1)(m^3-m^2)}\\
    &= \frac{n\cdot (2(k-1)+m^2-m)-2k(m^3-m^2)}
    {(2m+1)(m^3-m^2)}\\
    &= \frac{(2m+1) (2k-2+m^2-m)-2k(m^3-m^2)}
    {(2m+1)(m^3-m^2)}\\
    &= \frac{(-2k+2)m^3+(2k-1)m^2+(4k-5)m+(2k-2)}
    {(2m+1)m^2(m-1)}\\
    &\le 0& \text{for all $m \neq 2$}
\end{align*}
So  
\begin{align*}
     \frac{c_{asc}(n)-\delta(n)}{c_{max}(n)-\delta(n)} 
     &< \frac{2k}{n} =\frac{\floor{log_2(n)}}{n} & \text{for all $m \neq 2$}
\end{align*}
We have proved the theorem for all even $n\ge 4$, and for all odd $n=2m+1 \ge 4$ except when $m=2$, i.e. when $n=5$. By Corollary \ref{cor_dnodes2k} and Theorem \ref{dnodes_another_recurrence}, $\delta(5)=2$. By Corollary \ref{asc_list}, $c_{asc}(5)=3$. $c_{max}(5)=\frac{4\cdot 3}{2}=6$. So 
$$
N(c_{asc}(5))=\frac{3-2}{6-2}=\frac{1}{4}<\frac{2}{5}=\frac{\floor{log_2(5)}}{5}
$$
\end{proof}

The upper bound in Theorem \ref{theorem_upperbound_Colless_minD} is never met, since the  inequality is strict.
The global max for normalized MinD Colless indices occurs at $n=7$, and is approximately $0.38461538$. This is seen in Fig. \ref{fig:colless_minD}(\subref{fig:colless_minmax_values-c}), and is calculated from Example \ref{example2_colless_norm}.

In Corollary \ref{theorem_threebounds}, we establish three different bounds on the normalized Colless indices of the MinD trees on $n=2^k+r\ge 4$ leaves, where $r<2^k$. 
 These are clearly related, but may apply to different situations, so we state all three. 
\begin{itemize}
    \item The smallest bound, from Theorem  \ref{theorem_upperbound_Colless_minD},  is the tightest (but still is never met).
    \item The middle bound is slightly less tight, but has the advantage of continuity on $n>0$. 
    \item The largest is a constant-valued step function expressed solely in terms of $k$. 
\end{itemize}
\begin{theorem}\label{theorem_threebounds}
    Let $T$  be a MinD tree with $n=2^k+r$ leaves, where  $r<2^k$, and let $c_T(n)$ be its Colless index. Then
    $$
    N(c_T(n))<\frac{2k}{n}\le\frac{2\cdot \log_2(n)}{n}\le\frac{k}{2^{k-1}}
    $$
\end{theorem}
\begin{proof}
    The first inequality is Theorem \ref{theorem_upperbound_Colless_minD}. The second is true since $k=\floor{\log_2(n)}\le \log_2(n)$. 
    
    To see the third inequality, we consider the interval $[2^k, 2^{k+1})$, where $k>2$. We note that for all integers $k>2$, equality holds on that interval.
    The left side of the inequality is decreasing for $x>e$ (seen by differentiating, identifying the single critical point occurring at $e$, and observing that the function decreases monotonically for $x>e$), but the right side is constant on $[2^k, 2^{k+1})$. Thus, on any interval $[2^k, 2^{k+1})$ where $k\ge2$, the  inequality holds. 
\end{proof}
\begin{figure}[H]
    \centering
    \begin{subfigure}{0.48\textwidth} 
        \includegraphics[scale=0.6]{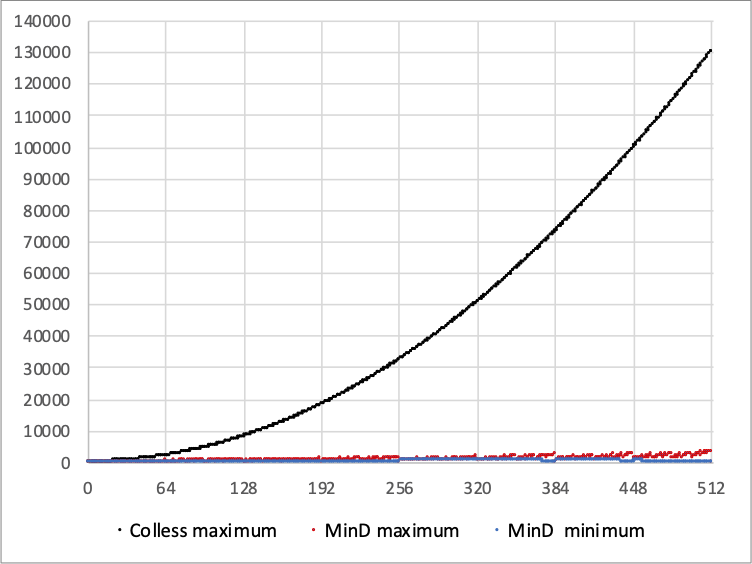}
        \caption{This non-logarithmic scale shows the small size of the Colless indices of the MinD trees, relative to the maximum Colless index, which grows quadratically.The minimum Colless index is not shown here, since at this scale, it is indistinguishable from the $x$-axis.}
        \label{fig:colless_minD-a}
    \end{subfigure}
    \hspace{1em}
    \begin{subfigure}{0.48\textwidth}
        \includegraphics[scale=0.6]{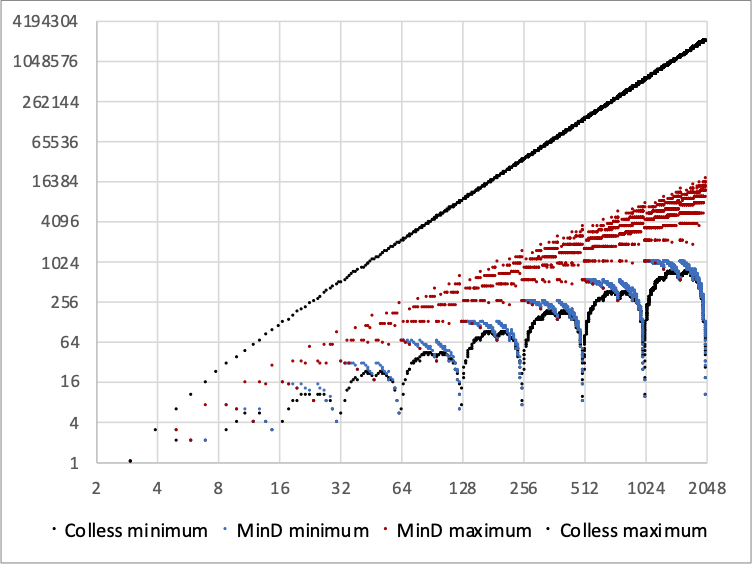}
        \caption{This logarithmic scale shows the Takagi structure of the Colless indices of the minimal divide-and-conquer trees and and also shows the Colless indices of the maximal ladder trees, as well as of the ascending and descending MinD trees.}
        \label{fig:colless_minD-b}
    \end{subfigure}
    
    \begin{subfigure}{0.48\textwidth} 
        \includegraphics[scale=0.6]{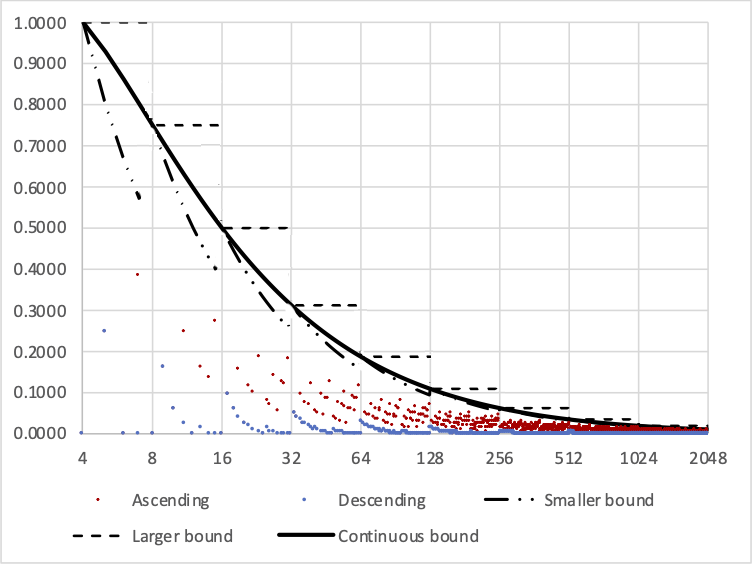}
        \caption{Colless indices on ascending and  descending MinD trees normalized against the minimum (0.0) and maximum (1.0) possible Colless indices. The logarithmic scale shows the decrease in terms of $\log_2(n)$. All except a few small $n$  are well below 10\% of the normalized max.}
        \label{fig:colless_minmax_values-c}
    \end{subfigure}
    \hspace{1em}
    \begin{subfigure}{0.48\textwidth}
        \includegraphics[scale=0.6]{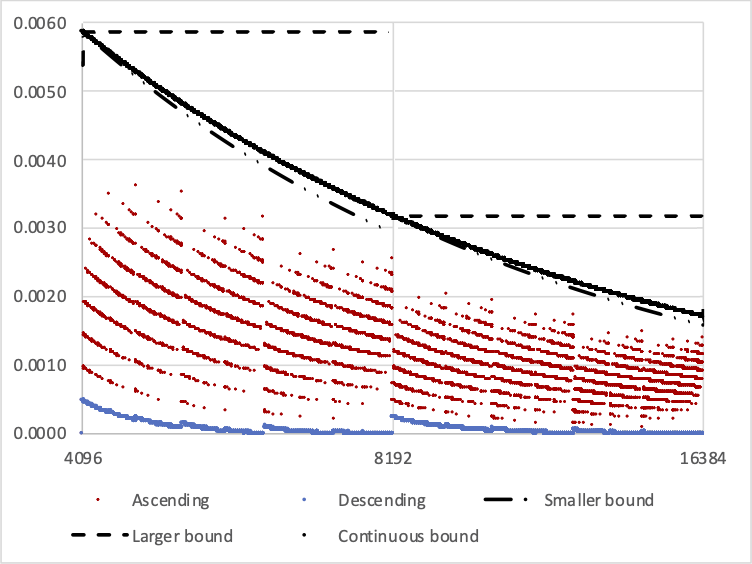}
        \caption{Here, a closer look at the graph at $\floor{\log_2(n)}=$ $12$ and $13$. For these values of $n$, all of the minD trees have normalized Colless index at less than 0.4\% of the normalized max. The bound is decreasing, so the MinD trees continue to improve in balance for $n>2^{13}$. \newline}
        \label{fig:colless_minmax_values-d}
    \end{subfigure}
    
    \mycaption{MinD trees have very good balance in terms of the Colless index} {Here  we compare the upper bounds (in red) and the lower bounds (in blue) on Colless indices of MinD trees with the maximum and minimum possible Colless indices on general trees. In subfigures (\subref{fig:colless_minmax_values-c}) and (\subref{fig:colless_minmax_values-d}), the indices are normalized by the minimum and maximum possible Colless indices. The MinD trees are very close to the best balanced divide-and-conquer minimum Colless index, so the class of MinD trees has near-best balance. The three bounds on the MinD Colless indices are shown in black. }
    \label{fig:colless_minD}
\end{figure}

\section{Conclusion}
Using the SD-tree structure, we have classified a number of commutative non-associative products. In particular, we have completely characterized trees having a minimal number of $D$-nodes, and shown that they are reasonably well-balanced. Table \ref{summary_table}  in the Appendix presents a summary of the main combinatorial results derived or referenced in this paper, with other interpretations and OEIS numbers where applicable.

\section{Acknowledgments}

We thank Vanessa Job for much discussion on this topic, and for her insight into the problem. We also thank Terry Grov\'e for asking the questions that inspired this work, also Andy DuBois, Shane Fogerty, Brett Neuman, Chris Mauney, and Bob Robey for many conversations on applications of the subject. We also thank the editors of the OEIS for their helpful remarks.

\newpage

\section{Appendix A: A family of sequences}\label{appendixA}

\begin{table}[h!tb]
	\small 
	\centering
	\mycaption{A family of sequences related to commutative non-associative products} {References, alternative interpretations, and OEIS numbers are provided where applicable.}
	 \label{summary_table}

	 \newcolumntype{C}{c}
	 \newcolumntype{T}{>{\raggedright\hspace{0pt}}p{0.45\linewidth}|>{\raggedright\hspace{0pt}}p{0.13\linewidth}|>{\raggedright\hspace{0pt}}p{0.22\linewidth}|p{0.07\linewidth}}
	 
    \begin{tabular} {T} 
	\textbf{Commutative non-associative product (CNAP) interpretation} & \textbf{References} & \textbf{Other interpretations} & \textbf{OEIS} \\
	\hline
	\end{tabular}

    \begin{tabular} {C}
    	\\[-.75em]
    	\textit{\textbf{Inequivalent parenthetic forms  }} \\
    	\\[-.75em]
	\end{tabular}

    \begin{tabular} {T} 
		\hline
    	Number of non-isomorphic parenthetic forms with $n$ leaf nodes & Prop. \ref{halfcat} & Half-Catalan numbers & \seqnum{A000992}\\
    	\hline
    	Number of non-isomorphic parenthetic forms with $n$ leaf nodes and 1 $S$-node (ladder product) & Cor. \ref{laddercor} & The constant 1 &  \seqnum{A000012}\\
		\hline
    	Number of non-isomorphic parenthetic forms with $n$ leaf nodes and 2 $S$-nodes& Prop. \ref{prop2Snodes} & Quarter-squares & \seqnum{A002620} \\
		\hline
    	Number of non-isomorphic parenthetic forms with $n$ leaf nodes and s $S$-nodes & Th. \ref{propkSnodes} &  &\seqnum{A335833}   \\
		\hline
    \end{tabular}
	
    \begin{tabular} {C}
    	\\[-.75em]
    	\textit{\textbf{Inequivalent commutative non-associative products}} \\
    	\\[-.75em]
	\end{tabular}

	\begin{tabular} {T} 
		\hline
		Number of commutative non-associative products on $n$ variables & Prop. \ref{ineq_sum} & Double factorial of odd numbers & \seqnum{A001147} \\
		\hline
		Number of commutative non-associative ladder products on $n$ variables& Prop. \ref{numladder} & No. of even perms. on $n$ elements & \seqnum{A001710} \\
		\hline
		Number of commutative non-associative pairwise operations on $n$ variables & Props. \ref{pairclosednew},  \ref{pairrecursiveold}  & No. of tournaments on $n$ teams &  \seqnum{A096351} \\
		\hline
	\end{tabular}

    \begin{tabular} {C}
	    \\[-.75em]
    	\textit{\textbf{$S$-nodes and $D$-nodes on trees of special form}} \\
	    \\[-.75em]
\end{tabular}

\begin{tabular} {T} 
	\hline
    Number of $S$-nodes in a divide-and-conquer tree with $n$ leaf nodes & Ths. \ref{exp_recursive}, \ref{exp_closed} & Cumulative deficient binary digit sum & \seqnum{A268289} \\
	\hline
    Number of $D$-nodes in a divide-and-conquer tree with $n$ leaf nodes & Th. \ref{exp_D_recursive} & Log of number of shapes of $n$-node divide and conquer trees & \seqnum{A296062} \\
	\hline
    Number of $D$-nodes in a complete full binary tree with $n$ leaf nodes & Prop. \ref{cfb_oeis} & Number of bits the same in $n$ and $n+1$ & \seqnum{A119387} \\
	\hline
\end{tabular}
\begin{tabular} {C}
	    \\[-.75em]
    	\textit{\textbf{Bounds}} \\
	    \\[-.75em]
\end{tabular}
\begin{tabular} {T} 
	\hline
	Minimum number of $S$-nodes in a full binary tree with $n$ leaf nodes & Prop. \ref{upper_bound} &The constant 1 &  \seqnum{A000012}\\
	\hline
	Maximum number of $S$-nodes in a full binary tree with $n$ leaf nodes & Cor. \ref{max_s_nodes} & Largest number $k$ where $2^k$ divides $n!$ & \seqnum{A011371} \\
	\hline
    Minimum number of $D$-nodes in a full binary tree with $n$ leaf nodes & Prop. \ref{min_dnodes} & $\omega(n)-1$, where $\omega(n)$ is the weight of $n$ & \seqnum{A048881} \\
	\hline
	Maximum order of the automorphism group of a CNAP tree with $n$ leaf nodes & Cor. \ref{max_s_nodes} & Largest $2^k$ that divides $n!$ & \seqnum{A060818} \\
	\hline	Lower bound for the number of CNAPs on $n$ variables all having the same SD-tree & Prop. \ref{lower_bound} & Largest odd divisor of $n!$ & \seqnum{A049606} \\
	\hline
	Upper bound for the number of CNAPs on $n$ variables all having the same SD-tree & Prop. \ref{upper_bound} & Number of even perms. on $n$ elements &  \seqnum{A001710} \\
	\hline
\end{tabular}
\end{table}
\section{Appendix B: Trees counted in terms of S- and D-nodes}\label{appendixB}

\subsection{Table counting trees in terms of S-nodes}
Table \ref{paren_nk} counts trees for small $n$, in terms of the number of $S$-nodes they have. This is calculated using Theorem \ref{propkSnodes}.
 There are several formulas for the number of parenthetic forms that are visible in this table. 
 	\begin{table}[h!tb]
	\footnotesize
	\centering
	\mycaption{The number of leaf-unlabeled SD-trees with $n$ leaf nodes and $s$ $S$-nodes, calculated up to $n=16$} {The  rows  represent $n$, the number of leaf nodes, and the columns represent $s$, the number of $S$-nodes.  $\alpha(n)$ is the sum of all the $\theta(n,s)$ in the $n^\text{th}$ row. The bold-faced number in each row indicates the number of $S$-nodes in the divide-and-conquer tree for $n$, and the italicized  indicates the number of $S$-nodes in the complete full binary tree for $n$.}
		\begin{tabular}{r|rrrrrrrrrrrrrrr|r}
			\backslashbox{\textbf{$n$}}{\textbf{$s$}} & \textbf{1} & \textbf{2}  & \textbf{3}   & \textbf{4}    & \textbf{5}    & \textbf{6}    & \textbf{7}    & \textbf{8}    & \textbf{9}   & \textbf{10}  &\textbf{11} &\textbf{12} &\textbf{13} &\textbf{14}   &\textbf{15}    & $\boldsymbol{\alpha(n)}$\\
			\hline
			\textbf{2}   & \textbf{\textit{1}} &   &     &      &      &      &      &      &     &     &   &   &   &   &  &  1   \\
			\textbf{3}   & \textbf{\textit{1}} & 0  &    &      &      &      &      &      &     &     &   &   &   &   &  &  1  \\
			\textbf{4}   & 1 & 0  & \textbf{\textit{1}}   &     &      &      &      &      &     &     &   &   &   &   &  &  2   \\
			\textbf{5}   & 1 & \textbf{\textit{1}}  & 1   & 0    &      &      &      &      &     &     &   &   &   &   &  &  3  \\
			\textbf{6}   & 1 & 2  & \textbf{2}   & \textit{1}    & 0    &     &      &      &     &     &   &   &   &   &  &   6  \\
			\textbf{7}   & 1 & 4  & 3   & \textbf{\textit{3}}    & 0    & 0    &       &      &     &     &   &   &   &   & & 11    \\
			\textbf{8}   & 1 & 6  & 7   & 6    & 3    & 0    & \textbf{\textit{1}}    &      &     &     &   &   &   &   &  &   24  \\
			\textbf{9}   & 1 & 9  & 14  & 13   & \textbf{\textit{8}}    & 1    & 1    & 0    &     &     &   &   &   &   &  &   47  \\
			\textbf{10}  & 1 & 12 & 27  & 28   & \textbf{23}   & 8    & \textit{3}    & 1    & 0   &      &   &   &   &   & &   103  \\
			\textbf{11}  & 1 & 16 & 49  & 58   & \textbf{54}   & 25   & \textit{8}    & 3    & 0   & 0   &    &   &   &   &  & 214    \\
			\textbf{12}  & 1 & 20 & 82  & 119  & 125  & 82   & \textbf{34}   & 15   & 2   & \textit{1}   & 0 &   &   &   & &  481   \\
			\textbf{13}  & 1 & 25 & 132 & 237  & 270  & 213  & \textbf{99}   & 42   & \textit{8}   & 3   & 0 & 0 &   &   & & 1030  \\
			\textbf{14}  & 1 & 30 & 199 & 449  & 578  & 542  & 322  & 151  & \textbf{51}  & 11  & \textit{3} & 0 & 0 &   & &   2337   \\
			\textbf{15}  & 1 & 36 & 294 & 821  & 1190 & 1255 & 867  & 440  & 173 & 39  & \textbf{\textit{15}} & 0 & 0 & 0 &  & 5131  \\ 
			\textbf{16}  & 1 & 42 & 414 & 1419  & 2394 & 2841 & 2338  & 1388  & 656 & 215  & 79 & 18 & 7 & 0 &  \textbf{\textit{1}} & 11813  \\ 
		\end{tabular}
		\label{paren_nk}
	\end{table}

\begin{itemize}
    \item The first column of Table \ref{paren_nk} counts the (unique) ladder parenthetic forms. These are discussed in Section \ref{ladder_pf}.
    \item The second column counts parenthetic forms with exactly 2 $S$-nodes. This column is the sequence of the quarter-squares, OEIS \seqnum{A002620} (with an offset) \cite{OEIS}. 
    This is discussed in Section \ref{2_snodes}.
    \item The rightmost non-zero entry in each row $n$ is the number of parenthetic forms on $n$ leaves having a maximal number of $S$-nodes for that $n$. This maximal number of $S$-nodes is ($n-\omega(n)$), where $\omega(n)$ is the weight of $n$, the number of bits that are $1$ in the binary representation of $n$. There are $(2\cdot \omega(n)-3)!!$ parenthetic forms having the maximal number of $S$-nodes. $\{1,1,1,1,1,3,1,1,1,3,1,3,3,15,1, \dots\}$ are the leading terms of this sequence. This is discussed in Section \ref{max_S_section}.
    \item When $n=2^k$, the maximal number of $S$-nodes is $s=2^k-1=n-1$. The rightmost $(n-1)^{\text{st}}$ column entry is $1$ in these rows (the rightmost entry is $0$ for every other row). The tree having this maximal number of $S$-nodes is the unique perfect tree on $2^k$ leaves. This is discussed in Proposition \ref{pfb} in Section \ref{max_S_section}.
    \item The number of $S$-nodes in the divide-and-conquer tree with $n$ leaves is the sequence OEIS \seqnum{A268289} 
    \cite{OEIS},
    the cumulative deficient binary digit sum. This does not form an easily distinguished set of entries in the table; to help distinguish these, we have put in bold-face the entry in each row whose column is the number of $S$-nodes of the divide-and-conquer tree. $\{1, 1, 3, 2, 3, 4, 7, 5, 5, 5, 7, 7, 9, 11, 15,  \dots\}$ are the leading terms of this sequence. There is an explicit formula for the number of $S$-nodes in a divide-and-conquer tree, and this number always falls into the range $[\floor{\frac{n}{2}},  n-1]$, and may be either of the extrema of this range. This is discussed in Section \ref{pairwise_cnaos}.
    \item The number of $S$-nodes in the complete full binary tree with $n$ leaves likewise does not form an easily distinguished set of entries in the table; to help distinguish these, we have italicized the entry in each row whose column is the number of $S$-nodes of the complete full binary tree. $\{1,1,1,1,1,3,1,8,3,8,1,8,3,15,1, \dots\}$ are the leading terms of this sequence. This tree is similar to the divide-and-conquer tree; however, the number of $S$-nodes is not always the same for the two trees. There is an explicit formula for the number of $S$-nodes in a complete full binary tree. This is discussed in Section \ref{cfb}.
    \item If SD-trees are defined such that the left child of a node must have at least as many leaf descendants as its right sibling, then the number of $S$-nodes can be used to count the automorphisms of such a tree. All trees having exactly $s$ $S$-nodes have $2^s$ SD-tree automorphisms, generated by the set of transpositions of children of $S$-nodes.  
    \item The row-sum column $\alpha(n)$ is the half-Catalan sequence OEIS \seqnum{A000992}
    \cite{OEIS}
    seen in Proposition \ref{halfcat} and Corollary \ref{rowsum}. 
    \item The entire table is sequence OEIS \seqnum{A335833}
    \cite{OEIS}
    . This is a new sequence in the OEIS.
\end{itemize}
Fig. \ref{fig:sdtrees6} shows a comparative example of some non-isomorphic forms having the same number of leaves.

\subsection{Table counting trees in terms of D-nodes}
Table \ref{table_dnodes} shows parenthetic forms in terms of their $D$-nodes. The data in the two tables are the same, since the number of $D$-nodes is $n-1-s$, where $s$ is the number of $S$-nodes, by Lemma~\ref{dnodes}. However the data are arranged differently in each table and exhibit different patterns. Table \ref{table_dnodes} shows the diagonals of Table \ref{paren_nk} as columns.
	\begin{table}[h!tb]
	\footnotesize
	\centering
	\mycaption{The number of leaf-unlabeled SD-trees with $n$ leaf nodes and $d$ $D$-nodes, calculated up to $n=16$} {The  rows  represent $n$, the number of leaf nodes, and the columns represent $d$, the number of $D$-nodes.  $\alpha(n)$ is the sum of all the forms in the $n^\text{th}$ row. The bold-faced number in each row indicates the number of $D$-nodes in the divide-and-conquer tree for $n$, and the italicized  indicates the number of $D$-nodes in the complete full binary tree for $n$.}
		\begin{tabular}{r|rrrrrrrrrrrrrrr|r}
			\backslashbox{\textbf{$n$}}{\textbf{$d$}} & \textbf{0} & \textbf{1}  & \textbf{2}   & \textbf{3}    & \textbf{4}    & \textbf{5}    & \textbf{6}    & \textbf{7}    & \textbf{8}   & \textbf{9}  &\textbf{10} &\textbf{11} &\textbf{12} &\textbf{13}   &\textbf{14}    & $\boldsymbol{\alpha(n)}$\\
			\hline
			\textbf{2}   & \textbf{\textit{1}} &   &     &      &      &      &      &      &     &     &   &   &   &   &  &  1   \\
			\textbf{3}   & 0 & \textbf{\textit{1}}  &    &      &      &      &      &      &     &     &   &   &   &   &  &  1  \\
			\textbf{4}   & \textbf{\textit{1}} & 0  & 1   &     &      &      &      &      &     &     &   &   &   &   &  &  2   \\
			\textbf{5}   & 0 & 1  & \textbf{\textit{1}}   & 1    &      &      &      &      &     &     &   &   &   &   &  &  3  \\
			\textbf{6}   & 0 & \textit{1}  & \textbf{2}   & 2    & 1    &     &      &      &     &     &   &   &   &   &  &   6  \\
			\textbf{7}   & 0 & 0  & \textbf{\textit{3}}   & 3    & 4    & 1    &       &      &     &     &   &   &   &   & & 11    \\
			\textbf{8}   & \textbf{\textit{1}} & 0  & 3   & 6    & 7    & 6    & 1    &      &     &     &   &   &   &   &  &   24  \\
			\textbf{9}   & 0 & 1  & 1  & \textbf{\textit{8}}   & 13    & 14    & 9   & 1    &    &     &   &   &   &   &  &   47  \\
			\textbf{10}  & 0 & 1 & \textit{3}  & 8   & \textbf{23}   & 28    & 27    & 12    & 1   &     &   &   &   &   & &   103  \\
			\textbf{11}  & 0 & 0 & 3  & \textit{8}   & 25   & \textbf{54}   & 58    & 49    & 16   & 1   &   &   &   &   &  & 214    \\
			\textbf{12}  & 0 & \textit{1} & 2  & 15  & \textbf{34}  & 82   & 125   & 119   & 82   & 20   & 1 &  &   &   & &  481   \\
			\textbf{13}  & 0 & 0 & 3 & \textit{8}  & 42  & \textbf{99}  & 213   & 270   & 237   & 132   & 25 & 1 &  &   & & 1030  \\
			\textbf{14}  & 0 & 0 & \textit{3} & 11  & \textbf{51}  & 151  & 322  & 542  & 578  & 449  & 199 & 30 & 1 &  & &   2337   \\
			\textbf{15}  & 0 & 0 & 0 & \textbf{\textit{15}}  & 39 & 173 & 440  & 867  & 1255 & 1190  & 821 & 294 & 36 & 1 & & 5131  \\ 
			\textbf{16}  & \textbf{\textit{1}} & 0 & 7 & 18  & 79 & 215 & 656  & 1388  & 2338 & 2841  & 2394 & 1419 & 414 & 42 &  1 & 11813  \\ 
		\end{tabular}
		\label{table_dnodes}
	\end{table}

\begin{itemize}
    \item The first column of Table \ref{table_dnodes} shows the forms having no $D$-nodes. These are the perfect trees, and occur only when $n=2^k$.
    \item The second column shows the forms having exactly one $D$-node. There is at most one such form for any $n$ and these forms only occur when $\omega(n)=2$, where $\omega(n)$ is the number of $1$s in the binary decomposition of $n$. The root of such a form is a $D$-node, and its two children are the perfect trees on $2^i$ and $2^j$ leaves, where $n=2^i+2^j$.
    \item Any tree with $n$ leaves must have at least $\omega(n)$ $D$-nodes, and this is sequence \seqnum{A000120}.  This corresponds to the observation in Table \ref{paren_nk} regarding trees with the maximal number of $S$-nodes, and is discussed in Section \ref{max_S_section}.
    \item The leftmost non-zero entry in each row $n$ is the number of parenthetic forms on $n$ leaves having a minimal number of $D$-nodes for that $n$. This minimal number of $D$-nodes is ($\omega(n)-1$), where $\omega(n)$ is the weight of $n$, the number of bits that are $1$ in the binary representation of $n$. There are $(2\cdot \omega(n)-3)!!$ parenthetic forms having the minimal number of $D$-nodes. $\{1,1,1,1,1,3,1,1,1,3,1,3,3,15,1, \dots\}$ are the leading terms of this sequence. This is discussed in Section \ref{max_S_section}.
    \item When $n=2^k$, the minimal number of $D$-nodes is $0$. The $0^{\text{th}}$ column entry is thus $1$ in these rows and $0$ is every other row. The tree having this minimal number of $D$-nodes is the unique perfect tree on $2^k$ leaves. This is discussed in Proposition \ref{pfb} in Section \ref{max_S_section}.

    \item The number of $D$-nodes in the divide-and-conquer tree with $n$ leaves is the sequence OEIS \seqnum{A296062} 
    \cite{OEIS}.
    This does not form an easily distinguished set of entries in the table; to help distinguish these, we have put in bold-face the entry in each row whose column is the number of $D$-nodes of the divide-and-conquer tree. The leading terms of this sequence are $\{0, 1, 0, 2, 2, 2, 0, 3, 4, 5, 4, 5, 4, 3, 0 \dots\}$. This sequence is  closely related to the Takagi function. There is an explicit formula for the number of $S$-nodes in a divide-and-conquer tree. This is discussed in Section \ref{pairwise_cnaos}.
    \item The number of $D$-nodes in the complete full binary tree with $n$ leaves likewise does not form an easily distinguished set of entries in the table; to help distinguish these, we have put italicized the entry in each row whose column is the number of $D$-nodes of the complete full binary tree. This sequence is OEIS \seqnum{A119387} 
    \cite{OEIS}, and $\{0, 1, 0, 2, 1, 2, 0, 3, 2, 3, 1, 3, 2, 3, 0, \dots\}$ are the leading terms of this sequence. The complete full binary tree is similar to the divide-and-conquer tree; however, the number of $D$-nodes is not always the same for the two trees. There is an explicit formula for the number of $D$-nodes in a complete full binary tree given in Theorem \ref{delta_cft}. These trees are discussed in Section \ref{cfb}.
    \item As in Table \ref{paren_nk}, the row-sum column $\alpha(n)$ is the half-Catalan sequence OEIS \seqnum{A000992}
    \cite{OEIS}.
\end{itemize}

\newpage

\bibliographystyle{alpha}
\bibliography{biblio}

\newcommand{\etalchar}[1]{$^{#1}$}
\begin{thebibliography}{NDTR12}

\bibitem[AK11]{allaart11}
Pieter Allaart and Kiko Kawamura.
\newblock The {T}akagi function: a survey.
\newblock {\em Real Analysis Exchange}, 37, 2011.

\bibitem[Bar19a]{baruchel19}
Thomas Baruchel.
\newblock Flattening {K}aratsuba’s recursion tree into a single summation.
\newblock {\em SN Computer Science}, 1(1), 2019.

\bibitem[Bar19b]{baruchel19_2}
Thomas Baruchel.
\newblock Properties of the cumulated deficient binary digit sum.
\newblock \url{https://arxiv.org/abs/1908.02250}, 2019.

\bibitem[BF85]{becker1985stationary}
P.~Becker and R.~Field.
\newblock Stationary concentration patterns in the {Oregonator} model of the
  {Belousov-Zhabotinskii} reaction.
\newblock {\em The Journal of Physical Chemistry}, 89(1):118--128, 1985.

\bibitem[BL76]{booth76}
Kellogg~S. Booth and George~S. Lueker.
\newblock Testing for the consecutive ones property, interval graphs, and graph
  planarity using {PQ}-tree algorithms.
\newblock {\em Journal of Computer and System Sciences}, 13(3):335 -- 379,
  1976.

\bibitem[BL06]{bohl06}
Erich Bohl and Peter Lancaster.
\newblock Implementation of a {M}arkov model for phylogenetic trees.
\newblock {\em Journal of Theoretical Biology}, 239(3):324 -- 333, 2006.

\bibitem[Cal09]{callan09}
David Callan.
\newblock A combinatorial survey of identities for the double factorial.
\newblock \url{https://arxiv.org/abs/0906.1317}, 2009.

\bibitem[CFH{\etalchar{+}}20]{coronado20}
Tomás~M. Coronado, Mareike Fischer, Lina Herbst, Francesc Rosselló, and
  Kristina Wicke.
\newblock On the minimum value of the {C}olless index and the bifurcating trees
  that achieve it.
\newblock {\em J. Math. Biol.}, 80:1993–2054, 2020.

\bibitem[Col80]{colless80}
Donald~H. Colless.
\newblock Congruence between morphometric and allozyme data for {M}enidia
  species: A reappraisal.
\newblock {\em Systematic Zoology}, 29(3):288 -- 299, 1980.

\bibitem[Col95]{colless95}
Donald~H. Colless.
\newblock Relative symmetry of cladograms and phenograms: An experimental
  study.
\newblock {\em Systematic Biology}, 44(1):102--108, 1995.

\bibitem[Dav88]{david88}
H.A. David.
\newblock {\em The method of paired comparisons}.
\newblock Griffin's statistical monographs and courses. C. Griffin, 1988.

\bibitem[DM93]{dale93}
M.R.T. Dale and J.W. Moon.
\newblock The permuted analogues of three {C}atalan sets.
\newblock {\em Journal of Statistical Planning and Inference}, 34(1):75 -- 87,
  1993.

\bibitem[Gol91]{goldberg91}
David Goldberg.
\newblock What every computer scientist should know about floating-point
  arithmetic.
\newblock {\em ACM Comput. Surv.}, 23(1):5–48, 1991.

\bibitem[Hea92]{heard92}
Stephen~B. Heard.
\newblock Patterns in tree balance among cladistic, phenetic, and randomly
  generated phylogenetic trees.
\newblock {\em Evolution}, 46(6):1818--1826, 1992.

\bibitem[Hig93]{higham93}
Nicholas~J. Higham.
\newblock The accuracy of floating point summation.
\newblock {\em SIAM Journal of Scientific Computing}, 14:783--799, 1993.

\bibitem[HJT17]{hwang17}
Hsien-Kuei Hwang, Svante Janson, and Tsung-Hsi Tsai.
\newblock Exact and asymptotic solutions of a divide-and-conquer recurrence
  dividing at half: Theory and applications.
\newblock {\em ACM Trans. Algorithms}, 13(4), 2017.

\bibitem[{ISO}18]{C18}
{ISO}{/}{IEC}.
\newblock Iso {I}nternational {S}tandard 9899:2018– {P}rogramming {L}anguage
  {C}.
\newblock \url{https://www.iso.org/standard/74528.html}, 2018.

\bibitem[JGF{\etalchar{+}}20]{job20}
Vanessa Job, Terry Grov\'e, Shane Fogerty, Brett Neuman, Chris Mauney, Laura
  Monroe, and Robert Robey.
\newblock Order matters: A case study on reducing floating point error in sums
  through ordering and grouping.
\newblock {\em Correctness 2020, International Conference for High-Performance
  Computing, Networking, Storage and Analysis}, 2020.

\bibitem[Jul13]{julia}
JuliaLang.org.
\newblock \url{https://github.com/JuliaLang/julia/pull/4039}, 2013.

\bibitem[Kah71]{kahan_1971_ASO}
W.~Kahan.
\newblock A survey of error analysis.
\newblock In {\em IFIP Congress}, 1971.

\bibitem[Kah73]{kahan_1973_implementation}
W.~Kahan.
\newblock Implementation of algorithms. part 1.
\newblock Technical report, California {U}niversity {B}erkeley {D}epartment of
  {C}omputer {S}ciences, 1973.

\bibitem[KF21]{kersting21}
Sophie~J. Kersting and Mareike Fischer.
\newblock Measuring tree balance using symmetry nodes -- a new balance index
  and its extremal properties.
\newblock \url{https://arxiv.org/abs/2105.00719v2}, 2021.

\bibitem[Knu97]{knuth97art1}
D.E. Knuth.
\newblock {\em The Art of Computer Programming: Volume 1: Fundamental
  Algorithms}.
\newblock Pearson Education, 1997.

\bibitem[Kr{\"u}07]{kruppel07}
M.~Kr{\"u}ppel.
\newblock On the extrema and the improper derivatives of {T}akagi’s
  continuous nowhere differentiable function.
\newblock {\em Rostocker Mathematisches Kolloquium}, 62:41--59, 2007.

\bibitem[KS93]{kirkpatrick93}
Mark Kirkpatrick and Montgomery Slatkin.
\newblock Searching for evolutionary patterns in the shape of a phylogenetic
  tree.
\newblock {\em Evolution}, 47(4):1171--1181, 1993.

\bibitem[Lag11]{lagarias11}
Jeffrey Lagarias.
\newblock The {T}akagi function and its properties.
\newblock {\em Functions in Number Theory and Their Probabilistic Aspects},
  2011.

\bibitem[MH97]{mooers97}
A.~Mooers and Stephen Heard.
\newblock Inferring evolutionary process from phylogenetic tree shape.
\newblock {\em Quarterly Review of Biology}, 72:31--54, 1997.

\bibitem[MRR18]{mir18}
Arnau~Torres Mir, Lucia Rotger, and Francesc Rosselló.
\newblock Sound {C}olless-like balance indices for multifurcating trees.
\newblock {\em PLOS ONE}, 13, 2018.

\bibitem[NDTR12]{Nicholaeff2012}
D.~Nicholaeff, N.~Davis, D.~Trujillo, and R.~W. Robey.
\newblock Cell-based adaptive mesh refinement implemented with general purpose
  graphics processing units.
\newblock Technical Report LA-UR-11-07127, Los Alamos National Laboratory,
  2012.

\bibitem[rI19]{IEEE754}
\relax IEEE.
\newblock {IEEE} {S}tandard for {F}loating-{P}oint {A}rithmetic.
\newblock {\em {IEEE} {S}td. 754-2019 ({R}evision of IEEE 754-2008)}, pages
  1--84, 2019.

\bibitem[rOFI20]{OEIS}
\relax OEIS Foundation~Inc.
\newblock The {O}n-{L}ine {E}ncyclopedia of {I}nteger {S}equences.
\newblock \url{http://oeis.org}, 2020.
\newblock References sequences A000992, A001147, A001710, A002620, A011371,
  A048881, A060818, A049606, A096351, A119387, A268289, A296062 and A335833.

\bibitem[Rog96]{rogers96}
James~S. Rogers.
\newblock {Central Moments and Probability Distributions of Three Measures of
  Phylogenetic Tree Imbalance}.
\newblock {\em Systematic Biology}, 45(1):99--110, 03 1996.

\bibitem[Ros19]{rosenberg19}
Noah~A. Rosenberg.
\newblock Enumeration of lonely pairs of gene trees and species trees by means
  of antipodal cherries.
\newblock {\em Advances in Applied Mathematics}, 102:1 -- 17, 2019.

\bibitem[Sci20]{numpysum}
SciPy.org.
\newblock
  \url{https://docs.scipy.org/doc/numpy/reference/generated/numpy.sum.html},
  2020.

\bibitem[SF97]{stanley97}
R.P. Stanley and S.~Fomin.
\newblock {\em Enumerative Combinatorics: Volume 2}.
\newblock Cambridge Studies in Advanced Mathematics. Cambridge University
  Press, 1997.

\bibitem[Tak01]{takagi01}
Teiji Takagi.
\newblock A simple example of the continuous function without derivative.
\newblock {\em Tokyo Sugaku-Butsurigakkwai Hokoku}, 1:176--177, 1901.

\bibitem[ZB19]{zou19}
Youming Zou and Paul~E. Black.
\newblock {D}ictionary of {A}lgorithms and {D}ata {S}tructures.
\newblock \url{https://xlinux.nist.gov/dads/HTML/perfectBinaryTree.html}, 2019.

\end{thebibliography}

\end{document}